\def\llncs{0}
\def\fullpage{1}
\def\anonymous{0}
\def\authnote{0}
\def\extabs{0}
\def\notxfont{0}
\def\submission{0}
\def\cameraready{0}

\ifnum\submission=1
\def\anonymous{1}
\def\llncs{1}
\def\authnote{0}
\else
\fi

\ifnum\cameraready=1
\def\submission{1}
\def\llncs{1}
\def\authnote{0}
\def\anonymous{0}
\else
\fi

\ifnum\anonymous=1
\def\llncs{1}
\def\authnote{0}
\else
\fi

\ifnum\llncs=1
	\documentclass[envcountsect,a4paper,runningheads]{llncs}
\else
	\documentclass[letterpaper,hmargin=1.05in,vmargin=1.05in,10pt]{article}
			\ifnum\fullpage=1
		\usepackage{fullpage}
		\fi
\fi

\usepackage[%
  colorlinks=true,
  citecolor=blue,
  pagebackref=true
]{hyperref}

\usepackage{amsmath, amsfonts, amssymb, mathtools,amscd}

\usepackage{amsthm}

\allowdisplaybreaks 
\usepackage{lmodern}
\usepackage[T1]{fontenc}
\usepackage[utf8]{inputenc}

\usepackage{etex}

\usepackage{arydshln} 
\usepackage{url}
\usepackage{ifthen}
\usepackage{bm}
\usepackage{multirow}
\usepackage[dvips]{graphicx}
\usepackage[usenames]{color}
\usepackage{xcolor,colortbl} 
\usepackage{threeparttable}
\usepackage{comment}
\usepackage{paralist,verbatim}
\usepackage{cases}
\usepackage{booktabs}
\usepackage{braket}
\usepackage{cancel} 
\usepackage{ascmac} 
\usepackage{framed}
\usepackage[hypcap=false]{caption}
\usepackage{authblk}

\usepackage{pifont}
\usepackage{qcircuit}
\usepackage{physics}
\definecolor{darkblue}{rgb}{0,0,0.6}
\definecolor{darkgreen}{rgb}{0,0.5,0}
\definecolor{forestgreen}{rgb}{0,0.6,0}
\definecolor{maroon}{rgb}{0.5,0.1,0.1}
\definecolor{dpurple}{rgb}{0.2,0,0.65}

\usepackage[capitalise,noabbrev]{cleveref}
\usepackage[absolute]{textpos}
\usepackage[final]{microtype}
\usepackage[absolute]{textpos}
\usepackage{everypage}
\DeclareMathAlphabet{\mathpzc}{OT1}{pzc}{m}{it}

\usepackage{algorithmic}
\usepackage{algorithm}
\usepackage{here}

\newtheoremstyle{thicktheorem}%
{\topsep}
{\topsep}
{\itshape}{}%
{\bfseries}%
{.}
{ }%
{\thmname{#1}\thmnumber{ #2}%
		\thmnote{ (#3)}%
}

\newtheoremstyle{remark}
{\topsep}
{\topsep}
	{}
	{}
	{}
	{.}
	{ }
	{\textit{\thmname{#1}}\thmnumber{ #2}
			\thmnote{ (#3)}%
	}

\ifnum\llncs=0
	\theoremstyle{plain}
	\newtheorem{theorem}{Theorem}[section]
	\newtheorem{lemma}[theorem]{Lemma}
	\newtheorem{corollary}[theorem]{Corollary}

        \theoremstyle{definition}
	\newtheorem{definition}[theorem]{Definition}
        \newtheorem{construction}[theorem]{Construction}

	\theoremstyle{remark}
	\newtheorem{claim}[theorem]{Claim}
	\newtheorem{remark}[theorem]{Remark}

\else
 \newtheorem{construction}[theorem]{Construction}
\fi

\Crefname{MyClaim}{Claim}{Claims}

	\crefname{theorem}{Theorem}{Theorems}
	\crefname{assumption}{Assumption}{Assumptions}
	\crefname{construction}{Construction}{Constructions}
	\crefname{corollary}{Corollary}{Corollaries}
	\crefname{conjecture}{Conjecture}{Conjectures}
	\crefname{definition}{Definition}{Definitions}
	\crefname{exmaple}{Example}{Examples}
	\crefname{experiment}{Experiment}{Experiments}
	\crefname{counterexample}{Counterexample}{Counterexamples}
	\crefname{lemma}{Lemma}{Lemmata}
	\crefname{observation}{Observation}{Observations}
	\crefname{proposition}{Proposition}{Propositions}
	\crefname{remark}{Remark}{Remarks}
	\crefname{claim}{Claim}{Claims}
	\crefname{fact}{Fact}{Facts}
	\crefname{note}{Note}{Notes}

\ifnum\llncs=1
 \crefname{appendix}{App.}{Appendices}
 \crefname{section}{Sec.}{Sections}
\else
\fi

\ifnum\llncs=1
\renewcommand*{\backref}[1]{}
\else
	\renewcommand*{\backref}[1]{(Cited on page~#1.)}
	\ifnum\notxfont=1
	\else
		\usepackage{newtxtext}
	\fi
\fi
\ifnum\authnote=0  
\newcommand{\jeff}[1]{}
\newcommand{\ryo}[1]{}
\newcommand{\takashi}[1]{}
\newcommand{\fuyuki}[1]{}

\else
\newcommand{\jeff}[1]{$\ll$\textsf{\color{violet} Jeff: { #1}}$\gg$}
\newcommand{\takashi}[1]{$\ll$\textsf{\color{orange} Takashi: { #1}}$\gg$}
\newcommand{\ryo}[1]{$\ll$\textsf{\color{darkgreen} Ryo: { #1}}$\gg$}
\newcommand{\fuyuki}[1]{$\ll$\textsf{\color{darkblue} Fuyuki: { #1}}$\gg$}

\fi

\newcommand{\hyb}{\mathsf{Hyb}}

\newcommand{\test}{\mathsf{Test}}

\newcommand{\UTE}{\mathsf{UTE}}
\newcommand{\UE}{\mathsf{UE}}

\newcommand{\OneUTE}{\mathsf{1UTE}}

\newcommand{\ketct}{\ket{\ct}}
\newcommand{\OWUTE}{\mathsf{OWUTE}}
\newcommand{\UTPKFE}{\mathsf{UTPKFE}}
\newcommand{\UTSKFE}{\mathsf{UTSKFE}}
\newcommand{\UTSS}{\mathsf{LRSS}}

\newcommand{\OSMAC}{\mathsf{OSM}}
\newcommand{\OSM}{\mathsf{OSM}}

\newcommand{\mvk}{\mathsf{mvk}}

\newcommand{\OW}{\mathsf{OW}}

\newcommand{\NCE}{\mathsf{NCE}}

\newcommand{\Fake}{\algo{Fake}}
\newcommand{\Reveal}{\algo{Reveal}}

\newcommand{\chosen}{\leftarrow}

\newcommand{\la}{\leftarrow}
\newcommand{\ra}{\rightarrow}

\newcommand{\concat}{\|}

\newcommand{\floor}[1]{\left\lfloor{#1}\right\rfloor}

\newcommand{\setbig}[1]{\left\{#1\right\}}
\newcommand{\setbk}[1]{\{#1\}}
\newcommand{\parens}[1]{\left(#1\right)}

\newcommand{\cA}{\mathcal{A}}
\newcommand{\Ahat}{\hat{\A}}

\newcommand{\Bhat}{\hat{\B}}

\newcommand{\cD}{\mathcal{D}}

\newcommand{\cF}{\mathcal{F}}

\newcommand{\cH}{\mathcal{H}}

\newcommand{\cK}{\mathcal{K}}

\newcommand{\cM}{\mathcal{M}}

\newcommand{\cO}{\mathcal{O}}

\newcommand{\cR}{\mathcal{R}}
\newcommand{\cS}{\mathcal{S}}
\newcommand{\cT}{\mathcal{T}}

\newcommand{\cX}{\mathcal{X}}
\newcommand{\cY}{\mathcal{Y}}
\newcommand{\cZ}{\mathcal{Z}}

\def\makeuppercase#1{
\expandafter\newcommand\csname tl#1\endcsname{\widetilde{#1}}
}

\def\makelowercase#1{
\expandafter\newcommand\csname tl#1\endcsname{\widetilde{#1}}
}

\newcommand{\N}{\mathbb{N}}

\newcommand{\eps}{\varepsilon}

\newcommand{\secp}{\lambda}
\newcommand{\lam}{\secp}

\newcommand{\extr}{\algo{ext}}

\newcommand{\A}{\entity{A}}
\newcommand{\B}{\entity{B}}
\newcommand{\C}{\entity{C}}

\newcommand{\AB}{(\A,\B)}

\newcommand{\insf}{\mathsf{in}}
\newcommand{\outsf}{\mathsf{out}}

\newcommand{\hybb}[1]{\mathsf{Hyb}_{#1}^{(b)}}

\newcommand{\Hyb}[2]{\mathsf{Hyb}_{#1}^{#2}}

\newcommand{\diff}[1]{{\color{forestgreen} #1}}

\newcommand*{\sk}{\keys{sk}}
\newcommand*{\pk}{\keys{pk}}
\newcommand*{\msk}{\keys{msk}}

\newcommand{\ct}{\keys{ct}}
\newcommand{\ctstar}{\ct^\ast}

\newcommand*{\pp}{\keys{pp}}
\newcommand*{\dk}{\keys{dk}}
\newcommand{\dkstar}{\dk^\ast}
\newcommand*{\vk}{\keys{vk}}
\newcommand*{\ek}{\keys{ek}}

\newcommand{\mstar}{m^\ast}

\newcommand*{\keys}[1]{\mathsf{#1}}

\newcommand*{\algo}[1]{\ensuremath{\mathsf{#1}}}

\newcommand*{\entity}[1]{\mathcal{#1}}

\newcommand{\compclass}[1]{\textbf{\textrm{#1}}}

\newenvironment{boxfig}[2]{\begin{figure}[#1]\fbox{\begin{minipage}{0.97\linewidth}
                        \vspace{0.2em}
                        \makebox[0.025\linewidth]{}
                        \begin{minipage}{0.95\linewidth}
            {{
                        #2 }}
                        \end{minipage}
                        \vspace{0.2em}
                        \end{minipage}}}{\end{figure}}

\newcommand{\mat}[1]{\boldsymbol{#1}}

\newcommand{\prf}{\algo{F}}

\newcommand{\win}{\mathtt{win}}

\newcommand{\Setup}{\algo{Setup}}

\newcommand{\Gen}{\algo{Gen}}
\newcommand{\KeyGen}{\algo{KeyGen}}

\newcommand{\KG}{\algo{KG}}
\newcommand{\Enc}{\algo{Enc}}

\newcommand{\Dec}{\algo{Dec}}

\newcommand{\Sign}{\algo{Sign}}
\newcommand{\Vrfy}{\algo{Vrfy}}

\newcommand{\st}{\algo{st}}

\newcommand{\TD}{\algo{TD}}

\newcommand\SKE{\algo{SKE}}

\newcommand{\share}{\algo{Share}}
\newcommand{\reconst}{\algo{Reconst}}
\newcommand{\ord}[1]{\ensuremath{{#1}^{\mathrm{th}}}}
\newcommand{\kin}{k_\insf}
\newcommand{\kout}{k_\outsf}
\newcommand{\val}{\mathsf{val}}
\newcommand{\phat}{\hat{P}}
\newcommand{\SSscheme}{\Pi_\mathsf{SS}}
\newcommand{\SSshare}{\mathsf{SS}.\share}
\newcommand{\SSreconst}{\mathsf{SS}.\reconst}
\newcommand{\gstar}{g^\ast}
\newcommand{\istar}{i^\ast}

\newcommand{\SDFE}{\mathsf{SDFE}}
\newcommand{\UTK}{\mathsf{FEUTK}}

\newcommand{\SKFE}{\mathsf{SKFE}}

\newcommand{\PKFE}{\mathsf{PKFE}}

\newcommand{\qstate}[1]{\mathpzc{#1}}
\newcommand{\tensor}{\otimes}
\newcommand*{\qreg}[1]{{\color{gray}{\mathsf{#1}}}}

\newcommand{\negl}{{\mathsf{negl}}}

\newcommand{\eqand}{\quad\text{and}\quad}
\newcommand{\poly}{{\mathrm{poly}}}
\newcommand{\polylog}{{\mathrm{polylog}}}

\newcommand{\zo}[1]{\{0,1\}^{#1}}

\newcommand{\getsr}{%
  \mathrel{\vbox{\offinterlineskip\ialign{%
    \hfil##\hfil\cr
    $\scriptstyle\normalfont\textsc{r}$\cr
    \noalign{\kern0.1ex}
    $\leftarrow$\cr
}}}}
\DeclareMathOperator*{\Exp}{\mathbb{E}}

\newcommand{\xor}{\oplus}

\newcommand{\Hinf}{\mathbf{H}_\infty}

\newcommand{\Ppoly}{\compclass{P}/\compclass{poly}}

\makeatletter
\DeclareRobustCommand
  \myvdots{\vbox{\baselineskip4\p@ \lineskiplimit\z@
    \hbox{.}\hbox{.}\hbox{.}}}
\makeatother

\newcommand{\ifelsesubmit}[2]{
\ifnum\submission=1
#1
\else
#2
\fi
}

\newcommand{\ifsubmit}[1]{\ifelsesubmit{#1}{}}
\newcommand{\ifnotsubmit}[1]{\ifelsesubmit{}{#1}}
\newcommand{\ifelsesubmiteq}[1]{\ifelsesubmit{\[ #1 \]}{$#1$}}

\begin{document}

\ifelsesubmit{
\title{Untelegraphable Encryption and its Applications}
\institute{}
\author{}
\date{}
}{
\title{Untelegraphable Encryption and its Applications 
}
\author[$\star$]{Jeffrey Champion\thanks{Part of this work was done while visiting NTT Social Informatics Laboratories for an internship.}}
\author[$\dagger$]{Fuyuki Kitagawa}
\author[$\dagger$]{Ryo Nishimaki}
\author[$\dagger$]{Takashi Yamakawa}
\affil[$\star$]{UT Austin}
\affil[ ]{\href{mailto:jchampion@utexas.edu}{\texttt{jchampion@utexas.edu}}}
\affil[$\dagger$]{NTT Social Informatics Laboratories}
\affil[ ]{
\{\href{mailto:fuyuki.kitagawa@ntt.com}{\texttt{fuyuki.kitagawa}},
\href{mailto:ryo.nishimaki@ntt.com}{\texttt{ryo.nishimaki}},
\href{mailto:takashi.yamakawa@ntt.com}{\texttt{takashi.yamakawa}}\}\texttt{@ntt.com}}
\date{}
}

\maketitle

\ifnum\extabs=1

\input{extabs}

\bibliographystyle{alpha}
\bibliography{reference,abbrev3,crypto}

\else
\begin{abstract}
    We initiate the study of untelegraphable encryption (UTE), founded on the no-telegraphing principle, which allows an encryptor to encrypt a message such that a binary string representation of the ciphertext cannot be decrypted by a user with the secret key, a task that is classically impossible.
    This is a natural relaxation of unclonable encryption (UE), inspired by the recent work of Nehoran and Zhandry (ITCS 2024), who showed a computational separation between the no-cloning and no-telegraphing principles. 

    \ifelsesubmit{\quad }{\indent }In this work, we define and construct UTE information-theoretically in the plain model.
    Building off this, we give several applications of UTE and study the interplay of UTE with UE and well-studied tasks in quantum state learning, yielding the following contributions:
    \begin{itemize}
        \item A construction of collusion-resistant UTE from standard secret-key encryption (SKE). We additionally show that hyper-efficient shadow tomography (HEST) is impossible assuming collusion-resistant UTE exists. By considering a relaxation of collusion-resistant UTE, we are able to show the impossibility of HEST assuming only pseudorandom state generators (which may not imply one-way functions). This almost unconditionally answers an open inquiry of Aaronson (STOC 2018).
        
        \item A construction of UTE from a quasi-polynomially secure one-shot message authentication code (OSMAC) in the classical oracle model, such that there is an explicit attack that breaks UE security for an \textit{unbounded} polynomial number of decryptors. 
        
        \item A construction of \textit{everlasting secure} collusion-resistant UTE, where the decryptor adversary can run in unbounded time, in the quantum random oracle model (QROM), and formal evidence that a construction in the plain model is a challenging task. 
        We additionally show that HEST with unbounded post-processing time (which we call \textit{weakly-efficient} shadow tomography) is impossible assuming everlasting secure collusion-resistant UTE exists. 
        
        \item A construction of secret sharing for all polynomial-size policies that is resilient to \textit{joint} and \textit{unbounded} classical leakage from collusion-resistant UTE and classical secret sharing for all policies.

        \item A construction (and definition) of collusion-resistant untelegraphable secret-key functional encryption (UTSKFE) from single-decryptor functional encryption and plain secret-key functional encryption, and a construction of collusion-resistant untelegraphable public-key functional encryption from UTSKFE, plain SKE, and plain public-key functional encryption.
    \end{itemize}
\end{abstract}

\section{Introduction}
Unclonable cryptography uses the no-cloning principle of quantum mechanics~\cite{Par70,WZ82,Die82} to enable cryptographic properties that are impossible to achieve classically. 
Constructing primitives with such properties has been an active area of research for the past four decades, starting with the well-studied notion of quantum money~\cite{Wiesner83}. 

\paragraph{Unclonable encryption.} In recent years, there has been particular interest in unclonable encryption (UE), which was introduced by Broadbent and Lord~\cite{TQC:BL20}.
Unclonable encryption is a one-time secure encryption scheme with the following guarantee: any adversary $\A$ that is given a ciphertext of some message $m$ in the form of a quantum state cannot output two (possibly-entangled) states to adversaries $\B$ and $\C$ such that $\B$ and $\C$ can recover information about $m$ even given the classical secret key $\sk$. 
When $m$ is uniformly chosen from the message space, and $\B$ and $\C$ must recover the entire message, we call this \textit{one-way security}. 
The more natural notion of security, called \textit{indistinguishability security}, is for $\A$ to chose two messages $m_0,m_1$, receive an encryption of $m_b$, and have $\B$ and $\C$ guess $b$. 

One-way security was achieved information theoretically in~\cite{TQC:BL20}, but indistinguishability security has yet to be achieved information-theoretically or from any standard assumption. 
In particular, positive results on unclonable indistinguishability security have come at either the cost of oracles~\cite{TQC:BL20,C:AKLLZ22}, relaxed definitions~\cite{KN23,AKY24,STOC:ColGun24}, or new conjectures~\cite{AB24}. 
Moreover, there are negative results that bring into question the plausibility of proving unclonable indistinguishability security information theoretically or from a standard assumption~\cite{MST21,C:AKLLZ22}.

\paragraph{This work: untelegraphable encryption.} Partly motivated by the challenges of constructing UE, the main focus of this work is a natural relaxation of unclonable encryption, which we call untelegraphable encryption (UTE), a name that comes from its relation to the no-telegraphing principle~\cite{Wer98}. Like unclonable encryption, the basic primitive we consider is a one-time secure encryption scheme. 
The security guarantee is as follows: any adversary $\A$ that is given a ciphertext of some message $m$ in the form of a quantum state cannot produce a classical string $\st$ that would allow an adversary $\B$ to recover information about $m$ given $\st$ and the classical secret key $\sk$. The one-way and indistinguishability security notions are analogous to the ones in UE. 

The core challenges in constructing unclonable encryption from standard assumptions are: (1) there are two (possibly entangled) second stage adversaries, which make most computational reduction and search-to-decision techniques fail (2) the second stage adversaries get the secret key of the encryption scheme, which goes against the foundations of classical cryptography. Untelegraphable encryption by definition removes the first challenge, so there is reason to believe that it is considerably more feasible than unclonable encryption. Furthermore, the study of untelegraphable encryption builds on the rich literature examining the relationship between classical and quantum information.

\paragraph{Quantum and classical information.}
The ``quantum vs classical information'' paradigm broadly spans the field of quantum computing.
In quantum complexity theory, the fundamental question of whether quantum states are stronger than classical strings as proofs to an efficient verifier (also known as the QMA vs QCMA problem~\cite{AN02}) is a very well-studied problem, with numerous works~\cite{AK07,FK18,NN23,NZ24,LLPY24} attempting to get closer to an answer via oracle separations. A similar question has also been studied in terms of efficient algorithms interacting with classical and quantum advice~\cite{AK07,LLPY24}. 
In quantum information, two of the fundamental no-go theorems are the no-cloning and no-telegraphing principles. Previously believed to be equivalent, the work of Nehoran and Zhandry~\cite{NZ24} shows that when it comes to computationally bounded algorithms, there is indeed a difference between telegraphing and cloning. We believe that we can enrich our understanding of unclonability by studying untelegraphability, even if the focus is on their differences. 

\paragraph{Learning quantum states.} Furthermore, the fundamental task of describing an unknown quantum state, called quantum state tomography~\cite{VR89}, aims to classically learn an unknown quantum state given many copies of the state. Due to the destructive nature of measurement, full quantum state tomography provably requires exponentially many copies of the state. This barrier lead Aaronson~\cite{Shadow2} to introduce the notion of \textit{shadow tomography}. 
Shadow tomography focuses on learning how a particular family of two-outcome measurements acts on an unknown state. Such a task can be done more efficiently, and may still extract the majority of ``useful'' information in the state. 

In the same work, Aaronson suggested the possibility of so-called \textit{hyper-efficient} shadow tomography (HEST), hoping to achieve a truly computationally feasible variant of shadow tomography.
In the subsequent work of Huang, Kueng, and Preskill~\cite{NatPhys:HKP20}, the notion of \textit{classical shadows} was introduced, which aims for efficiency in the deconstruction of the quantum state into a classical string, such that the classical output can be post-processed (inefficiently) to learn interesting properties of a state. 
As we will show, HEST and classical shadows are strongly connected to untelegraphable encryption. 

\paragraph{Guarantees and practicality.} Finally, we further motivate untelegraphable encryption by remarking on the importance of having provable guarantees from well-founded assumptions in cryptography. Such guarantees can give us confidence in constructions that would otherwise derive their security entirely through cryptanalysis. This is the motivation for the notion of non-adaptive security in cryptography, which in general involves an adversary committing to how they will attack a protocol's security. 
Even though non-adaptive security fails to model the real world, it allows us to gain confidence in constructions, and can later lead to provable adaptive security.

We believe that in a similar way, showing a provable guarantee like untelegraphable security for a scheme can be an avenue to gain confidence in the scheme's unclonable security, assuming no explicit cloning attack is known. Furthermore, classical interaction with quantum computers is already occurring, so preventing the classical exfiltration of important quantumly-stored information is a very useful guarantee with respect to near-term quantum computer capabilities and distribution. 

\subsection{Our Results}
In this section, we give a summary of our contributions:
\begin{enumerate}
    \item We formally define untelegraphable variants of several cryptographic primitives: untelegraphable encryption and variants on the basic security such as collusion-resistant and everlasting security (\cref{sec:crypto-prelims}), and untelegraphable functional encryption (UTFE) (\cref{sec:utfe-defs}). 
    
    \item We construct untelegraphable encryption (UTE) with indistinguishability security information theoretically (\cref{sec:ind-ute}), collusion-resistant UTE from standard secret-key encryption (SKE) (\cref{sec:cr-ute}), and show that hyper-efficient shadow tomography (HEST) for general mixed quantum states is impossible assuming collusion-resistant UTE (\cref{sec:hest}). By instantiating SKE with succinct query-bounded SKE from a pseudorandom state (PRS) generator (shown in \cref{sec:qske}), we are able to show HEST is impossible assuming just PRS generators.
    
    \item We give a separation between indistinguishability secure UTE and one-way secure UE from a quasi-polynomially secure one-shot message authentication code (OSMAC) in the classical oracle model (\cref{sec:clonable-ute}). The OSMAC can be instantiated with the learning with errors (LWE) assumption.  
    
    \item We construct everlasting secure collusion-resistant UTE in the quantum random oracle model (QROM) (\cref{sec:everlasting-ute}) and  show that weakly-efficient shadow tomography (WEST) for general mixed quantum states is impossible assuming everlasting secure collusion-resistant UTE (\cref{sec:west}). We also show that public-key everlasting one-way secure UTE is separated from all falsifiable assumptions assuming the sub-exponential hardness of one-way functions (\cref{sec:everlasting-impossible}).
    
    \item We construct unbounded joint-leakage-resilient secret sharing (UJLRSS) for polynomial-size policies from collusion-resistant UTE and classical secret sharing for polynomial-size policies (\cref{sec:utss}), both of which can be instantiated from standard SKE.

    \item We construct collusion-resistant untelegraphable secret-key functional encryption (UTSKFE) from single-decryptor functional encryption and plain secret-key functional encryption (SDFE) (\cref{sec:utskfe}). We also construct collusion-resistant untelegraphable public-key functional encryption from UTSKFE, plain SKE, and plain public-key functional encryption (\cref{sec:utpkfe}). The SDFE scheme can be instantiated from sub-exponentially secure indistinguishability obfuscation (iO) and one-way functions, along with LWE.
\end{enumerate}

\noindent We now discuss how some of our results compare with related works.

\paragraph{Impossibility of hyper-efficient shadow tomography.}
The initial impossibility results shown for HEST by Aaronson~\cite{Shadow2} and Kretschmer~\cite{Kre21} use oracles, so they are incomparable to our result.
In the work of \c{C}akan and Goyal~\cite{CG24}, impossibility of HEST is shown assuming sub-exponential post-quantum iO and LWE.
We improve the assumption needed to just PRS generators~\cite{C:JiLiuSon18}, which may be weaker than one-way functions~\cite{Kre21}. 
Moreover, the WEST impossibility result is unique to our work.
Interestingly, the impossibility of WEST somewhat complements a positive quantum state learning result in the work of Huang, Kueng, and Preskill~\cite{NatPhys:HKP20}, which essentially says that WEST is possible when the family of measurements are all projections to pure quantum states, while we show WEST is impossible in an instance where the measurements are not projections to pure states.

\paragraph{Separating cloning and telegraphing.}
The work of Nehoran and Zhandry~\cite{NZ24} showed the first computational separation of cloning and telegraphing. They show that there exists a quantum oracle and set of states such that the states can be efficiently cloned with respect to the oracle but not efficiently telegraphed. Moreover, their set of untelegraphable states is hard to deconstruct (make classical) even in unbounded time in a way that reconstruction (recovering the state) is efficient. Our separation between UE and UTE roughly shows the existence of a classical oracle and family of states where it is easy to clone many bits of information about the states, but it is hard to efficiently telegraph even a single bit. The results and assumptions needed for each result are incomparable.
In the opposite direction, \c{C}akan, Goyal, Liu-Zhang, and Ribeiro~\cite{CGLR24} show that unclonability likely does not imply ``interactive'' untelegraphability, where the classical string can come from multiple rounds of classical interaction.

\paragraph{Unbounded leakage-resilient secret sharing.}
The work of~\cite{CGLR24} introduces the notion of unbounded leakage-resilient secret sharing (ULRSS). In this variant of secret sharing,  security holds even with respect to unbounded \textit{local} classical leakage on each share, meaning the shares must be quantum. Our construction in this work extends security to the setting of \textit{joint} leakage, where the leakage function can depend on unqualified sets of shares that are disjoint from one another. The distinction between these definitions is strict, as the construction in~\cite{CGLR24} is insecure under our definition. However, for Boolean formula policies, the construction in~\cite{CGLR24} is statistically secure, while our construction relies on PRS generators.

The work of Ananth, Goyal, Liu, and Liu~\cite{AGLL24} defines and constructs the notion of unclonable secret sharing (USS). 
Intuitively, USS allows for collusion among many or even all parties, where collusion between two parties means they are entangled and cannot communicate otherwise. 
In our definition of UJLRSS, we consider collusion between groups of parties that cannot reconstruct the secret, but allow for full sharing of information between parties in a group. 
Furthermore, we construct a secret sharing scheme for all polynomial-size policies, whereas~\cite{AGLL24} only consider $n$-out-of-$n$ secret sharing. 
We compare the definitions more in~\cref{sec:utss-def}.

\paragraph{Untelegraphable keys.} 
\c{C}akan, Goyal, Liu-Zhang, and Ribeiro~\cite{CGLR24} also show how to construct quantum keys resilient to unbounded classical (interactively generated) leakage. This can be thought of as cryptography with untelegraphable keys in our terminology. Specifically, they show how to protect keys for public-key encryption, digital signatures, and pseudorandom functions. Additionally, the work of~\cite{NZ24} shows how to construct an encryption scheme with keys that can be cloned but not exfiltrated.

\paragraph{Independent work.} Independent of our work, \c{C}akan and Goyal~\cite{CG24pre} study LOCC leakage-resilient encryption, signatures, and NIZKs. In our terminology, they construct an extension of untelegraphable encryption where the first and second stage adversaries can interact when generating the classical state (before the second stage adversary gets the decryption key). They construct SKE and PKE with this security property, improving our results in~\cref{sec:ind-ute,sec:cr-ute}. However, apart from constructing this extension of UTE, the work of~\cite{CG24pre} studies entirely different facets of untelegraphability from our work. 

\subsection{Technical Overview}\label{sec:tech-overview}
In this section, we provide a high-level overview of our main constructions.

\paragraph{UTE with indistinguishability security.} 
We begin with our construction of one-time untelegraphable encryption with statistical indistinguishability security. 
The construction uses a one-way secure UTE scheme (which we have from one-way secure UE) and a universal hash family as its main components, and can be described as follows:
\begin{itemize}
    \item The secret key is a secret key $\sk_\OW$ for the one-way secure UTE, a random hash function $h$ from the hash family, and a uniformly random string $r$. 
    
    \item A ciphertext for a message $m$ is an encryption $\ket{\ct_\OW}$ of a random message $x$ for the one-way scheme, along with the string $r'=r\xor h(x)\xor m$, where $\xor$ is the bitwise XOR operation. 
\end{itemize}
To decrypt, simply decrypt the random message $x$ and compute $m=r'\xor r\xor h(x)$.
The key observation for proving security is that when the first stage adversary $\A$ outputs its two messages $(m_0,m_1)$, the string $r'$ in the ciphertext can be sampled uniformly and given to $\A$ along with the honestly generated $\ket{\ct_\OW}$. 
When $\A$ outputs classical string $\st$ for the second stage adversary $\B$, the challenger can then sample the hash function $h$ and set $r=r'\xor h(x)\xor m_b$.
Such a trick can be thought of like a one-time version of ``non-committing encryption'', a notion that will be discussed more later in this overview.

Thus, in the second stage, we can now treat the components $(\sk_\OW,\st,r')$ as ``leakage'' on the randomly chosen $x$. 
In particular, by the one-way security of the underlying UTE scheme, we can argue that $x$ has high entropy from the perspective of $\B$, and furthermore, the hash $h$ is independent of $x$ and its leakage. 
This allows us to treat $h(x)$ as a uniform string by a variant of the well-known leftover hash lemma. We refer to~\cref{sec:ind-ute} for the full details. 
It is unknown whether such a scheme satisfies UE security, but the main issue in proving security stems from the fact that one of the second stage adversaries $\B$ and $\C$ can obtain the honest ciphertext in its entirety. This makes it impossible to reason about the traditional notion of entropy from the perspectives of $\B$ and $\C$.

\paragraph{Collusion-resistant UTE.} 
Given UTE with indistinguishability security, we can now consider constructing collusion-resistant UTE. By collusion-resistant security, we mean that the first stage adversary $\A$ can now make numerous message queries of the form $(m_0,m_1)$ and receive encryptions of $m_b$ in return. The term collusion-resistant comes from the fact that in the UE setting, the goal of an adversary given $q$ encryptions would be to transmit the bit $b$ to $q+1$ second stage adversaries. 

Similarly to the one-time case, our construction of collusion-resistant UTE makes use of a hybrid encryption technique using non-committing encryption (NCE)~\cite{EC:JarLys00,TCC:CanHalKat05}, as has been used previously to construct advanced encryption schemes with certified deletion~\cite{AC:HMNY21}. 
Intuitively, NCE is an encryption scheme with an algorithm for making ``fake'' ciphertexts that contain no message information, and a corresponding algorithm that ``reveals'' the true message $m$ via outputting a decryption key $\dk$. The security guarantee says that a fake ciphertext and decryption key look indistinguishable from a real ciphertext and decryption key on some challenge message, given access to an encryption oracle. Due to the nature of its security guarantee, NCE requires encryption and decryption keys to be separately defined.
We make use of such a scheme along with a one-time UTE scheme to construct collusion-resistant UTE (with separate encryption and decryption keys) as follows:
\begin{itemize}
    \item The encryption key $\ek$ and decryption key $\dk$ are directly sampled from the NCE scheme.
    \item A ciphertext for a message $m$ consists of a one-time UTE encryption of the message $\ket{\ct_\OneUTE}$ along with a ciphertext $\ct_\NCE$ of the secret key $\sk_\OneUTE$ for the one-time UTE scheme.
\end{itemize}
To decrypt given $\dk$, simply recover $\sk_\OneUTE$ and decrypt the one-time UTE ciphertext. The intuition for security is that for each query $(m_0,m_1)$ made by the first stage adversary $\A$, we can switch the ciphertext $\ct_\NCE$ that encrypts $\sk_\OneUTE$ to a fake one, and have the decryption key $\dk$ be produced by the revealing algorithm for the message $\sk_\OneUTE$. We can then appeal to one-time UTE security to switch the message $m_b$ to 0, say. We refer to~\cref{sec:cr-ute} for the full details.
Note that such a construction does not yield collusion-resistant security in the UE case, since by the formulation of UE, there has not been a way to appeal to its security multiple times in a single proof. This prevents one from using a query-by-query security proof.  

\paragraph{Impossibility of hyper-efficient shadow tomography.} 
We now illustrate the connection between hyper-efficient shadow tomography (HEST) and collusion-resistant UTE. To start, we formalize the notion of hyper-efficient shadow tomography. 
Let $E$ be a quantum circuit that takes as input some index $i\in[M]$ and an $n$-qubit quantum state $\rho$, and outputs a single classical bit.
A shadow tomography procedure is an algorithm $\cT$ that takes $E$ and $k$ copies of $\rho$ and outputs (a classical description of) a  quantum circuit $C$ such that $C(i)$ estimates the probability that $E(i,\rho)=1$ up to some error $\eps$ with failure probability $\delta$. The algorithm $\cT$ is \textit{hyper-efficient} if its runtime and the required number of copies $k$ are both $\poly(n, \log M, 1/\eps,\log 1/\delta)$.

The key observation is to map the state $\rho$ to a UTE ciphertext, map $M$ to be the enumeration of UTE decryption keys $\dk$, and map the circuit $E$ to be the UTE decryption algorithm (assuming 1-bit messages). Since UTE algorithms must be efficient, the parameters $n$ and $\log M$ are polynomial in the security parameter. This implies that running the procedure $\cT(E,\ket{\ct}^{\otimes k})$, will efficiently generate a classical description of a circuit $C$, such that $C(\dk)$ estimates $\Dec(\dk,\ketct)$ up to error $\eps$. By setting $\eps<1/2$ and $\delta$ to be negligible, this immediately gives an attack on collusion-resistant UTE with \textit{overwhelming advantage}, so long as the UTE scheme is correct. Since the ciphertext and key sizes in the above scheme are fixed for any polynomial number of queries, we can rule out the existence of all HEST procedures assuming NCE exists. NCE can be constructed from CPA-secure SKE~\cite{C:KNTY19} and thus one-way functions (OWFs), so this immediately gives the impossibility of HEST from OWFs. We provide the full details in~\cref{sec:hest}.

\paragraph{Weakening the assumption to PRS generators.} 
A notable feature of the above impossibility result is that it hinges on the the ciphertext and key sizes of the underlying UTE scheme being fixed for any polynomial number of queries.
Additionally, the copy threshold $k$ is a fixed value, so it is fine for the UTE scheme to have its setup depend on $k$, so long as the aforementioned succinctness properties hold. Thus, the result would still hold if we were to construct ``bounded-query'' NCE with succinct ciphertexts and decryption keys, where the initial key generation can depend on the a query bound $q$ along with the security parameter $\lam$, analogous to notions like bounded-collusion functional encryption. Thus, in hopes of strengthening our separation even further, we aim to construct succinct bounded-query NCE from pseudorandom state (PRS) generators, which may not imply OWFs. 

We do this by constructing succinct bounded-query CPA secure SKE from a pseudorandom function-like state (PRFS) generator with a polynomial-sized domain, which is implied by a PRS generator. Recall that a PRFS generator is a QPT algorithm $G$ such that for a uniform random classical key $k$, tuples $(x,G(k,x))$ are indistinguishable from tuples $(x,\ket{\psi})$, where each $\ket{\psi}$ is an independent Haar-random state and inputs $x$ must be chosen non-adaptively. 
Reminiscent of bounded-collusion attribute-based encryption constructions (see~\cite{ISVWY17}), our construction is as follows:
\begin{itemize}
    \item The secret key consists of $\lam$ random keys for the PRFS generator $(k_1,\dots,k_\lam)$. The dependence on $q$ comes from the fact that our PRFS generator $G$ will have domain $\zo{d}$ where $d$ depends on $q$.
    \item To encrypt a message $m\in\zo{}$, sample shares $s_1,\cdots,s_\lam$ uniformly such that $s_1\xor\cdots\xor s_\lam=m$ and sample $r_1,\dots,r_\lam$ uniformly from $\zo{d-1}$. The ciphertext consists of pairs $(r_i, G(k_i,(r_i,s_i)))$ for all $i\in[\lam]$.
\end{itemize}
Decryption involves testing the PRFS generator outputs: for each $i$, a decryptor checks whether the output is $G(k_i,(r_i,0))$. This can be done for PRSF generators with negligible error~\cite{C:AnaQiaYue22}. 

To invoke the PRFS generator security, it is crucial that the queried inputs $x$ are distinct from each other. Since for each $i$ the values $r_i$ in the $q$ ciphertexts queried are uniformly random, we can ensure they do not collide by setting the input size $d$ to be large enough relative to $q$. Namely, by setting $d$ such that the $r_i$ are sampled from a set of size $q^2$, we can ensure that for a given $i$ the probability of the $q$ strings $r_i$ colliding is some constant less than 1. Therefore, the chance of a collision happening for all $i\in[\lam]$ is exponentially small in $\lam$, since the $r_i,r_j$ are sampled independently for $i\neq j$. 
With this in mind, we can switch the output of $G$ to Haar-random states on the $q$ inputs $(r_{\istar},s_{\istar})$, where $\istar$ is the first index such that the $q$ strings $r_{\istar}$ do not collide, and the $q$ bits $s_{\istar}$ are sampled randomly (which matches their marginal distribution). Since the PRFS outputs on index $\istar$ are switched for Haar-random states, they no longer depend on the share $s_{\istar}$, meaning the remaining shares $s_i$ for $i\neq\istar$ are uniform strings independent of the message, yielding security.

Since the size of $\zo{d-1}$ needs to be at least $q^2$ for security, setting $d=2\log q+1$ suffices. Additionally, the output and key length of $G$ need not depend on $q$, which implies ciphertexts and keys are succinct. By inspecting the transformations from CPA-secure SKE to NCE and NCE to collusion-resistant UTE, it is easy to see that succinctness is retained. Due to the $\poly(\lam,\log q)$-size ciphertexts and keys, for every HEST scheme needing $k$ copies, we can find $q=\poly(\lam)$ such that $q>k$, which means HEST contradicts $q$-bounded collusion resistance of UTE. Since $q$-bounded collusion resistance can be achieved for any $q=\poly(\lam)$, the impossibility follows. Note that $k$ circularly depends on $q$ here, so if parameters are not succinct the number of copies $k$ may always be larger than $q$, breaking the impossibility argument. 
We provide the full details in~\cref{sec:qske}.

\paragraph{Separating UE and UTE.} 
Our next result is a construction of a clonable UTE scheme: that is, a secure UTE scheme such that the message information in the ciphertext can be cloned and thus violate UE security. Note that if the number of second stage adversaries is bounded, the question is trivial, since one could just define encryption to encrypt the message many times with a collusion-resistant UTE scheme. Thus, we aim to break UE security for an unbounded polynomial number of second stage adversaries. 

To do this, we will want to make use of the fact that we can repeatedly rerun a given second stage adversary $\B$ for UTE on a state $\st$. Specifically, imagine we have a classical oracle $\cO$ (think of it as an ideal obfuscation) that checks some condition and outputs the message $m$ if the condition holds.
From the perspective of attacking UE, we would like the property that $\A$ can easily sample many inputs that pass the check in the first stage if its output is allowed to be quantum.  
From a UTE security perspective, if $\B$ wins the game, then it must query $\cO$ at an input which passes the check, and we want the property that we can sample prohibitively many inputs that pass the check by running $\B$ repeatedly, enough to break security of some primitive. 

We instantiate this framework using a one-shot message authentication code (OSMAC). An OSMAC is a relaxation of one-shot signatures~\cite{AGKZ20}, which support public verification, and two-tier one-shot signatures~\cite{MPY24}, which support partial public verification. 
An OSMAC consists of a setup algorithm that outputs the public parameters $\pp$ and master verification key $\mvk$, a key generation algorithm that uses $\pp$ to sample a verification key $\vk$ along with a quantum signing key $\ket{\sk}$, a signing algorithm that uses a signing key $\ket{\sk}$ to sign a message $m$ to output a MAC $\sigma$, and finally a verification algorithm that checks $\sigma$ with respect to $\mvk,\pp,\vk,m$. 
Security says that an adversary cannot output valid MACs on two different messages with respect to the same $\vk$ given $\pp$. Our approach will be to use an ``OSMAC chain'' to realize the desired functionality. We illustrate our approach for a chain of length 2, since it is sufficient to elucidate the core idea. Our clonable UTE scheme has separate encryption and decryption keys, and is defined as follows:
\begin{itemize}
    \item The encryption key $\ek$ is a tuple $(\mvk,\pp)$ from the OSMAC scheme, along with a random string $s$, while the decryption key $\dk$ is just $s$.
    
    \item A ciphertext for a message $m$ is a signing key $\ket{\sk_\eps}$, the public parameters $\pp$, and a classical oracle $\cO$ that has $\ek$, $m$, and $\vk_\eps$ hardwired inside it, where $\vk_\eps$ is the verification key generated with $\ket{\sk_\eps}$. Given a strings $s'$ and $r=(r_1,r_2)\in\zo{2}$ along with tuples 
    \[(\vk_0,\vk_1,\sigma_\eps),
    (\vk_{r_1},\vk_{r_1\|0},\vk_{r_1\|1},\sigma_{r_1}),
    (\vk_r,\sigma_r),\] 
    the function $\cO$ checks that $s'=s$, the MAC $\sigma_v$ is valid on message $\vk_{v\|0}\|\vk_{v\|1}$ with respect to $\vk_v$ for $v\in\setbk{\eps,r_1}$, and finally that $\sigma_r$ is a valid MAC on $s'$ with respect to $\vk_r$.
\end{itemize}
To decrypt given $(\ket{\sk_\eps},\pp,\cO)$ and $s$, compute the MAC chain on $00$: sample pairs $(\ket{\sk_v},\vk_v)$ for $v\in\setbk{0,1,00,01}$, compute the MACs $\sigma_\eps$ on $\vk_{0}\|\vk_{1}$ using $\ket{\sk_\eps}$, $\sigma_0$ on $\vk_{00}\|\vk_{01}$ using $\ket{\sk_0}$, and $\sigma_{00}$ on $s$ using $\ket{\sk_{00}}$. Submit the tuples along with $s$ and $00$ to get $m$.
The UE attack for 4 second stage adversaries is simply to compute the MAC tree without the leaves: 
\[(\vk_0,\vk_1,\sigma_\eps),
    (\vk_{0},\vk_{00},\vk_{01},\sigma_{0}),
    (\vk_{1},\vk_{10},\vk_{11},\sigma_{1}).\] 
Additionally send one of $(\vk_{r},\ket{\sk_{r}})$ for $r\in\zo{2}$ to each second stage adversary. When a second stage adversary is given $s$, it can sign $s$ with $\ket{\sk_{r}}$ to complete the MAC chain and recover $m$. Thus, this attack breaks one-way UE security given a single ciphertext in the first stage and 4 second stage adversaries. By increasing the length of $r$ to $\omega(\log \lam)$, the attack works for an unbounded polynomial number of second stage adversaries. 

To see UTE security, first note that the oracle $\cO$ is useless to the first stage adversary $\A$ since $s$ is uniform and independent of $\A$'s view. Furthermore, for the second stage adversary $\B$ to win the security game, it must query the oracle $\cO$ with a valid MAC chain with noticeable probability. Therefore, we can query $\B$ with randomly sampled strings $s$ and measure a random query of $\B$, check that the value of $s$ matches, and add that query to a list $L$. By doing this enough times, we can argue that with high probability there are at least 5 valid MAC chains contained in $L$, which means by the pigeonhole principle, there must exist two chains in $L$ such that the values of $r$ are equal and the values of $s$ are not (which follows by setting $s$ to have sufficient length). 
By picking two elements of $L$ at random, we can find these two MAC chains with roughly $1/\abs{L}^2$ probability. These MAC chains must contradict OSMAC security, because either we have two signatures that verify with $\vk_r$ or the two $\vk_r$ values are different, which means the messages signed with respect to $\vk_{r_1}$ are different. Then, we either have two signatures that verify with $\vk_{r_1}$ or the two $\vk_{r_1}$ values are different, implying there are definitely two signatures that verify with $\vk_\eps$, since $\vk_\eps$ is always fixed. If the length of $r$ is set to $\log^2 \lam$ for the attack, this strategy works by assuming quasi-polynomial security of OSMAC. 

\paragraph{Everlasting secure UTE.}
Numerous works in cryptography have studied the natural notion of ``everlasting security''~\cite{MU07}, which roughly says that after some notable event occurs in a security game, a security property is upheld even given unbounded time after the event. 
Given the two-stage nature of UTE (and UE), it is natural to consider whether everlasting untelegraphability (or unclonability) is possible. Concretely, this means that in the UTE security game, the second stage adversary can run in \textit{unbounded} time. We show that the techniques used for certified everlasting hiding commitments in~\cite{C:HMNY22} adapted to the UTE setting indeed yields everlasting secure UTE.
Given a one-time secure UTE scheme, CPA-secure SKE, and a hash function $H$ modeled as a random oracle, the construction is as follows:
\begin{itemize}
    \item The secret key is the key $k$ for the SKE scheme along with the hash function $H$.
    \item A ciphertext for a message $m$ consists of a one-time UTE encryption $\ket{\ct_\OneUTE}$ of $m$, an encryption $\ct_\SKE$ of a random string $r$, and a string $h=H(r)\xor\sk_\OneUTE$, that masks the fresh one-time UTE key $\sk_\OneUTE$.
\end{itemize}
To decrypt, simply recover $r$ to get $H(r)$ and decrypt the one-time UTE ciphertext.
The basic intuition for security is that by SKE security, the oracle given to $\A$ in the first stage is useless, and thus $\sk_\OneUTE$ is statistically hidden by $H(r)$ until $\B$ learns it in the second stage. Security then follows by security of the one-time UTE scheme, since it is statistically secure. Extending to the collusion-resistant setting follows readily by a hybrid argument. We give the full details in \cref{sec:everlasting-ute}.
This technique can be thought of as leveraging the random oracle to get a stronger variant of NCE. A notable aspect of this construction is that in the single-ciphertext setting, it can be adapted to get everlasting indistinguishability secure public-key UE by using public-key encryption instead of SKE, since indistinguishability secure UE exists in the QROM~\cite{C:AKLLZ22}.   

\paragraph{Impossibility of weakly-efficient shadow tomography.} 
We define weakly-efficient shadow tomography in a similar way to HEST. The main difference is that we explicitly split the procedure into two stages: the first efficient stage outputs any classical string $\st$, and the second stage can use unbounded time to estimate $E(i,\rho)$ up to error $\eps$ given $\st$ and $i$. 
Thus, the attack on collusion-resistant everlasting secure UTE follows the same outline as the HEST attack on collusion-resistant UTE. The main notable difference is that our instantiation of collusion-resistant everlasting secure UTE is secure in the QROM. Thus, we need to slightly adapt the attack to make oracle access to the hash $H$ instead of having $E$ generically compute the decryption circuit, which would need to include a concrete description of $H$ (in the case where the tomography algorithm $\cT$ needs the circuit described in terms of atomic gates only). This can be done by modifying $E$ to be the decryption circuit starting from after $H(r)$ is recovered (we can put $\ct_\SKE$ and $h$ as part of $\A$'s output in the UTE game). We refer to~\cref{sec:west} for the details.

\paragraph{Impossiblity of public-key everlasting secure UTE.}
A natural question is whether a construction of everlasting secure UTE (or UE) is possible in the plain model, as our positive result heavily relies on the QROM. However, we observe that this is not likely to be the case, at least using standard proof strategies which make black-box use of the adversary. 

To start, notice that any efficient reduction $\cR$ that makes black-box use of a UTE adversary $(\A,\B)$ in the everlasting one-way security game to break some candidate falsifiable assumption $A$ cannot query adversary $\B$, since $\B$ is potentially unbounded. Furthermore, there is no requirement on the structure of the string $\st$ output by $\A$, other than it being classical.
Now, suppose that given a ciphertext $\ketct$ of a random message $m$, adversary $\A$ could recover $m$. While $\A$ could just output $m$ in the clear and win the game, what if it were to output some encoding $\Enc(m)$ of the message? If different messages under $\Enc$ are statistically far apart, the unbounded adversary $\B$ could still recover $m$ and win the game. Moreover, if encodings of different messages are computationally close, the reduction $\cR$ would not be able to tell the difference between the encodings $\Enc(m)$ and $\Enc(0)$, meaning $\cR$ cannot even determine whether $(\A,\B)$ wins!

Thus, if we consider an inefficient adversary $\A$ that decrypts the ciphertext $\ketct$ it is given to get a message $m$ and outputs $\Enc(m)$, the reduction $\cR^\A$ that has oracle access to $\A$ must break the candidate assumption $A$, since adversary $\B$ is unbounded and can distinguish encodings of different messages. Since the ciphertext is quantum and there are no guarantees on its structure, we consider public-key UTE (PKUTE), which allows adversary $\A$ to recover $m$ by finding a corresponding secret key to the public key it is given (e.g. by rejection sampling key pairs). Now, given that $\cR^\A$ must break the candidate assumption $A$, we indistinguishably switch oracle access to $\A$ with oracle access to an efficient simulator $\cS$, that will essentially ignore its input and output $\Enc(0)$. Since encodings of different messages are computationally close, $\cR^\cS$ must break the assumption $A$. But $\cR^\cS$ is efficient, so $A$ must be false. This means there is either no black-box reduction to a falsifiable assumption $A$ that shows everlasting one-way PKUTE security, or $A$ is false. We supply the formal details in~\cref{sec:everlasting-impossible}.

\paragraph{Unbounded joint-leakage-resilient secret sharing.}
In~\cref{sec:utss}, we show that collusion-resistant UTE and classical secret sharing for all policies can be used to construct unbounded joint-leakage-resilient secret sharing (UJLRSS) for all polynomial-size monotone Boolean circuits. In this overview, we will show how to construct UJLRSS for all polynomial-size monotone Boolean formulas, since it is simpler and still captures the basic intuition.

We start by defining UJLRSS more formally. The syntax of UJLRSS is the same as regular $n$-party secret sharing, except that the shares of a classical message $m$ output by the sharing algorithm are quantum states, and thus the input to the reconstruction algorithm contains quantum states. Security has two stages:
\begin{itemize}
    \item In the first stage, an adversary $\A_0$ outputs messages $m_0,m_1$ and a ``collusion partition'' of disjoint sets $V^\ast,V_1,\dots,V_\ell$ such that their union is the set of all indices $[n]$, and each $V$ does not satisfy the policy. 
    The challenger then shares $m_b$, and for each $i\in[\ell]$ gives the shares $\ket{s_j}$ for $j\in V_i$ to adversary $\A_i$. Each adversary $\A_i$ outputs a state $\st_i$.
    \item In the second stage, adversary $\B$ is given $\st_i$ for each $i\in[\ell]$, along with shares $\ket{s_j}$ for $j\in V^\ast$, and tries to guess the bit $b$.
\end{itemize}
This definition captures the scenario where an adversarial group that does not satisfy the policy is able to exfiltrate the remaining shares classically. In such a case, our security notion guarantees that the adversarial group learns nothing about the secret. Alternatively, this definition can be thought of as ensuring that each individual share, along with each possible set of unqualified shares, is untelegraphable.

We now present a construction of UJLRSS for a monotone Boolean formula $P$ from collusion-resistant UTE and a classical secret sharing scheme for $P$: 
\begin{itemize}
    \item To share a message $m$ with respect to a policy $P$, sample classical shares $(s'_1,\dots,s'_n)$ for $m$ and $(\sk_1,\dots,\sk_n)$ for a fresh UTE secret key $\sk_\UTE$. A share $\ket{s_i}$ for party $i\in[n]$ is $\sk_i$ along with an encryption $\ket{\ct_i}$ of $s'_i$. 
    \item To reconstruct given a satisfying group of parties, reconstruct the key $\sk_\UTE$, decrypt the ciphertexts to get shares of $m$, and reconstruct $m$ from those shares. 
\end{itemize}
The intuition for security is that for each set $V_i$, we can appeal to collusion-resistant UTE security to switch the ciphertexts to encryptions of 0 instead of shares of $m$. 
To simulate the secret key shares $\sk_j$ for $\A_i$, sample them uniformly, since classical secret sharing for formulas is information-theoretically simulation secure. After $\sk_\UTE$ is revealed in the UTE security game, derive the remaining shares and run $\A_j$ for $j\neq i$ to get the remaining states, which can then be given to the second stage UJLRSS adversary to complete the reduction. After switching ciphertexts for all $V_i$, security follows by information theoretic security of the classical scheme, since the only classical shares left are the ones in $V^\ast$. Since classical secret sharing for all monotone Boolean circuits is computationally secure, we can no longer easily simulate the key shares for an unsatisfying group, so the construction and proof strategy must be changed to adapt to this. We provide the full details in~\cref{sec:utss}.

\paragraph{Untelegraphable functional encryption.}
Finally, we give a sketch of our untelegraphable secret-key and public-key functional encryption constructions. To start, we formalize the security definition of untelegraphable functional encryption. The first stage adversary $\A$ can query function keys and still must choose messages $m_0,m_1$ such that $f(m_0)=f(m_1)$ for all functions $f$ queried, much like regular functional encryption, except $\A$ now outputs a state $\st$. The second stage adversary $\B$ can query function keys as well, but with no constraints on the function $f$, even functions such that $f(m_0)\neq f(m_1)$. One can also consider the stronger definition where $\B$ is given the master secret key in the second stage, but in this work we only consider the unrestricted key query setting.

We start by constructing untelegraphable secret-key functional encryption (UTSKFE) from challenge-only single-decryptor functional encryption (SDFE)\footnote{We technically need to modify the definition to the untelegraphable setting such that it is stated quite differently than SDFE in~\cite{AC:KitNis22}, but for ease of exposition, we call it SDFE here.}~\cite{AC:KitNis22} and plain secret-key functional encryption (SKFE). 
This construction essentially switches the roles of the SDFE function keys and ciphertext. We describe it as follows:
\begin{itemize}
    \item The master key is an SDFE key pair $(\pk_\SDFE,\msk_\SDFE)$ along with a plain SKFE master key $\msk_\SKFE$.
    
    \item A function key for a function $f$ is an SDFE encryption $\ct_f$ of a freshly sampled function key $\sk_f$ from the SKFE scheme.
    
    \item A ciphertext for a message $m$ is a function key $\ket{\sk_G}$ for the SDFE scheme, where $G$ has a SKFE ciphertext $\ct_\SKFE$ of $m$ hardwired, and on input an SKFE function key $\sk_f$, decrypts $\ct_\SKFE$ with $\sk_f$. 
\end{itemize}
To decrypt, simply run the decryption algorithm of the SDFE scheme. Intuitively, security follows by appealing to SDFE security to switch each second stage key query to encrypt 0 instead of $\sk_f$. Since the first stage is essentially the plain SKFE security game, security now follows by SKFE security. Since the SDFE scheme is challenge-only secure, this yields single-ciphertext secure untelegraphable SKFE. See~\cref{sec:utskfe} for further details.

For our final construction, we show how to build untelegraphable public-key functional encryption from single-ciphertext secure UTSKFE, plain public-key functional encryption (PKFE), plain SKE, and a secure PRF. This construction follows the hybrid encryption blueprint from~\cite{C:ABSV15}. We sketch the construction as follows:
\begin{itemize}
    \item The public key is just the public-key $\pk_\PKFE$ for the PKFE scheme.
    The master key is the master key $\msk_\PKFE$ for the PKFE scheme along with a secret key $\sk_\SKE$ for the SKE scheme.
    
    \item A function key for a function $f$ is a PKFE function key $\sk_G$, where $G[\ct_\SKE,f]$ is a function with $f$ and a ``dummy ciphertext'' $\ct_\SKE$ of 0 under $\sk_\SKE$ hard-wired. The function $G$ takes as input an indicator bit $b$ along with three keys $\msk_\SKFE$, $k$, and $\sk_\SKE$, where $k$ is a PRF key and $\msk_\SKFE$ is a master key for the UTSKFE scheme. 
    If $b=0$, $G$ outputs $\sk_f$ using $\msk_\SKFE$ along with randomness from the PRF on $f$, and if $b=1$, $G$ outputs the decryption of $\ct_\SKE$ using $\sk_\SKE$.
    
    \item A ciphertext for a message $m$ is a ciphertext $\ket{\ct_\SKFE}$ of $m$ for freshly sampled $\msk_\SKFE$ along with a PKFE ciphertext $\ct_\PKFE$ of the message $(0,\msk_\SKFE,k,\bot)$, where $k$ is a fresh PRF key. 
\end{itemize}
To decrypt, run PKFE decryption to get a key $\sk_f$ for the UTSKFE scheme and use $\sk_f$ to decrypt $\ket{\ct_\SKFE}$ for the message. Security follows the approach from~\cite{C:ABSV15}. First, we switch the dummy ciphertext $\ct_\SKE$ to encrypt $\sk_f$ of the UTSKFE scheme with randomness derived from the PRF on $f$. Then we switch the PKFE message from $(0,\msk_\SKFE,k,\bot)$ to $(1,\bot,\bot,\sk_\SKE)$, which crucially relies on the fact that 
\[G[\ct_\SKE,f](0,\msk_\SKFE,k,\bot)=G[\ct_\SKE,f](1,\bot,\bot,\sk_\SKE)\] 
for all $f$, even when $f(m_0)\neq f(m_1)$, a property that was not utilized in~\cite{C:ABSV15}. To complete the proof, we switch the key generation randomness for generating $\sk_f$ from the UTSKFE scheme from the PRF on $f$ to a random string, and finally appeal to UTSKFE security. See~\cref{sec:utpkfe} for details.

\section{Preliminaries}
Let $\lam$ denote the security parameter unless otherwise specified. We write $\poly(\cdot)$ to denote an arbitrary polynomial and $\negl(\cdot)$ to denote an arbitrary negligible function. We say an event happens with overwhelming probability if the probability is at least $1-\negl(\lam)$. 

For strings $x$ and $y$, we use $x \concat y$ and $(x,y)$ to denote the concatenation of $x$ and $y$.
Let $[\ell]$ denote the set of integers $\{1, \cdots, \ell \}$.
In this paper, for a finite set $X$ and a distribution $D$, $x \getsr X$ denotes selecting an element from $X$ uniformly at random, $x \chosen D$ denotes sampling an element $x$ according to $D$. Let $y \gets \A(x)$ denote assigning to $y$ the output of a probabilistic or quantum algorithm $\A$. When we explicitly show that $\A$ uses randomness $r$, we write $y \gets \A(x;r)$.
We say an algorithm runs in quantum polynomial time (QPT) if it runs in quantum polynomial time in the size of its input, and we say an algorithm runs in probabalistic polynomial time (PPT) if it runs in classical polynomial time in the size of its input.

\subsection{Quantum Information}\label{sec:quantum-info}
Let $\cH$ be a finite-dimensional complex Hilbert space. A (pure) quantum state is a vector $\ket{\psi}\in \cH$.
Let $\cS(\cH)$ be the space of Hermitian operators on $\cH$. A density matrix is a Hermitian operator $\rho \in \cS(\cH)$ with $\Trace(\qstate{X})=1$, which is a probabilistic mixture of pure states.
A quantum state over $\mathcal{H}=\mathbb{C}^2$ is called qubit, which can be represented by the linear combination of the standard basis $\setbk{\ket{0},\ket{1}}$. More generally, a quantum system over $(\mathbb{C}^2)^{\tensor n}$ is called an $n$-qubit quantum system for $n \in \mathbb{N} \setminus \setbk{0}$.

We write $\cH_{\qreg{R}}$ to denote that the Hilbert space $\cH$ is tied to the register $\qreg{R}$.
Also, we sometimes write $\qstate{X}_{\qreg{R}}$ to emphasize that the operator $\qstate{X}$ acts on the register $\qreg{R}$.
When we apply $\qstate{X}_{\qreg{R}_1}$ to registers $\qreg{R}_1$ and $\qreg{R}_2$, $\qstate{X}_{\qreg{R}_1}$ is identified with $\qstate{X}_{\qreg{R}_1} \tensor \mat{I}_{\qreg{R}_2}$.

A unitary operation is represented by a complex matrix $\mat{U}$ such that $\mat{U}\mat{U}^\dagger = \mat{I}$. The operation $\mat{U}$ transforms $\ket{\psi}\in\cH$ and $\rho\in\cS(\cH)$ into $\mat{U}\ket{\psi}$ and $\mat{U}\rho\mat{U}^\dagger$, respectively.
A projector $\mat{P}$ is a Hermitian operator ($\mat{P}^\dagger =\mat{P}$) such that $\mat{P}^2 = \mat{P}$.

For a quantum state $\rho$ over two registers $\qreg{R}_1$ and $\qreg{R}_2$, we denote the state in $\qreg{R}_1$ as $\rho[\qreg{R}_1]$, where $\rho[\qreg{R}_1]= \Trace_2[\rho]$ is a partial trace of $\rho$ (trace out $\qreg{R}_2$).
For quantum states $\rho_1$ and $\rho_2$, $\TD(\rho_1,\rho_2)$ denotes the trace distance between $\rho_1$ and $\rho_2$.

\subsection{Quantum Accessible Oracle and Useful Lemma}\label{sec:qrom}
Given a function $F: X\ra Y$, a quantum-accessible oracle $O$ of $F$ is modeled by a unitary transformation $\mat{U}_F$ operating on two registers $\qreg{in}\tensor\qreg{out}$, in which $\ket{x}\ket{y}$ is mapped to $\ket{x}\ket{y\oplus F(x)}$, where $\oplus$ denotes XOR group operation on $Y$.
We write $\A^{\ket{O}}$ to denote that the algorithm $\A$'s oracle $O$ is a quantum-accessible oracle.

\paragraph{Simulation of quantum random oracles.}
In this paper, following many previous works in the QROM, we give quantum-accessible random oracles to reduction algorithms if needed. This is just a convention. We can efficiently simulate quantum-accessible random oracles perfectly by using $2q$-wise independent hash function~\cite{C:Zhandry12}, where $q$ is the number of queries to the quantum-accessible random oracles by an adversary.

\paragraph{One-Way to Hiding (O2H) Lemma.} We now recall the one-way to hiding lemma~\cite{C:AmbHamUnr19}.

\begin{lemma}[O2H Lemma~\cite{C:AmbHamUnr19}]\label{lem:O2H}
Let $G,H:X\ra Y$ be any functions, $z$ be a random value, and $S\subseteq X$ be a random set such that $G(x)=H(x)$ holds for every $x\notin S$.
The tuple $(G,H,S,z)$ may have arbitrary joint distribution.
Furthermore, let $\A$ be a quantum oracle algorithm that makes at most $q$ quantum queries.
Let $\B$ be an algorithm that on input $z$, choose $i\gets[q]$, runs $\A^H(z)$, measures $\A$'s $i$-th query, and outputs the measurement outcome.
Then, we have
\begin{align*}
\abs{\Pr[\A^{H}(z)=1]-\Pr[\A^{G}(z)=1]} \leq 2q\cdot\sqrt{\Pr[\B^{H}(z)\in S]}
\enspace.
\end{align*}
\end{lemma}

\subsection{Cryptographic Primitives}\label{sec:crypto-prelims}
In this section, we introduce the core primitive of this work, which we call untelegraphable encryption. In doing so, we recall the notion of unclonable encryption, which will be needed at points throughout this work as well. 
\ifelsesubmit{Additionally, we provide the definition of CPA-secure secret-key encryption (SKE) for reference in~\cref{sec:additional-prelim}.}
{Additionally, we provide the definition of CPA-secure secret-key encryption (SKE) for reference.}
\paragraph{Unclonable encryption.} 
We now recall the definition of unclonable secret key encryption (USKE)~\cite{TQC:BL20,EPRINT:AnaKal21}, which is often referred to as just unclonable encryption (UE).
\begin{definition}[Unclonable Secret Key Encryption]\label{def:UE}
    An unclonable secret key encryption scheme with message space $\cM=\setbk{\cM_\lam}_{\lam\in\N}$ is a tuple of QPT algorithms $\Pi_\UE=(\Gen,\Enc,\Dec)$ with the following syntax:
    \begin{itemize}
        \item $\Gen(1^\lam)\to\sk:$ On input the security parameter $\lam$, the key generation algorithm outputs a (classical) secret key $\sk$.
        \item $\Enc(\sk,m)\to\ket{\ct}:$ On input the secret key $\sk$ and a message $m\in\cM_\lam$, the encryption algorithm outputs a quantum ciphertext $\ket{\ct}$.
        \item $\Dec(\sk,\ket{\ct})\to m:$ On input the secret key $\sk$ and a quantum ciphertext $\ket{\ct}$, the decryption algorithm outputs a message $m\in\cM_\lam$.
    \end{itemize}
    We require that $\Pi_\UE$ satisfy the following correctness property, along with at least one of the following security properties:
    \begin{itemize}
        \item \textbf{Correctness:} For all $\lam\in\N$ and any message $m\in\cM_\lam$, we have
        \[
        \Pr\left[m'=m:
        \begin{array}{c}
        \sk\gets\Gen(1^\lam)\\
        \ketct\gets\Enc(\sk,m)\\
        m'\gets\Dec(\sk,\ketct)
        \end{array}\right]\geq 1-\negl(\lam)
        \]
        \item \textbf{One-way Security:} For a security parameter $\lam$ and a two-stage adversary $(\A,\B,\C)$, we define the one-way security game as follows:
        \begin{enumerate}
            \item At the beginning of the game, the challenger samples $\sk\gets\Gen(1^\lam)$, $m\getsr\cM_\lam$, and $\ketct\gets\Enc(\sk,m)$. It gives $\ketct$ to $\A$.
            \item Adversary $\A$ outputs a quantum state $\rho_{BC}$ in register $B$ and $C$, and sends the corresponding registers to $\B$ and $\C$. 
            \item The challenger gives $\sk$ to $\B$ and $\C$. Adversaries $\B$ and $\C$ output strings $m_\B$ and $m_\C$ respectively. The output of the experiment is a bit $b$ which is 1 when $m=m_\B=m_\C$ and 0 otherwise.
        \end{enumerate}
        We say a USKE scheme satisfies (computational) one-way security if for all (efficient) adversaries $(\A,\B,\C)$ there exists a negligible function $\negl(\cdot)$ such that for all $\lam\in\N$, $\Pr[b=1]=\negl(\lam)$ in the one-way security game.
    
        \item \textbf{Indistinguishability Security:} For a security parameter $\lam$ and a two-stage adversary $(\A,\B,\C)$, we define the indistinguishability security game as follows:
        \begin{enumerate}
            \item Adversary $\A$ outputs two messages $m_0,m_1\in\cM_\lam$.
            \item The challenger samples $\sk\gets\Gen(1^\lam)$, $b\getsr\zo{}$, and $\ketct\gets\Enc(\sk,m_b)$. It gives $\ketct$ to $\A$.
            \item Adversary $\A$ outputs a quantum state $\rho_{BC}$ in register $B$ and $C$, and sends the corresponding registers to $\B$ and $\C$. 
            \item The challenger gives $\sk$ to $\B$ and $\C$. Adversaries $\B$ and $\C$ output bits $b_\B$ and $b_\C$ respectively. The output of the experiment is a bit $b$ which is 1 when $b=b_\B=b_\C$ and 0 otherwise.
        \end{enumerate}
        We say a USKE scheme satisfies statistical (computational) indistinguishability security if for all (efficient) adversaries $(\A,\B,\C)$ there exists a negligible function $\negl(\cdot)$ such that for all $\lam\in\N$, $\Pr[b=1]\leq 1/2+\negl(\lam)$ in the above indistinguishability security game. 
    \end{itemize}
\end{definition}

\begin{definition}[UE with $k$ second stage adversaries]\label{def:k-UE}
    We say $\Pi_\UE$ satisfies \textit{$k$-adversary security} if it is secure for first stage adversaries $\A$ in~\cref{def:UE} that output a quantum state $\rho$ to $k$ registers $B_1,\dots,B_k$ that are given to corresponding second stage adversaries $(\B_1,\dots,\B_k)$. All adversaries $B_1,\dots,B_k$ must guess $m$ in the one-way security game or the challenge bit $b$ in the indistinguishability security game. 
\end{definition}

\begin{theorem}[{Statistical One-Way UE~\cite{TQC:BL20,C:AKLLZ22}}]\label{thm:OWUE}
    There exists statistically one-way secure USKE where the maximum advantage of the adversary is $2^{-\Omega(\lam)}$.
\end{theorem}

\paragraph{Untelegraphable encryption.} We now define the core primitive considered in this work. Untelegraphable encryption has the same syntax and correctness requirements as unclonable encryption above. The main security difference is that the output of the first stage adversary must be purely classical. In particular, this means there need not be two second stage adversaries as in unclonable encryption, since the classical string can be copied. 

\begin{definition}[Untelegraphable Secret Key Encryption]\label{def:UTE}
    An untelegraphable secret key encryption (UTSKE) scheme with message space $\cM=\setbk{\cM_\lam}_{\lam\in\N}$ is a tuple of QPT algorithms $\Pi_\UTE=(\Gen,\Enc,\Dec)$ with the same syntax and correctness requirement as \cref{def:UE}.
    We require that $\Pi_\UTE$ satisfy at least one of the following security properties:
    \begin{itemize}
        \item \textbf{One-way Security:} For a security parameter $\lam$ and a two-stage adversary $(\A,\B)$, we define the one-way security game as follows:
        \begin{enumerate}
            \item At the beginning of the game, the challenger samples $\sk\gets\Gen(1^\lam)$, $m\getsr\cM_\lam$, and $\ketct\gets\Enc(\sk,m)$. It gives $\ketct$ to $\A$.
            \item Adversary $\A$ outputs a classical string $\st$ which is given to adversary $\B$. 
            \item The challenger gives $\sk$ to $\B$. Adversary $\B$ outputs a string $m_\B$. The output of the experiment is a bit $b$ which is 1 when $m=m_\B$ and 0 otherwise.
        \end{enumerate}
        We say a UTSKE scheme satisfies (computational) one-way security if for all (efficient) adversaries $(\A,\B)$ there exists a negligible function $\negl(\cdot)$ such that for all $\lam\in\N$, $\Pr[b=1]=\negl(\lam)$ in the one-way security game.
    
        \item \textbf{Indistinguishability Security:} For a security parameter $\lam$, a bit $b\in\zo{}$, and a two-stage adversary $(\A,\B)$, we define the indistinguishability security game as follows:
        \begin{enumerate}
            \item Adversary $\A$ outputs two messages $m_0,m_1\in\cM_\lam$.
            \item The challenger samples $\sk\gets\Gen(1^\lam)$ and $\ketct\gets\Enc(\sk,m_b)$. It gives $\ketct$ to $\A$.
            \item Adversary $\A$ outputs a classical string $\st$ which is given to adversary $\B$.  
            \item The challenger gives $\sk$ to $\B$. Adversary $\B$ outputs a bit $b_\B$, which is the output of the experiment.
        \end{enumerate}
        We say $\Pi_\UTE$ satisfies statistical (computational) indistinguishability security if for all (efficient) adversaries $(\A,\B)$ there exists a negligible function $\negl(\cdot)$ such that for all $\lam\in\N$, 
        \[\abs{\Pr[b_\B=1|b=0]-\Pr[b_\B=1|b=1]}= \negl(\lam)\] in the above indistinguishability security game.
    \end{itemize}
\end{definition}

\begin{definition}[Collusion-Resistant Security]\label{def:cr-sec}
    We say that an untelegraphable (unclonable) SKE scheme $\Pi_\UTE$ ($\Pi_\UE$) satisfies \textit{collusion-resistant} security if the first-stage adversary $\A$ in \cref{def:UTE} (\cref{def:UE}) can make many message queries $\setbk{m^{(i)}_0,m^{(i)}_1}_{i\in[Q]}$ adaptively and get encryptions of $\setbk{m^{(i)}_b}_{i\in[Q]}$ in return. 
    We say the scheme is \textit{$t$-copy secure} if the first-stage adversary specifies a single pair of messages $(m_0,m_1)$ and gets $t$ encryptions of $m_b$ in return. 
\end{definition}

\begin{definition}[Everlasting Security]\label{def:everlasting}
    We say that an untelegraphable (unclonable) SKE scheme $\Pi_\UTE$ ($\Pi_\UE$) satisfies \textit{everlasting} security if the first stage adversary $\A$ in \cref{def:UTE} (\cref{def:UE}) is efficient while the second stage adversary is allowed to run in unbounded time. 
\end{definition}

\begin{corollary}[Statistical One-Way UTSKE]\label{cor:ow-utske}
    There exists statistically one-way secure UTSKE where the maximum advantage of the adversary is $2^{O(\lam)}$.
\end{corollary}
\begin{proof}
    Follows immediately by~\cref{thm:OWUE}. In particular, any telegraphing adversary defines a cloning adversary which simply gives $\st$ to both $\B$ and $\C$.
\end{proof}

\newcommand{\skedef}{
\paragraph{CPA-secure encryption.} We finally recall the basic notion of CPA-secure secret-key encryption, which will be used at various points throughout this work. Classical CPA-secure encryption can be constructed from any one-way function.

\begin{definition}[Secret-Key Encryption]\label{def:ske}
    A secret key encryption scheme with message space $\cM=\setbk{\cM_\lam}_{\lam\in\N}$ is a tuple of PPT algorithms $\Pi_\SKE=(\Gen,\Enc,\Dec)$ with the following syntax:
    \begin{itemize}
        \item $\Gen(1^\lam)\to\sk:$ On input the security parameter $\lam$, the key generation algorithm outputs a secret key $\sk$.
        \item $\Enc(\sk,m)\to\ct:$ On input the secret key $\sk$ and a message $m\in\cM_\lam$, the encryption algorithm outputs a ciphertext $\ct$.
        \item $\Dec(\sk,\ct)\to m:$ On input the secret key $\sk$ and a ciphertext $\ct$, the decryption algorithm outputs a message $m\in\cM_\lam$.
    \end{itemize}
    We require that $\Pi_\SKE$ satisfy the following correctness property, along with at least one of the following security properties:
    \begin{itemize}
        \item \textbf{Correctness:} For all $\lam\in\N$ and any message $m\in\cM_\lam$, we have
        \[
        \Pr\left[m'=m:
        \begin{array}{c}
        \sk\gets\Gen(1^\lam)\\
        \ct\gets\Enc(\sk,m)\\
        m'\gets\Dec(\sk,\ct)
        \end{array}\right]\geq 1-\negl(\lam)
        \]
        \item \textbf{One-way CPA Security:} For a security parameter $\lam$ and an adversary $\A$, we define the one-way security game as follows:
        \begin{enumerate}
            \item At the beginning of the game, the challenger samples $\sk\gets\Gen(1^\lam)$, $m^\ast\getsr\cM_\lam$, and $\ct\gets\Enc(\sk,m^\ast)$. It gives $\ct$ to $\A$.
            \item Adversary $\A$ has access to an encryption oracle that on input $m\in\cM_\lam$ outputs $\Enc(\sk,m)$. 
            \item At the end of the game, adversary $\A$ outputs $m'\in\cM_\lam$. The output of the experiment is a bit $b$ which is 1 when $m'=m^\ast$ and 0 otherwise.
        \end{enumerate}
        We say an SKE scheme satisfies one-way CPA security if for all efficient adversaries $\A$ there exists a negligible function $\negl(\cdot)$ such that for all $\lam\in\N$, $\Pr[b=1]=\negl(\lam)$ in the one-way security game.
    
        \item \textbf{CPA Security:} For a security parameter $\lam$, a bit $b\in\zo{}$, and an adversary $\A$, we define the indistinguishability security game as follows:
        \begin{enumerate}
            \item At the beginning of the game, the challenger samples $\sk\gets\Gen(1^\lam)$.
            \item Adversary $\A$ can make encryption queries on pairs of messages $m_0,m_1\in\cM_\lam$. The challenger replies with $\ct\gets\Enc(\sk,m_b)$ 
            \item At the end of the game, adversary $\A$ outputs $b'\in\zo{}$, which is the output of the experiment.
        \end{enumerate}
        We say an SKE scheme satisfies CPA-security if for all efficient adversaries $\A$ there exists a negligible function $\negl(\cdot)$ such that for all $\lam\in\N$, \[\abs{\Pr[b'=1|b=0]-\Pr[b'=1|b=1]}= \negl(\lam)\] in the above indistinguishability security game. We say $\Pi_\SKE$ is \textit{query-bounded} if the adversary can make at most $q$ total queries, where $q$ is additionally given to $\Gen$.  
    \end{itemize}
\end{definition}

\begin{remark}[Quantum Ciphertexts]
    We can easily generalize \cref{def:ske} to the setting where ciphertexts are quantum, which we abbreviate by QSKE. 
\end{remark}

}
\ifnotsubmit{\skedef}

\section{Untelegraphable Encryption and the Impossibility of Hyper-Efficient Shadow Tomography}\label{sec:ute-basic}
In this section, we construct untelegraphable SKE with indistiguishability security from one-way secure untelegraphable SKE. 
Next, we show how to upgrade the UTSKE to satisfy collusion-resistant security by additionally assuming the existence of non-committing encryption for receiver (NCER), which can be constructed from CPA-secure encryption. 
We then recall the notion of hyper-efficient shadow tomography (HEST), and use its existence to attack any collusion-resistant UTSKE scheme.
This implies that HEST cannot exist for general states under the assumption that CPA-secure encryption exists.
Finally, we show pseudorandom state (PRS) generators are sufficient for the impossibility of HEST, by constructing bounded-query CPA-secure encryption with succinct quantum ciphertexts and classical keys from PRS generators, where succinctness is with respect to the bound parameter. \ifsubmit{We will state the main theorems and corollaries in this section, but defer many of the proofs to~\cref{sec:omitted}.}

\subsection{Building Blocks}
In this section, we describe preliminaries and notions that are needed for our constructions. 

\paragraph{Min-entropy.} We recall some basic definitions on min-entropy. Our definitions
are adapted from those in~\cite{DRS04}.
For a (discrete) random variable $X$, we write
$\Hinf(X) = -\log(\max_x \Pr[X = x])$ to denote its min-entropy. For two (possibly correlated)
discrete random variables $X$ and $Y$, we define the average min-entropy of $X$ given $Y$ to be
$\Hinf(X \mid Y) = -\log(\mathbb{E}_{y \gets Y} \max_x \Pr[X = x \mid Y = y])$. The optimal probability
of an unbounded adversary guessing $X$ given the correlated value $Y$ is
$2^{-\Hinf(X \mid Y)}$.

\paragraph{Leftover hash lemma.} Our construction will also rely on the generalized
leftover hash lemma (LHL) from~\cite{BDKPPSY11}:
\begin{theorem}[{LHL with Conditional Min-Entropy~\cite[Theorem 3.2, adapted]{BDKPPSY11}}]
  \label{thm:cond-LHL}
    Let $(X, Z)$ be random variables sampled from some
    joint distribution $\cD$ over $\cX\times\cZ$.
    Let $\cH=\{h \colon \cX\to\zo{v}\}$ be a family of universal hash functions, and
    let  $L = \Hinf(X\mid Z)-v$ be the entropy loss.
    Let $\cA(r, h, z)$ be a (possibly probabilistic) distinguisher where
    \[ \Pr[\cA(r, h, z)=1 : r \getsr \zo{v}, h \getsr \cH, (x, z) \gets \cD] \leq \eps. \]
    Then, the distinguishing advantage of $\cA$ on the following distributions is at most $\sqrt{\eps/2^L}$:
    \begin{equation*}
      \left\{ 
        (h(x), h, z) : 
        \begin{array}{c} (x, z) \gets \cD \\ h \getsr \cH \end{array}
      \right\} \text{ and }
      \left\{ 
        (r, h, z) : 
        \begin{array}{c} (x, z) \gets \cD \\ r \getsr \zo{v}, h \getsr \cH \end{array}
      \right\}
    \end{equation*}
\end{theorem}

\begin{corollary}[{LHL with Conditional Min-Entropy}]\label{cor:cond-LHL}
    Let $(X,Z)$ be random variables sampled from some
    joint distribution $\cD$ over $\cX\times\cZ$.
    Let $\cH=\{h \colon \cX\to\zo{v}\}$ be a family of universal hash functions.
    Let $L = \Hinf(X\mid Z)-v$ be the entropy loss. Then the statistical distance between
    the following distributions is at most $2^{-L / 2}$:
    \begin{equation*}
      \left\{ 
        (h(x), h, z) : 
        \begin{array}{c} (x, z) \gets \cD \\ h \getsr \cH \end{array}
      \right\} \text{ and }
      \left\{ 
        (r, h, z) : 
        \begin{array}{c} (x, z) \gets \cD \\ r \getsr \zo{v}, h \getsr \cH \end{array}
      \right\}
    \end{equation*}
\end{corollary}
\begin{proof}
  Follows by setting $\eps=1$ in \cref{thm:cond-LHL} (which captures {\em all} distinguishers).
\end{proof}

\paragraph{Pseudorandom function-like state generators.} We recall the notion of a pseudorandom function-like state (PRFS) generator~\cite{C:AnaQiaYue22}. 
Intuitively, this is the quantum analog of a pseudorandom function (PRF). 
For polynomial sized domains, PRFS generators can be constructed from pseudorandom state (PRS) generators~\cite{C:JiLiuSon18}, which may not imply one-way functions~\cite{Kre21}.
We use the definition of a PRFS generator since it will be more convenient.

\begin{definition}[{PRFS Generator~\cite{C:AnaQiaYue22}}]\label{def:prfs}
    Let $\lam$ be a security parameter. 
    We say that a QPT algorithm $G$ is a selectively secure pseudorandom function-like state generator if for all polynomials $s=s(\lam),t=t(\lam)$, QPT adversaries $\A$ and a family of distinct indices $(\setbk{x_1,\dots,x_s}\subseteq\zo{d(\lam)})_\lam$, there exists a negligible function $\negl(\lam)$ such that for all $\lam\in\N$,
    \begin{align*}
        \bigg\lvert
        &\Pr_{k\getsr\zo{\lam}}[\A(\setbk{x_i,G_\lam(k,x_i)^{\otimes t}}_{i\in[s]})=1]\\
        &-\Pr_{\ket{\psi_1},\dots,\ket{\psi_s}\gets \mu_n}[
        \A(\setbk{x_i,\ket{\psi_i}^{\otimes t}}_{i\in[s]})=1
        ]
        \bigg\rvert=\negl(\lam),
    \end{align*}
    where $\mu_n$ is the Haar measure on $n$-qubit states. 
    We note that there is an efficient algorithm $\test(k,x,G(k,x'))$ that determines if $x'=x$ with some $\negl(\lam)$ error so long as $n=\omega(\log\lam)$~\cite{C:AnaQiaYue22}. 
\end{definition}

\paragraph{Non-committing encryption.} We also recall the notion of non-committing encryption (NCE). Specifically, we define secret-key non-committing encryption for receiver (SK-NCER), which is the secret-key variant of NCER \cite{EC:JarLys00,TCC:CanHalKat05}. 
Roughly speaking, SK-NCER is SKE which allows one to generate a ``fake'' ciphertext so that it can be later revealed to any message along with a ``fake'' decryption key.
The formal definition is given below.

\begin{definition}[Secret-Key Non-Committing Encryption for Receiver]\label{def:nce}
Let $\lam$ be a security parameter. 
A secret key non-committing encryption for receiver (SK-NCER) scheme with message space $\cM=\setbk{\cM_\lam}_{\lam\in\N}$ is a tuple of efficient algorithms $\Pi_\NCE=(\KG, \Enc, \Dec, \Fake, \Reveal)$ with the following syntax: 
\begin{itemize}
\item $\KG(1^\lam)\to(\ek,\dk)$: On input the security parameter $1^\lambda$, the key generation algorithm outputs an encryption key $\ek$ and decryption key $\dk$.\footnote{One may wonder why encryption and decryption keys are separately defined though we consider a secret-key primitive. The reason is that the security of SK-NCER involves an adversary that obtains a decryption key, and we cannot include all secret information in a decryption key.}

\item $\Enc(\ek,m)\to\ct$: On input an encryption key $\ek$ and message $m \in \cM_\lam$, the encryption algorithm outputs a ciphertext $\ct$.

\item $\Dec(\dk,\ct)\to m$: On input a decryption key $\dk$ and ciphertext $\ct$, the decryption algorithm outputs a message $m \in \cM_\lam$.

\item $\Fake(\ek)\to(\ct,\st)$: On input an encryption key $\ek$, the ciphertext faking algorithm outputs a fake ciphertext $\ct$ and state $\st$.

\item $\Reveal(\st,\mstar)$: On input a state $\st$ and message $\mstar\in\cM_\lam$, the reveal algorithm outputs a fake decryption key $\dkstar$.
\end{itemize}
We require that $\Pi_\NCE$ satisfy the following properties:
\begin{itemize}
    \item \textbf{Correctness:} For all $\lam\in\N$ and $m\in\cM_\lam$,
    \[
        \Pr[\Dec(\dk, \Enc(\ek, m)) = m : (\ek,\dk)\gets\KG(1^\lam)]\geq 1-\negl(\lam).    
    \]
    \item \textbf{Non-committing Security:} For a security parameter $\lam$, a bit $b\in\zo{}$, and an adversary $\A$, we define the non-committing security game as follows:
    \begin{enumerate}
        \item At the beginning of the game, the challenger samples $(\ek,\dk)\la \KG(1^{\secp})$.
        \item Adversary $\A$ makes arbitrarily many encryption queries and a single challenge query in any order:
        \begin{itemize}
            \item \textbf{Encryption Query:} When $\A$ makes an encryption query $m\in\cM_\lam$, the challenger computes $\ct \la \Enc(\ek, m)$ and returns $\ct$ to $\A$.
            \item \textbf{Challenge Query:} When $\A$ makes a challenge query $\mstar\in\cM_\lam$, the challenger returns a tuple $(\ctstar,\dkstar)$. 
            If $b=0$, the challenger computes $\ctstar\gets\Enc(\ek,\mstar)$ and sets $\dkstar=\dk$.
            If $b=1$, the challenger computes $(\ctstar,\st)\gets\Fake(\ek)$ and $\dkstar\gets\Reveal(\st,\mstar)$.
        \end{itemize}
        \item Adversary $\A$ outputs a bit $b'\in\zo{}$, which is the output of the experiment.
    \end{enumerate}
    We say $\Pi_\NCE$ satisfies non-committing security if for all efficient adversaries $\A$ there exists a negligible function $\negl(\cdot)$ such that for all $\lam\in\N$, \[\abs{\Pr[b'=1|b=0]-\Pr[b'=1|b=1]}= \negl(\lam)\] in the above security game. We say $\Pi_\NCE$ is query-bounded if the adversary can make at most $q$ total queries, where $q$ is additionally given to $\KG$. 
\end{itemize}
\end{definition}

\begin{theorem}[{NCE from CPA-secure SKE~\cite{C:KNTY19}}]\label{thm:cpa-to-nce}
    If CPA-secure secret-key encryption exists, then there exists secret-key NCE for receiver. 
\end{theorem}

\begin{corollary}[Bounded Query NCE from Bounded Query CPA-secure SKE]
\label{cor:bounded-cpa-to-nce}
    If bounded query CPA-secure secret-key encryption exists, then there exists bounded query secret-key NCE for receiver, where the ciphertext and decryption key sizes are $2\ell\abs{\ct}$ and $\ell\abs{\sk}$, where $\abs{\ct},\abs{\sk}$ are the sizes of the SKE ciphertext and key, and $\ell$ is the message length. 
\end{corollary}
\begin{proof}
    Follows by \cref{thm:cpa-to-nce} and inspection of the construction in~\cite{C:KNTY19}.
\end{proof}

\paragraph{Generalizing to quantum ciphertexts.} We can generalize \cref{def:nce} to support quantum ciphertexts with a small change in the syntax. This will allow us to construct query-bounded NCE from pseudorandom state generators in~\cref{sec:qske}. 

\paragraph{Hyper-efficient shadow tomography.} We finally recall the notion of hyper-efficient shadow tomography (HEST), introduced in~\cite{Shadow2}. Intuitively, shadow tomography is a quantum state learning task where the objective is to learn a how a specific quantum circuit acts on a state given many copies of the state. 

\begin{definition}[{Hyper-Efficient Shadow Tomography~\cite{Shadow2}}]
\label{def:hest} 
Let $E$ denote a uniform quantum circuit family with classical binary output that takes as input $i\in[M]$ and an $n$-qubit quantum state $\rho$. Then, a shadow tomography procedure $\cT$ takes as input $E$ and $k$ copies of $\rho$, and outputs a quantum circuit $C$ such that $\Pr[\abs{C(i)-\Pr[E(i,\rho)=1]}\le\eps]\ge 1-\delta$. The procedure is said to be \textit{hyper-efficient} if the number of copies $k$ and the runtime are both $\poly(n, \log M, 1/\eps,\log 1/\delta)$.
\end{definition}

\subsection{Untelegraphable Encryption with Indistinguishability Security}
\label{sec:ind-ute}
In this section, we construct a one-time secure untelegraphable SKE scheme with indistinguishability security from one-way secure UTSKE and a universal hash function.

\begin{construction}[One-Time UTE]\label{cons:OT-UTE}
    Let $\lam$ be a security parameter. Our construction relies on the following additional primitives:
    \begin{itemize}
        \item Let $\Pi_\OWUTE=(\OWUTE.\Gen,\OWUTE.\Enc,\OWUTE.\Dec)$ be a one-way secure UTSKE scheme with message space $\cM=\setbk{\cM_\lam}_{\lam\in\N}$.
        \item Let $\cH$ be a universal hashing family of size at most $2^{\poly(\lam)}$ where each function $h:\cM_\lam\to\zo{}$ has domain $\cM_{\lam}$ and range $\zo{}$.
    \end{itemize}
    We now construct our UTSKE scheme $\Pi_\UTE=(\Gen,\Enc,\Dec)$ as follows:
    \begin{itemize}
        \item $\Gen(1^\lam):$ On input the security parameter $\lam$, the key generation algorithm samples $\sk_\OW\gets\OWUTE.\Gen(1^\lam),h\getsr\cH,r\getsr\zo{}$ and outputs $\sk=(\sk_\OW,h,r)$.
        \item $\Enc(\sk,m):$ On input the secret key $\sk=(\sk_\OW,h,r)$ and a message $m\in\zo{}$, the encryption algorithm samples $x\getsr\cM_\lam$, computes $\ket{\ct_\OW}\gets\OWUTE.\Enc(\sk_\OW,x)$, and outputs the quantum ciphertext \ifelsesubmiteq{\ketct=(\ket{\ct_\OW},r\xor h(x)\xor m).}
        \item $\Dec(\sk,\ket{\ct}):$ On input the secret key $\sk=(\sk_\OW,h,r)$ and a quantum ciphertext $\ket{\ct}=(\ket{\ct_\OW},r')$, the decryption algorithm outputs \ifelsesubmiteq{m=r'\xor r\xor h(\OWUTE.\Dec(\sk_\OW,\ket{\ct_\OW})).}
    \end{itemize}
\end{construction}

\begin{theorem}[Correctness]\label{thm:ot-ute-correct}
    Suppose $\Pi_\OWUTE$ is correct. Then, \cref{cons:OT-UTE} is correct. 
\end{theorem}
\newcommand{\otutecorrectproof}{
    Fix any $\lam\in\N$ and $m\in\zo{}$. Let $\sk=(\sk_\OW,h,r)\gets\Gen(1^\lam)$ and $\ketct=(\ket{\ct_\OW},r\xor h(x)\xor m)\gets\Enc(\sk,m)$. We consider the output of the decryption algorithm $\Dec(\sk,\ketct)$. 
    By correctness of $\Pi_\OWUTE$, we have $x\gets\OWUTE.\Dec(\sk_\OW,\ket{\ct_\OW})$ with overwhelming probability. 
    It is then immediate that $r\xor h(x)\xor m \xor r \xor h(x)=m$, so $m$ is the output of $\Dec(\sk,\ketct)$ with overwhelming probability. }
\begin{proof}
    \ifelsesubmit{We give the proof in~\cref{sec:ot-ute-correct}.}{\otutecorrectproof}
\end{proof}

\begin{theorem}[Indistinguishability Security]\label{thm:ot-ute-secure}
    Suppose $\Pi_\OWUTE$ is statistically one-way secure, where the maximum advantage of the adversary is $2^{-\Omega(\lam)}$. Then, \cref{cons:OT-UTE} satisfies statistical indistinguishability security.
\end{theorem}
\newcommand{\otutesecureproof}{
    We prove the theorem by showing encryptions of 0 and 1 are statistically indistinguishable. 
    Suppose there exists a two-stage adversary $(\A,\B)$ that has non-negligible advantage $\delta$ in the indistinguishability security game. 
    We define a sequence of hybrid experiments:
    \begin{itemize}
        \item $\hybb{0}$: This is the real indistinguishability game with bit $b\in\zo{}$. In particular:
        \begin{itemize}
            \item At the beginning of the game, the challenger samples \ifelsesubmiteq{\sk_\OW\gets\OWUTE.\Gen(1^\lam),h\getsr\cH,r\getsr\zo{}} and sets $\sk=(\sk_\OW,h,r)$. The challenger then samples $x\getsr\cM_\lam$, computes $\ket{\ct_\OW}\gets\OWUTE.\Enc(\sk_\OW,x)$, and gives the quantum ciphertext $\ketct=(\ket{\ct_\OW},r'=r\xor h(x)\xor b)$ to $\A$.
            \item Adversary $\A$ then outputs a classical string $\st$, which is given to $\B$ along with $\sk$ from the challenger. Adversary $\B$ outputs a bit $b_\B$, which is the output of the experiment. 
        \end{itemize}
        \item $\hybb{1}$: Same as $\hybb{0}$ except the challenger samples the $r'$ component of the challenge ciphertext as $\diff{r'\getsr\zo{}}$, and sets the $r$ component of the secret key as $\diff{r=r'\xor h(x)\xor b}$.
        \item $\hybb{2}$: Same as $\hybb{1}$ except the challenger samples $\diff{u\getsr\zo{}}$ and sets $\diff{r=r'\xor u}$.
    \end{itemize}
    We write $\hybb{i}(\A,\B)$ to denote the output distribution of
    an execution of $\hybb{i}$ with adversary $(\A,\B)$. We now argue that each adjacent pair of distributions are indistinguishable.
    \begin{lemma}\label{lem:OT-UTE0-1}
        The experiments $\hybb{0}(\A,\B)$ and $\hybb{1}(\A,\B)$ are identically distributed.
    \end{lemma}
    \begin{proof}
        Since $r\getsr\zo{}$ in $\hybb{0}(\A,\B)$, the marginal distribution of $r'$ is uniformly random. Moreover, by setting $r=r'\xor h(x)\xor b$, the relation $r'=r\xor h(x)\xor b$ still holds in $\hybb{1}$. Thus, the view of $\A$ and $\B$ is identical in the two experiments.  
    \end{proof}
    \begin{lemma}\label{lem:OT-UTE1-2}
        Suppose $\Pi_\OWUTE$ is statistically one-way secure, where the maximum advantage of the adversary is $2^{-\Omega(\lam)}$. Then, there exists a negligible function $\negl(\cdot)$ such that for all $\lam\in\N$, we have
        \[
        \abs{\Pr[\hybb{1}(\A,\B)=1]-\Pr[\hybb{2}(\A,\B)=1]}=\negl(\lam).
        \]
    \end{lemma}
    \begin{proof}
        Suppose $\AB$ distinguishes $\hybb{1}$ and $\hybb{2}$ with non-negligible probability $\delta'$. 
        We use $\AB$ to contradict \cref{cor:cond-LHL}.  
        Define the joint distribution $\cD_\A$ over $\cX\times\cZ$ as follows:
        \begin{itemize}
            \item $\cX$: Sample $x\getsr\cM_\lam$. Output $x$.
            \item $\cZ$: Sample $\sk_\OW\gets\OWUTE.\Gen(1^\lam),r'\getsr\zo{}$, and compute $\ket{\ct_\OW}\gets\OWUTE.\Enc(\sk_\OW,x)$. Run $\A$ on input $(\ket{\ct_\OW},r')$ to get $\st$. Output $(\sk_\OW,\st,r')$. 
        \end{itemize}
        Since $\Pi_\OWUTE$ is statistically one-way secure, where the maximum advantage of the adversary is $2^{-\Omega(\lam)}$, and $r'$ is sampled independently of $x$, it must be that $\Hinf(\cX\mid\cZ)\geq \Omega(\lam)$ holds. 
        We now construct a distinguisher $\A'$ that contradicts \cref{cor:cond-LHL}:
        \begin{enumerate}
            \item Algorithm $\A'$ gets an LHL challenge $(u,h,z)$ with joint distribution $\cD_\A$ as described above.  
            \item Algorithm $\A'$ parses $z=(\sk_\OW,\st,r')$, sets $r=r'\xor u\xor b$, sets $\sk=(\sk_\OW,h,r)$, runs adversary $\B$ on input $(\st,\sk)$, and outputs whatever $\B$ outputs.
        \end{enumerate}
        If $u\gets h(x)$, then $\hybb{1}\AB$ is perfectly simulated by definition of $\cD_\A$. 
        If $u\getsr\zo{}$, then $\hybb{2}\AB$ is perfectly simulated, since the bit $b$ is fixed. Thus, $\A'$ is able to distinguish the distributions in \cref{cor:cond-LHL} with advantage $\delta'=\omega(2^{-\Omega(\lam)})=\omega(2^{-L/2})$, a contradiction. 
    \end{proof}
    \begin{lemma}\label{lem:OT-UTE-final}
        The experiments $\Hyb{2}{(0)}(\A,\B)$ and $\Hyb{2}{(1)}(\A,\B)$ are identically distributed.
    \end{lemma}
    \begin{proof}
        By construction, the challenger's behavior in $\hyb_{2}$ is {\em independent} of the challenge bit $b \in \zo{}$, so the adversary's view in the two distributions is identical.
    \end{proof}
    \noindent Combining \cref{lem:OT-UTE0-1,lem:OT-UTE1-2,lem:OT-UTE-final}, statistical indistinguishability security follows by a standard hybrid argument. 
    }
\begin{proof}
    \ifelsesubmit{We give the proof in~\cref{sec:ot-ute-secure}.}{\otutesecureproof}
\end{proof}

\begin{remark}[Longer Messages]
    Since the only limitation to message length is the entropy loss $L$ in \cref{cor:cond-LHL}, by setting the security parameter for the one-way scheme to be $\lam^2$, we can achieve indistinguishability security for $\lam$-bit messages, which extends to $\poly(\lam)$-bit messages (assuming one-way functions) via key encapsulation techniques. 
\end{remark}

\subsection{Collusion-Resistant Untelegraphable Encryption}
\label{sec:cr-ute}
In this section, we give a construction of collusion-resistant untelegraphable encryption from one-time UTE with indistinguishability security and SK-NCER. 
Note that due to the definition of SK-NCER (\cref{def:nce}), we inherit the need for separate encryption and decryption keys, where only the decryption key gets revealed to the second stage adversary. 
\begin{construction}\label{cons:cr-ute}
    Let $\lam$ be a security parameter and $\cM=\setbk{\cM_\lam}_{\lam\in\N}$ be a message space. Our construction relies on the following additional primitives:
    \begin{itemize}
        \item Let $\Pi_\OneUTE=(\OneUTE.\Gen,\OneUTE.\Enc,\OneUTE.\Dec)$ be a one-time secure UTSKE scheme with message space $\cM$ and key space $\cK=\setbk{\cK_\lam}_{\lam\in\N}$.
        \item Let $\Pi_\NCE=(\NCE.\KG, \NCE.\Enc, \NCE.\Dec, \NCE.\Fake, \NCE.\Reveal)$ be an SK-NCER scheme with message space $\cK$. 
    \end{itemize}
    We construct our collusion-resistant UTE scheme $\Pi_\UTE=(\Gen,\Enc,\Dec)$ as follows:
    \begin{itemize}
        \item $\Gen(1^\lam):$ On input the security parameter $\lam$, the key generation algorithm outputs a pair of keys $(\ek,\dk)\gets\NCE.\KG(1^\lam)$.
        \item $\Enc(\ek,m):$ On input the encryption key $\ek$ and a message $m\in\cM_\lam$, the encryption algorithm samples $\sk_\OneUTE\gets\OneUTE.\Gen(1^\lam)$, computes $\ket{\ct_\OneUTE}\gets\OneUTE.\Enc(\sk_\OneUTE,m),\ct_\NCE\gets\NCE.\Enc(\ek,\sk_\OneUTE)$ and outputs \ifelsesubmiteq{\ketct=(\ket{\ct_\OneUTE},\ct_\NCE).}
        \item $\Dec(\dk,\ket{\ct}):$ On input the decryption key $\dk$ and a quantum ciphertext $\ket{\ct}=(\ket{\ct_1},\ct_2)$, the decryption algorithm outputs $\OneUTE.\Dec(\NCE.\Dec(\dk,\ct_2),\ket{\ct_1})$.
    \end{itemize}
\end{construction}
\begin{theorem}[Correctness]\label{thm:cr-ute-correct}
    Suppose $\Pi_\OneUTE$ and $\Pi_\NCE$ are correct. Then, \cref{cons:cr-ute} is correct. 
\end{theorem}
\newcommand{\crutecorrectproof}{
    Fix any $\lam\in\N$ and $m\in\cM_\lam$. Let $(\ek,\dk)\gets\Gen(1^\lam),\ketct=(\ket{\ct_\OneUTE},\ct_\NCE)\gets\Enc(\ek,m)$. We consider the output of $\Dec(\dk,\ketct)$. By correctness of $\Pi_\NCE$ and definition of $\Enc$, we have $\sk_\OneUTE\gets\NCE.\Dec(\dk,\ct_\NCE)$ with overwhelming probability. Similarly, by $\Pi_\OneUTE$ correctness and definition of $\Enc$, we have $m\gets\OneUTE.\Dec(\sk_\OneUTE,\ket{\ct_\OneUTE})$ with overwhelming probability, which completes the proof.}
\begin{proof}
    \ifelsesubmit{We give the proof in~\cref{sec:cr-ute-correct}.}{\crutecorrectproof}
\end{proof}

\begin{theorem}[Collusion-Resistant Security]\label{thm:cr-ute-secure}
    Suppose $\Pi_\OneUTE$ satisfies indistinguishability security (\cref{def:UTE}) and $\Pi_\NCE$ satisfies non-committing security for QPT adversaries. Then, \cref{cons:cr-ute} satisfies collusion-resistant security. 
\end{theorem}
\newcommand{\crutesecureproof}{
    Suppose there exists a QPT two-stage adversary $(\A,\B)$ that has non-negligible advantage $\delta$ in the collusion-resistant security game and makes at most $Q$ message queries. 
    We define a sequence of hybrid experiments:
    \begin{itemize}
        \item $\hyb_0$: This is the collusion-resistant game with bit $b=0$. 
        In particular:
        \begin{itemize}
            \item At the beginning of the game, the challenger samples $(\ek,\dk)\gets\NCE.\KG(1^\lam)$. 
            
            \item  When adversary $\A$ makes its $\ord{i}$ query $(m^{(i)}_0,m^{(i)}_1)$, the challenger samples $\sk^{(i)}_\OneUTE\gets\OneUTE.\Gen(1^\lam)$ and computes 
            \[
            \ket{\ct^{(i)}_\OneUTE}\gets \OneUTE.\Enc(\sk^{(i)}_\OneUTE,m^{(i)}_0)
            \eqand \ct^{(i)}_\NCE\gets\NCE.\Enc(\ek,\sk^{(i)}_\OneUTE).
            \] 
            The challenger gives the ciphertext $\ket{\ct^{(i)}}=\parens{\ket{\ct^{(i)}_\OneUTE},\ct^{(i)}_\NCE}$ to $\A$.
            
            \item Adversary $\A$ then outputs a classical string $\st$, which is given to $\B$ along with $\dk$ from the challenger. Adversary $\B$ outputs a bit $b_\B$, which is the output of the experiment. 
        \end{itemize}
        \item $\hyb_i$: Same as $\hyb_{i-1}$ except the challenger now computes  \ifelsesubmiteq{\diff{\ket{\ct^{(i)}_\OneUTE}\gets \OneUTE.\Enc(\sk^{(i)}_\OneUTE,m^{(i)}_1)}.}
    \end{itemize}
    Note that $\hyb_Q$ is the collusion-resistant game with bit $b=1$. 
    We write $\hyb_{i}(\A,\B)$ to denote the output distribution of
    an execution of $\hyb_{i}$ with adversary $(\A,\B)$.
    We now show that for each $i\in[Q]$, experiments $\hyb_{i-1}(\A,\B)$ and $\hyb_i(\A,\B)$ are indistinguishable. 
    \begin{lemma}\label{lem:cr-ute}
        Suppose $\Pi_\OneUTE$ satisfies indistinguishability security and $\Pi_\NCE$ satisfies non-committing security for QPT adversaries. Then, for all $i\in[Q]$, there exists a negligible function $\negl(\cdot)$ such that for all $\lam\in\N$, we have
        \[
        \abs{\Pr[\hyb_{i-1}(\A,\B)=1]-\Pr[\hyb_{i}(\A,\B)=1]}=\negl(\lam).
        \]
    \end{lemma}
    \begin{proof}
        We define the following intermediate hybrid experiments:
        \begin{itemize}
            \item $\hyb'_{i-1}$: Same as $\hyb_{i-1}$ except the challenger now computes \ifelsesubmiteq{\diff{(\ct^{(i)}_\NCE,\st_\NCE)\gets\NCE.\Fake(\ek)}} and gives $\diff{\dkstar\gets\NCE.\Reveal(\st_\NCE,\sk^{(i)}_\OneUTE)}$ to $\B$.    
            \item $\hyb''_{i-1}$: Same as $\hyb'_{i-1}$ except the challenger now computes \ifelsesubmiteq{\diff{\ket{\ct^{(i)}_\OneUTE}\gets \OneUTE.\Enc(\sk^{(i)}_\OneUTE,m^{(i)}_1)}.} 
        \end{itemize}
        We show that each adjacent pair of experiments from the list $(\hyb_{i-1},\hyb'_{i-1},\hyb''_{i-1},\hyb_{i})$ are indistinguishable.
        \begin{claim}\label{claim:cr-ute-nce}
            Suppose $\Pi_\NCE$ satisfies non-committing security for QPT adversaries. Then, for all $i\in[Q]$, there exists a negligible function $\negl(\cdot)$ such that for all $\lam\in\N$, we have
            \[
            \abs{\Pr[\hyb_{i-1}(\A,\B)=1]-\Pr[\hyb'_{i-1}(\A,\B)=1]}=\negl(\lam).
            \]
        \end{claim}
        \begin{proof}
            Suppose $(\A,\B)$ distinguish the experiments with non-negligible probability $\delta'$. 
            We use $(\A,\B)$ to construct an adversary $\A'$ that breaks non-committing security as follows:
            \begin{enumerate}
                \item For each $j\in[Q]$, adversary $\A'$ samples $\sk^{(j)}_\OneUTE\gets\OneUTE.\Gen(1^\lam)$.
                \item When adversary $\A$ makes the $\ord{j}$ query $(m^{(j)}_0,m^{(j)}_1)$, adversary $\A'$ does the following:
                \begin{itemize}
                    \item For all $j\in[Q]$, compute \[\ket{\ct^{(j)}_\OneUTE}\gets \OneUTE.\Enc(\sk^{(j)}_\OneUTE,m^{(j)}_{b_j}),\] where $b_j=0$ for $j\geq i$ and $b_j=1$ for $j<i$. 
                    \item For $j\neq i$, compute $\ct^{(j)}_\NCE$ by making an encryption query on message $\sk^{(j)}_\OneUTE$.
                    \item For $j=i$, make a challenge query on $\sk^{(i)}_\OneUTE$ to get $(\ct^{(i)}_\NCE,\dkstar)$. 
                \end{itemize}
                Adversary $\A'$ replies with 
                $\ket{\ct^{(j)}}=\parens{\ket{\ct^{(j)}_\OneUTE},\ct^{(j)}_\NCE}$.
                
                \item When adversary $\A$ outputs $\st$, adversary $\A'$ gives $(\st,\dkstar)$ to $\B$, and outputs whatever $\B$ outputs.
            \end{enumerate}
            Clearly $\A'$ is QPT since $(\A,\B)$ is. 
            If the ciphertext and decryption key are computed normally, $\A'$ perfectly simulates $\hyb_{i-1}(\A,\B)$. 
            If the ciphertext and decryption key are faked, $\A'$ perfectly simulates $\hyb'_{i-1}(\A,\B)$.
            Thus, $\A'$ has advantage $\delta'$ in the non-committing security game, a contradiction. 
        \end{proof}
        \begin{claim}\label{claim:cr-ute-1ute}
            Suppose $\Pi_\OneUTE$ satisfies indistinguishability security. Then, for all $i\in[Q]$, there exists a negligible function $\negl(\cdot)$ such that for all $\lam\in\N$, we have
            \[
            \abs{\Pr[\hyb'_{i-1}(\A,\B)=1]-\Pr[\hyb''_{i-1}(\A,\B)=1]}=\negl(\lam).
            \]
        \end{claim}
        \begin{proof}
            Suppose $(\A,\B)$ distinguish the experiments with non-negligible probability $\delta'$. 
            We use $(\A,\B)$ to construct an adversary $(\A',\B')$ that breaks indistinguishability security as follows:
            \begin{enumerate}
                \item Adversary $\A'$ samples $(\ek,\dk)\gets\NCE.\KG(1^\lam)$.
                \item When adversary $\A$ makes the $\ord{j}$ query $(m^{(j)}_0,m^{(j)}_1)$, adversary $\A'$ does the following:
                \begin{itemize}
                    \item If $j\neq i$, sample $\sk^{(j)}_\OneUTE\gets\OneUTE.\Gen(1^\lam)$ and compute 
                    \[\ket{\ct^{(j)}_\OneUTE}\gets \OneUTE.\Enc(\sk^{(j)}_\OneUTE,m^{(j)}_{b_j})
                    \eqand \ct^{(j)}_\NCE\gets\NCE.\Enc(\ek,\sk^{(j)}_\OneUTE),\]
                    where $b_j=0$ for $j> i$ and $b_j=1$ for $j<i$.
                \item If $j=i$, submit challenge messages $m^{(i)}_0,m^{(i)}_1$ to the UTE challenger to get $\ket{\ct^{(i)}_\OneUTE}$ and compute $(\ct^{(i)}_\NCE,\st_\NCE)\gets\NCE.\Fake(\ek)$.
                \end{itemize}
                Adversary $\A'$ replies with $\ket{\ct^{(j)}}=\parens{\ket{\ct^{(j)}_\OneUTE},\ct^{(j)}_\NCE}$.
                
                \item When adversary $\A$ outputs $\st$, adversary $\A'$ submits $(\st,\st_\NCE)$ to its challenger to be given to $\B'$. 
                \item On input $(\st,\st_\NCE)$ and $\sk^{(i)}_\OneUTE$, adversary $\B'$ computes $\dkstar\gets\NCE.\Reveal(\st_\NCE,\sk^{(i)}_\OneUTE)$, gives $(\st,\dkstar)$ to $\B$, and outputs whatever $\B$ outputs.
            \end{enumerate}
            Clearly $(\A',\B')$ is QPT since $(\A,\B)$ is. 
            If the UTE challenger encrypts $m^{(i)}_0$, adversary $(\A',\B')$ perfectly simulates $\hyb'_{i-1}(\A,\B)$. 
            If the UTE challenger encrypts $m^{(i)}_1$, adversary $(\A',\B')$ perfectly simulates $\hyb''_{i-1}(\A,\B)$.
            Thus, $(\A',\B')$ has advantage $\delta'$ in the indistinguishability security game, a contradiction. 
        \end{proof}
        \begin{claim}\label{claim:cr-ute-nce2}
            Suppose $\Pi_\NCE$ satisfies non-committing security for QPT adversaries. Then, for all $i\in[Q]$, there exists a negligible function $\negl(\cdot)$ such that for all $\lam\in\N$, we have
            \[
            \abs{\Pr[\hyb''_{i-1}(\A,\B)=1]-\Pr[\hyb_{i}(\A,\B)=1]}=\negl(\lam).
            \]
        \end{claim}
        \begin{proof}
            Analogous to the proof of \cref{claim:cr-ute-nce}.
        \end{proof}
        \noindent Combining \cref{claim:cr-ute-nce,claim:cr-ute-1ute,claim:cr-ute-nce2} proves the lemma by a hybrid argument. 
    \end{proof}
    \noindent The theorem follows by a hybrid argument via \cref{lem:cr-ute}.
    }
\begin{proof}
    \ifelsesubmit{We give the proof in~\cref{sec:cr-ute-secure}.}{\crutesecureproof}
\end{proof}

\begin{remark}[NCE with Quantum Ciphertexts]
    The proof of security is analogous even if the SK-NCER scheme has quantum ciphertexts, so long as the state produced by $\Fake$ is classical. 
\end{remark}

\begin{remark}[Second-Stage Encryption Queries]\label{rem:second-stage-query}
    With separate encryption and decryption keys, an adversary obtaining the decryption key in the second-stage cannot run encryption. Thus, it is natural to consider a notion where the second-stage adversary has access to an encryption oracle. 
    The above proof can be easily generalized to handle this, since SK-NCER allows the adversary to make encryption queries. 
\end{remark}

\begin{remark}[Generalizing to PKE and ABE]
    The above result generalizes to the setting of public-key encryption (PKE) and attribute-based encryption (ABE) by replacing SK-NCER with non-committing PKE and non-committing ABE, which can be constructed from PKE and indistinguishability obfuscation, respectively (see \cite{AC:HMNY21}). 
\end{remark}

\subsection{Impossibility of Hyper-Efficient Shadow Tomography}
\label{sec:hest}
In this section, we show an efficient attack on any collusion-resistant UTE that succeeds with overwhelming probability from hyper-efficient shadow tomography (HEST). In particular, this implies the impossibility of HEST for general mixed states assuming only pseudorandom states. 

\begin{theorem}[HEST Breaks $t$-Copy Secure UTE]\label{thm:hest-attack}
    Suppose there exists an algorithm $\cT$ that satisfies \cref{def:hest} for all mixed quantum states with number of copies $k=k(n,\log M,1/\eps,\log(1/\delta))$ for polynomial $k$. 
    Then, there does not exist $k$-copy secure untelegraphable encryption that also satisfies correctness. 
\end{theorem}
\newcommand{\hestattackproof}{
    Let $\Pi_\UTE=(\Gen,\Enc,\Dec)$ with $\poly(\lam)$ be a candidate correct UTE scheme with message space $\zo{}$ (without loss of generality) and decryption key space $\cK=\setbk{\cK_\lam}_{\lam\in\N}$. 
    Let $n=n(\lam)$ be the maximum size of a ciphertext, $M=\abs{\cK_\lam}$, $\eps<1/2$ be a constant, and $\delta=2^{-\lam}$.
    We define the quantum circuit $E_\lam(i,\rho)$ to output $\Dec(\sk,\ketct)\in\zo{}$ where the input is parsed as $i=\dk$ and $\rho=\ketct$.
    We now construct an adversary $(\A,\B)$ that breaks $k$-copy security of $\Pi_\UTE$:
    \begin{enumerate}
        \item On input $\setbk{\ket{\ct_i}}_{i\in[k]}$, adversary $\A$ runs $\cT(E_\lam,\setbk{\ket{\ct_i}}_{i\in[k]})$ to get a classical description of a quantum circuit $C$, which $\A$ outputs as its classical state $\st$.
        \item On input $\dk$ and $\st=C$, adversary $\B$ outputs 1 if $C(\dk)>1/2$ and 0 otherwise.
    \end{enumerate}
    By definition of UTE, $n=\poly(\lam)$ and $\log M=\poly(\lam)$. By the choice of $\delta$, $\log(1/\delta)=\lam$ and $1/\eps=3$. Thus, $\cT$ runs in QPT, which implies $C$ is polynomial size, and thus $(\A,\B)$ runs in QPT. Since $\Pr[\abs{C(i)-\Pr[E(i,\rho)=1]}\le\eps]\ge 1-\delta$, by the setting of $\eps,\delta$, we have \[\Pr[\abs{C(i)-\Pr[E(i,\rho)=1]}<1/2]\ge 1-2^{-\lam}.\]
    By correctness of UTE, we have $\Pr[E(i,\rho)=1]=\negl(\lam)$ when $\setbk{\ket{\ct_i}}_{i\in[k]}$ are encryptions of 0 and $\Pr[E(i,\rho)=1]\ge1-\negl(\lam)$ when $\setbk{\ket{\ct_i}}_{i\in[k]}$ are encryptions of 1. This implies $C(\dk)>1/2$ for encryptions of 1 and $C(\dk)<1/2$ for encryptions of 0 with probability $1-2^{-\lam}$. Thus, $(\A,\B)$ wins the $k$-copy security game with overwhelming probability. }
\begin{proof}
    \ifelsesubmit{We give the proof in~\cref{sec:hest-attack}.}{\hestattackproof}
\end{proof}

\begin{corollary}[Impossibility of HEST from SK-NCER]
\label{cor:nce-hest}
    Assuming the existence of SK-NCER (\cref{def:nce}), hyper-efficient shadow tomography for all mixed quantum states does not exist. 
\end{corollary}
\begin{proof}
    Follows immediately by \cref{thm:cr-ute-correct,thm:cr-ute-secure,thm:hest-attack} and the fact that collusion-resistant security implies $t$-copy security (see \cref{def:cr-sec}).
\end{proof}

\begin{corollary}[Impossiblity of HEST from OWFs]\label{cor:owf-hest}
    Assuming the existence of one-way functions, hyper-efficient shadow tomography for all mixed quantum states does not exist.
\end{corollary}
\begin{proof}
    Follows immediately by \cref{thm:cpa-to-nce,cor:nce-hest}. 
\end{proof}

\subsection{Impossibility of HEST from Pseudorandom State Generators}
\label{sec:qske}
In this section, we show how to weaken the assumption needed to assert the impossibility of hyper-efficient shadow tomography to just pseudorandom state generators, which are believed to be possibly weaker than one-way functions~\cite{Kre21}. 
Specifically, we show that the attack in \cref{thm:hest-attack} will still work with similar parameters on a bounded collusion-resistant UTE scheme, so long as the sizes of the ciphertext and decryption key is \textit{polylogarithmic} in the collusion bound $q$. 
To achieve a bounded collusion-resistant UTE scheme with the required succinctness property, we first construct query-bounded CPA-secure SKE with quantum ciphertexts (QSKE) from a PRFS generator with a polynomial sized domain (which is implied by a PRS generator). 
Then, as a corollary, we get query-bounded SK-NCER with quantum ciphertexts and bounded collusion-resistant UTE from PRS generators, where the succinctness properties will be retained. 

\paragraph{Query-Bounded CPA-Secure QSKE.} We start by formally constructing a query-bounded CPA-secure secret-key encryption with quantum ciphertexts from pseudo-random function-like state (PRFS) generators (\cref{def:prfs}). We show correctness, security, and analyze the succinctness of the ciphertext and decryption key. 

\begin{construction}[Query-Bounded CPA-Secure QSKE]\label{cons:qske}
    Let $\lam$ be a security parameter and $q$ be a query bound. 
    Let $G=G_{\lam,q}$ be a PRFS generator with input length $d=d(\lam,q)$ and output length $n=n(\lam)$. 
    We construct our QSKE scheme $\Pi_{\SKE}=(\Gen,\Enc,\Dec)$ as follows:
    \begin{itemize}
        \item $\Gen(1^\lam,1^q)$: On input the security parameter $\lam$, the key generation algorithm samples $k_1,\dots,k_\lam\getsr\zo{\lam}$ and outputs $\sk=(k_1,\dots,k_\lam)$.
        
        \item $\Enc(\sk,m)$: On input the secret key $\sk=(k_1,\dots,k_\lam)$ and a message $m\in\zo{}$, the encryption algorithm samples $s_1,\cdots,s_\lam$ uniformly from $\zo{}$ such that $s_1\xor\cdots\xor s_\lam=m$ and samples $r_1,\dots,r_\lam\getsr\zo{d-1}$. Compute $\rho_i\gets G(k_i,(r_i,s_i))$ for $i\in[\lam]$ and output $\ketct=(\setbk{r_i,\rho_i}_{i\in[\lam]})$. 
        
        \item $\Dec(\sk, \ketct)$: On input the secret key $\sk=(k_1,\dots,k_\lam)$ and a quantum ciphertext $\ketct=(\setbk{r_i,\rho_i}_{i\in[\lam]})$, the decryption algorithm runs $\test(k_i,(r_i,0),\rho_i)$ to get $s_i$ for $i\in[\lam]$. It outputs $m=s_1\xor\cdots\xor s_\lam$.
    \end{itemize}
\end{construction}

\begin{theorem}[Correctness]\label{thm:qske-correct}
    If $n=\omega(\log\lam)$, then \cref{cons:qske} is correct. 
\end{theorem}
\newcommand{\qskecorrectproof}{
    Fix $\lam\in\N$ and $m\in\zo{}$. 
    Let $\sk=(k_1,\dots,k_\lam)\gets\Gen(1^\lam,1^q)$ and \[\ketct=(\setbk{r_i,\rho_i}_{i\in[\lam]})\gets\Enc(\sk,m).\] 
    We consider the output of $\Dec(\sk, \ketct)$. 
    By definition of $\test$ (see \cref{def:prfs}), the output of $\test(k_i,(r_i,0),\rho_i)$ determines if $s_i=0$ up to negligible error, meaning the derived $s_i$ is correct with overwhelming probability. 
    By the distribution of $s_1,\cdots,s_\lam$, we have $m=s_1\xor\cdots\xor s_\lam$ if the $s_i$ are correct.  
    By a union bound, this holds with overwhelming probability, as desired. }
\begin{proof}
    \ifelsesubmit{We give the proof in~\cref{sec:qske-correct}.}{\qskecorrectproof}
\end{proof}

\begin{theorem}[Query-Bounded CPA Security]\label{thm:qske-secure}
    Let $\lam$ be a security parameter and $q=q(\lam)$ be any polynomial. 
    Suppose $d\geq 2\log q + 1$ and $G$ is a selectively secure PRFS generator.
    Then, \cref{cons:qske} is CPA-secure for up to $q$ queries. 
\end{theorem}
\newcommand{\qskesecureproof}{
    We prove the theorem by showing encryptions of 0 and 1 are computationally indistinguishable. 
    Suppose there exists a QPT adversary $\A$ that has non-negligble advantage $\delta$ in the query-bounded CPA security game. We define a sequence of hybrid experiments:
    \begin{itemize}
        \item $\hybb{0}$: This is the CPA security game with challenge bit $b$. In particular:
        \begin{itemize}
            \item At the beginning of the game, the challenger samples $k_1,\dots,k_\lam\getsr\zo{\lam}$.
            \item When adversary $\A$ makes the $\ord{j}$ encryption query $m^{(j)}_0,m^{(j)}_1\in\zo{}$ for some $j\in[q]$, the challenger samples $s^{(j)}_{1},\dots,s^{(j)}_{\lam}$ uniformly from $\zo{}$ such that $s^{(j)}_1\xor\cdots\xor s^{(j)}_\lam=m^{(j)}_b$ and samples $r^{(j)}_1,\dots,r^{(j)}_\lam\getsr\zo{d-1}$. 
            The challenger computes $\rho^{(j)}_i\gets G(k_i,(r^{(j)}_i,s^{(j)}_i))$ for $i\in[\lam]$ and gives $\ket{\ct^{(j)}}=(\setbk{r^{(j)}_i,\rho^{(j)}_i}_{i\in[\lam]})$ to $\A$. 
            \item After $q$ queries, adversary $\A$ outputs a bit $b'\in\zo{}$, which is the output of the experiment.
        \end{itemize}
        
        \item $\hybb{1}$: Same as $\hybb{0}$ except the challenger now samples $r^{(j)}_i\getsr\zo{d-1}$ for $i\in[\lam],j\in[q]$ at the start of the game and \diff{outputs $\bot$ if for all $i\in[\lam]$, there exists $j,j'\in[q]$ such that $j\neq j'$ and $r^{(j)}_i=r^{(j')}_i$}. 
        
        \item $\hybb{2}$: Same as $\hybb{1}$ except the challenger samples $\diff{\rho^{(j)}_{\istar}\gets\mu_n}$ for all $j\in[q]$, where $\istar\in[\lam]$ is the first index such that $r^{(1)}_{\istar},\dots, r^{(q)}_{\istar}$ are distinct. 
    \end{itemize}
    We write $\hybb{i}(\A)$ to denote the output distribution of
    an execution of $\hybb{i}$ with adversary $\A$. We now argue that each adjacent pair of distributions are indistinguishable.
    \begin{lemma}\label{lem:qske0-1}
        Suppose $d\geq 2\log q + 1$. Then, there exists a negligible function $\negl(\cdot)$ such that for all $\lam\in\N$, we have
        \[
        \abs{\Pr[\hybb{0}(\A)=1]-\Pr[\hybb{1}(\A)=1]}=\negl(\lam).
        \]
    \end{lemma}
    \begin{proof}
        First, note that sampling $r^{(j)}_i$ at the start is purely syntactic, since $r^{(j)}_i$ is independently sampled from each $m^{(j)}_b$.
        We show the lemma by showing that the probability that the challenger outputs $\bot$ at the start of the game is negligible, as this is the only other difference between the experiments. 
        Since $d\geq 2\log q + 1$, the set $\zo{d-1}$ has size $2^{d-1}\geq 2^{2\log q}=q^2$. This implies that for each index $i\in[\lam]$, the probability that there exists $j,j'\in[q]$ such that $j\neq j'$ and $r^{(j)}_i=r^{(j')}_i$ is at most 
        \[
        1-\parens{1-\frac{q-1}{q^2}}^q \leq 1-1/e.
        \]
        Since $r^{(1)}_{i},\dots, r^{(q)}_{i}$ are sampled independently for each $i\in[\lam]$, the probability of a collision for all $i\in[\lam]$ is at most $(1-1/e)^\lam$, which is negligible.
    \end{proof}

    \begin{lemma}\label{lem:qske1-2}
        Suppose $G$ is a selectively secure PRFS generator. Then, there exists a negligible function $\negl(\cdot)$ such that for all $\lam\in\N$, we have
        \[
        \abs{\Pr[\hybb{1}(\A)=1]-\Pr[\hybb{2}(\A)=1]}=\negl(\lam).
        \]
    \end{lemma}
    \begin{proof}
        Suppose $\A$ distinguishes the experiments with non-negligible probability $\delta'$. We use $\A$ to construct an adversary $\B$ that breaks PRFS security as follows:
        \begin{enumerate}
            \item Adversary $\B$ starts by sampling $r^{(j)}_i\getsr\zo{d-1}$ for $i\in[\lam],j\in[q]$ and checking the abort condition. 
            If it does not abort, it finds the first index $\istar\in[\lam]$ such that $r^{(1)}_{\istar},\dots, r^{(q)}_{\istar}$ are distinct.

            \item Adversary $\B$ samples $s^{(j)}_{\istar}\getsr\zo{}$ for $j\in[q]$ and submits the queries $\setbk{(r^{(j)}_{\istar},s^{(j)}_{\istar})}_{j\in[q]}$ to the PRFS challenger to get back $\rho^{(j)}_{\istar}$ for $j\in[q]$. 

            \item Adversary $\B$ samples $k_i\getsr\zo{\lam}$ for $i\neq\istar$. When advervsary $\A$ makes its $\ord{j}$ encryption query $m^{(j)}_0,m^{(j)}_1\in\zo{}$, the $\B$ samples $s^{(j)}_{1},\dots,s^{(j)}_{\istar-1},s^{(j)}_{\istar+1},\dots,s^{(j)}_{\lam}$ uniformly from $\zo{}$ such that \[s^{(j)}_1\xor\cdots s^{(j)}_{\istar-1}\xor s^{(j)}_{\istar+1}\cdots\xor s^{(j)}_\lam=m^{(j)}_b\xor s^{(j)}_{\istar}.\]
            Adversary $\B$ computes $\rho^{(j)}_i\gets G(k_i,(r^{(j)}_i,s^{(j)}_i))$ for $i\in[\lam]\setminus\setbk{\istar}$ and gives $\ket{\ct^{(j)}}=(\setbk{r^{(j)}_i,\rho^{(j)}_i}_{i\in[\lam]})$ to $\A$.

            \item At the end of the game, adversary $\B$ outputs whatever $\A$ outputs. 
        \end{enumerate}
        Clearly adversary $\B$ is QPT if $\A$ is. 
        If $\rho^{(j)}_{\istar}$ are Haar-random states for $j\in[q]$, adversary $\B$ simulates $\hybb{2}(\A)$. 
        If $\rho^{(j)}_{\istar}$ are PRFS generator outputs for $j\in[q]$, adversary $\B$ simulates $\hybb{1}(\A)$.
        Thus, adversary $\B$ has advantage $\delta'$ in the PRFS security game, a contradiction. 
    \end{proof}

    \begin{lemma}\label{lem:qske-final}
        The experiments $\Hyb{2}{(0)}(\A)$ and $\Hyb{2}{(1)}(\A)$ are identically distributed.
    \end{lemma}
    \begin{proof}
        Since for all $j\in[q]$ the state $\rho^{(j)}_{\istar}$ no longer depends on $s^{(j)}_{\istar}$, each $s^{(j)}_{i}$ for $i\neq \istar$ is uniform and independent of $b$, meaning the adversary's view in the two experiments is identical. 
    \end{proof}
    \noindent Combining \cref{lem:qske0-1,lem:qske1-2,lem:qske-final}, the theorem follows by a hybrid argument. 
    }
\begin{proof}
    \ifelsesubmit{We give the proof in~\cref{sec:qske-secure}.}{\qskesecureproof}
\end{proof}

\begin{theorem}[Succinct Ciphertexts]\label{thm:qske-succinct}
    Let $\lam\in\N$ be a security parameter and $q\in\N$ be a query bound.
    Then, the size of a ciphertext in \cref{cons:qske} is $\poly(\lam,\log q)$ and the size of the secret key is $\poly(\lam)$. 
\end{theorem}
\begin{proof}
    For \cref{thm:qske-correct}, the parameter requirement is $n=\omega(\log\lam)$ and for \cref{thm:qske-secure} the parameter requirement is $d\geq 2\log q + 1$.
    Thus, setting $n=\lam$ and $d=2\log q + 1$ suffices. 
    A ciphertext consists of $\lam$ pairs of $r\in\zo{d-1}$ and $\rho$ which is the output of the PRFS. Thus, the total number of qubits in a ciphertext is $\lam(2\log q+\lam)=2\lam\log q+\lam^2$, as desired.
    Moreover, the key size of a PRFS generator can be made independent of the domain size, so the size of a secret key is $\poly(\lam)$. 
\end{proof}

\begin{corollary}[Succinct Bounded-Query SK-NCER]\label{cor:bounded-sk-ncer}
    Let $\lam\in\N$ be a security parameter and $q(\lam)$ be any polynomial query bound.
    If PRS generators exist, then there is a $q$-bounded-query SK-NCER scheme with $\poly(\lam,\log q)$-sized quantum ciphertexts and $\poly(\lam)$-sized decryption keys. 
\end{corollary}
\begin{proof}
    Follows by \cref{thm:qske-correct,thm:qske-secure,thm:qske-succinct,cor:bounded-cpa-to-nce}.
\end{proof}

\begin{corollary}[Succinct Bounded Collusion-Resistant UTE]\label{cor:bounded-cr-ute}
    Let $\lam\in\N$ be a security parameter and $q(\lam)$ be any polynomial query bound.
    If PRS generators exist, then there is a $q$-bounded collusion-resistant UTE scheme with $\poly(\lam,\log q)$-sized ciphertexts and $\poly(\lam)$-sized decryption keys.  
\end{corollary}
\begin{proof}
    Follows by considering the analog of \cref{cons:cr-ute} and the proofs of \cref{thm:cr-ute-correct,thm:cr-ute-secure} in the bounded collusion-resistant setting, along with \cref{cor:bounded-sk-ncer}. 
    The size of a ciphertext is the sum of the sizes of a one-time UTE ciphertext and the SK-NCER ciphertext. 
    The size of the one-time UTE ciphertext is $\poly(\lam)$ (since it does not have a bound parameter), and by \cref{cor:bounded-sk-ncer} the size of the SK-NCER ciphertext is $\poly(\lam,\log q)$.
    The size of a decryption key is the same as the underlying SK-NCER scheme, which is $\poly(\lam)$. 
\end{proof}

\paragraph{Impossibility of HEST.} We now show an attack on a succinct bounded $q$-copy secure UTE from hyper-efficient shadow tomography with number of copies $k$. 

\begin{theorem}[HEST Breaks Bounded $q$-Copy Secure UTE]\label{thm:bounded-hest-attack}
    Let $\lam$ be a security parameter.
    Suppose there exists an algorithm $\cT$ that satisfies \cref{def:hest} for all mixed quantum states with number of copies that is upperbounded by $k=k(n,\log M,1/\eps,\log(1/\delta))$ for polynomial $k$. 
    Then, there exists $q=\poly(\lam)$ such that bounded $q$-copy secure untelegraphable encryption with succinct ciphertexts and decryption keys does not exist.
\end{theorem}
\begin{proof}
    The attack is analogous to the one in~\cref{thm:hest-attack}. 
    Thus, we just show that that parameters can be satisfied. Recall that $n$ corresponds to the number of qubits in a ciphertext. 
    If the ciphertexts or keys are not succinct, there is a circularity issue: namely, $k$ grows with $\poly(n,\log M)=\poly(q)$, so $q$ may always be less than $k$. 
    However, when ciphertexts and keys are succinct, we can set the parameters such that this circularity is not an issue. 
    Specifically, fixing any $\lam\in\N$ and letting $k_0=\poly(\lam,\log q)$ be the value of $k$ when $q=\lam$, we set $q=\lam k_0^2$, which will correspondingly bound $k$ by $k_0 \cdot \polylog(k_0)$. 
    This means the UTE scheme can support $k$-copy security when $q\geq\lam k_0^2$, which we know is $\poly(\lam)$ by definition of $k_0$. 
\end{proof}

\begin{corollary}[Impossiblity of HEST from PRS Generators]\label{cor:prsg-hest}
    Assuming the existence of pseudorandom state generators, hyper-efficient shadow tomography for all mixed quantum states does not exist.
\end{corollary}
\begin{proof}
    The corollary follows by \cref{thm:bounded-hest-attack,cor:bounded-cr-ute}.
\end{proof}

\ifnotsubmit{
\section{Separating Untelegraphable Encryption and Unclonable Encryption}
In this section, we construct an untelegraphable encryption scheme that satisfies UTE security, but simultaneously has an adversary that wins the UE security game with non-negligible advantage for an \textit{unbounded} polynomial number of second-stage adversaries. 
The construction is secure under the LWE assumption in the classical oracle model. 
In the classical oracle model, we assume algorithms can generate oracles to classical functions efficiently so long as the function has a polynomial-size description. 
This can be thought of as an extension of ideal obfuscation.   
Note that for any fixed bound $n$ on the number of second-stage adversaries, one can construct $n$-clonable UTE trivially from collusion-resistant UTE, by defining encryption to generate $n$ fresh UTE ciphertexts.  
\subsection{Building Blocks}
In this section, we describe notions that are needed for our construction of unbounded clonable UTE.
\paragraph{One-shot MAC.} We introduce the notion of a one-shot message authentication code (OSMAC). 
An OSMAC functions similar to a standard MAC, except users now generate a one-time quantum signing key, along with a classical verification key. The security notion states that it is computationally hard for an adversary to generate two signatures of two different messages under any single verification key that is generated by the adversary. 
One-shot MACs are a relaxation of one-shot signatures~\cite{AGKZ20}, which are publicly verifiable, and two-tier one-shot signatures~\cite{MPY24}, which have partial public verification. In \cite{MPY24}, two-tier one-shot signatures are constructed from the standard learning with errors (LWE) assumption~\cite{Reg05}. 
\begin{definition}[One-shot MAC]
    A one-shot message authentication code (OSMAC) scheme with message space $\cM=\setbk{\cM_\lam}_{\lam\in\N}$ is a tuple of QPT algorithms $\Pi_{\OSMAC}=(\Setup,\KeyGen,\Sign,\Vrfy)$ with the following syntax:
    \begin{itemize}
        \item $\Setup(1^\lam)\to(\pp,\mvk)$: On input the security parameter $\lam$, the setup algorithm outputs the public parameters $\pp$ and a master verification key $\mvk$.
        \item $\KeyGen(\pp)\to(\ket{\sk},\vk)$: On input the public parameters $\pp$, the key generation algorithm outputs a quantum signing key $\ket{\sk}$ along with a verification key $\vk$.
        \item $\Sign(\ket{\sk},m)\to\sigma$: On input a quantum signing key $\ket{\sk}$ and a message $m\in\cM_\lam$, the signing algorithm outputs a signature $\sigma$.
        \item $\Vrfy(\pp,\mvk,\vk,m,\sigma)\to b$: On input the public parameters $\pp$, the master verification key $\mvk$, a verification key $\vk$, a message $m\in\cM_\lam$, and a signature $\sigma$, the verification algorithm outputs a bit $b\in\zo{}$. 
    \end{itemize}
    We require that $\Pi_{\OSMAC}$ satisfy the following properties:
    \begin{itemize}
        \item \textbf{Correctness:} For all $\lam\in\N$ and $m\in\cM_\lam$, 
        \[
        \Pr\left[\Vrfy(\pp,\mvk,\vk,m,\sigma)=1:
        \begin{array}{c}
        (\pp,\mvk)\gets\Setup(1^\lam)\\
        (\ket{\sk},\vk)\gets\KeyGen(\pp)\\
        \sigma\gets\Sign(\ket{\sk},m)
        \end{array}\right]\geq 1-\negl(\lam).
        \]
        \item \textbf{Two-Signature Security:} For all QPT adversaries $\A$, there exists a negligible function $\negl(\cdot)$ such that for all $\lam\in\N$, 
        \[
        \Pr\left[
        \begin{array}{c}
        m_0\neq m_1 \ \land
        \Vrfy(\pp,\mvk,\vk,m_0,\sigma_0)=1\\
        \land \ \Vrfy(\pp,\mvk,\vk,m_1,\sigma_1)=1 
        \end{array}
        :
        \begin{array}{c}
        (\pp,\mvk)\gets\Setup(1^\lam)\\
        (\vk,m_0,\sigma_0,m_1,\sigma_1)\gets\A(\pp)
        \end{array}\right]\leq\negl(\lam).
        \]
    \end{itemize}
\end{definition}

\begin{theorem}[{Two-tier one-shot signature~\cite[Theorem 3.2]{MPY24}}] \label{thm:two-tier-OSS}
    Assuming the quantum hardness of LWE, there exists a two-tier one-shot signature scheme with message space $\zo{}$. 
\end{theorem}

\begin{theorem}[One-shot MAC]
    Assuming the quantum hardness of LWE, there exists a one-shot MAC scheme with message space $\zo{\poly(\lam)}$ for any fixed polynomial $\poly(\cdot)$. 
\end{theorem}
\begin{proof}
    A scheme with message space $\zo{}$ follows immediately from \cref{thm:two-tier-OSS}, and extending to multi-bit messages follows by running many single-bit OSMAC schemes in parallel, where each scheme corresponds to a message bit. 
\end{proof}

\subsection{Constructing Unbounded Clonable Untelegraphable Encryption}
\label{sec:clonable-ute}
In this section, we describe our construction of an unbounded clonable UTE scheme assuming the existence of an OSMAC in the classical oracle model. In the classical oracle model, we will assume algorithms can efficiently generate functions as black-boxes, meaning an adversary will only have oracle access to said function in the security proof. Our untelegraphable encryption construction is again in the setting where there are separate encryption and decryption keys, where only the decryption key is revealed to the second stage adversary.

\begin{construction}[Unbounded Clonable UTE]\label{cons:clonable-ute}
    Let $\lam$ be a security parameter, $n=n(\lam)$ be length parameter, and $\cM=\setbk{\cM_\lam}_{\lam\in\N}$ be a message space.
    Let $\Pi_\OSMAC=(\OSMAC.\Setup,\OSMAC.\KeyGen,\OSMAC.\Sign,\OSMAC.\Vrfy)$ be an OSMAC scheme with message space $\zo{\ell(\lam)}$. We construct our unbounded clonable UTE scheme $\Pi_\UTE=(\Gen,\Enc,\Dec)$ as follows:
    \begin{itemize}
        \item $\Gen(1^\lam)$: On input the security parameter $\lam$, the key generation algorithm samples $s\getsr\zo{n+\lam}$ and $(\pp,\mvk)\gets\OSM.\Setup(1^\lam)$. It outputs $(\ek=(\pp,\mvk,s),\dk=s)$.
        \item $\Enc(\ek,m)$: On input the encryption key $\ek=(\pp,\mvk,s)$ and a message $m\in\cM_\lam$, the encryption algorithm samples $(\ket{\sk_\eps},\vk_\eps)\gets\OSM.\KeyGen(\pp)$ and generates the following classical oracle $\cO_{\ek,m,\vk_\eps}$:
        
        \begin{minipage}{\linewidth}
        \begin{framed}
        \textbf{Hard-wired:} master verification key $\mvk$, public parameters $\pp$, string $s\in\zo{n+\lam}$, message $m\in\cM_\lam$, verification key $\vk_\eps$ \\
        \textbf{Input:} strings $s'\in\zo{n+\lam},r\in\zo{n}$, set of signature tuples $\setbk{\st_v}_{v\in V_r}$ where the set $V_r=\setbk{\eps,r_1,r[1,2],\dots,r[1,n-1],r}$ contains all prefixes of $r$ including the empty string $\eps$
        \vspace{0.2cm}

        \begin{enumerate}
            \item Parse $\st_\eps=(\vk_0,\vk_1,\sigma_\eps)$, $\st_v=(\vk_v,\vk_{v\|0},\vk_{v\|1},\sigma_v)$ for $v\in V_r\setminus\setbk{\eps,r}$, $\st_r=(\vk_r,\sigma_r)$.
            \item Output $m$ if $s=s'$, $\OSM.\Vrfy(\pp,\mvk,\vk_v,\vk_{v\|0}\|\vk_{v\|1},\sigma_v)=1$ for all $v\in V_r\setminus\setbk{r}$, and $\OSM.\Vrfy(\pp,\mvk,\vk_r,s',\sigma_r)=1$. Otherwise, output $\bot$.
        \end{enumerate}
        \end{framed}
        \captionof{figure}{Classical oracle $\cO_{\ek,m,\vk_\eps}$}
        \label{fig:mac-chain-oracle}
        \end{minipage}
        
        The encryption algorithm outputs the ciphertext $\ketct=(\ket{\sk_\eps},\cO_{\ek,m,\vk_\eps},\pp)$. 
        \item $\Dec(\dk,\ketct)$: On input the decryption key $\dk$ and a quantum ciphertext $\ketct=(\ket{\sk_\eps},\cO,\pp)$, the decryption algorithm does the following:
        \begin{enumerate}
            \item For all $v\in V_{0^n}\setminus\setbk{\eps}$ sample
            \[
            (\ket{\sk_v},\vk_v)\gets\OSM.\KeyGen(\pp)\text{ and }(\ket{\sk_{v\oplus1}},\vk_{v\oplus1})\gets\OSM.\KeyGen(\pp),
            \] 
            where $v\oplus1$ is $v$ with its last bit flipped for a variable length bitstring $v$.
            \item Compute the signatures $\sigma_v\gets\OSM.\Sign(\ket{\sk_v},\vk_{v\|0}\|\vk_{v\|1})$ for all $v\in V_{0^n}\setminus\setbk{0^n}$ and $\sigma_{0^n}\gets\OSM.\Sign(\ket{\sk_{0^n}},\dk)$.
            \item Set $\st_\eps=(\vk_0,\vk_1,\sigma_\eps)$, $\st_v=(\vk_v,\vk_{v\|0},\vk_{v\|1},\sigma_v)$ for $v\in V_{0^n}\setminus\setbk{\eps,0^n}$, $\st_{0^n}=(\vk_{0^n},\sigma_{0^n})$, and output $\cO(\dk,0^n,\setbk{\st_v}_{v\in V_{0^n}})$.
        \end{enumerate}
    \end{itemize}
\end{construction}

\begin{theorem}[Correctness]\label{thm:clonable-ute-correct}
    Suppose $\Pi_\OSM$ is correct. Then, \cref{cons:clonable-ute} is correct. 
\end{theorem}
\begin{proof}
    Fix any $\lam\in\N$ and $m\in\cM_\lam$. Let $(\ek=(\pp,\mvk,s),\dk=s)\gets\Gen(1^\lam)$ and $\ketct=(\ket{\sk_\eps},\cO_{\ek,m,\vk_\eps},\pp)\gets\Enc(\ek,m)$. We consider the output of $\Dec(\dk,\ketct)$. By inspection, $\Dec$ honestly signs the same MAC chain checked by $\cO_{\ek,m,\vk_\eps}$ (\cref{fig:mac-chain-oracle}) for $r=0^n$. Thus, by correctness of $\Pi_\OSM$, we have $m\gets\cO_{\ek,m,\vk_\eps}(\dk,0^n,\setbk{\st_v}_{v\in V_{0^n}})$ with overwhelming probability, as desired. 
\end{proof}

\begin{theorem}[Indistinguishability Security]\label{thm:clonable-ute-secure}
    Let $\lam\in\N$ be a security parameter and suppose $n=n(\lam)=\Omega(\log\lam)$, $n=o(\lam)$, and $\Pi_\OSM$ satisfies two-signature security for adversaries that run in time $2^n\cdot\poly(\lam)$ and have advantage at most $2^{-n^2}$. Then, \cref{cons:clonable-ute} satisfies one-time indistinguishability security in the classical oracle model. 
\end{theorem}
\begin{proof}
    Suppose there exists a QPT adversary $(\A,\B)$ such that $\A$ and $\B$ make exactly $q$ queries (without loss of generality) and wins the UTE indistinguishability security game with non-negligible advantage $\delta>0$. 
    Let $\cO_\bot$ be an oracle which always outputs $\bot$.
    Now, consider an adversary $\Ahat$ which is given an oracle $\cO$ that is either $\cO_\bot$ or $\cO_{\ek,m_b,\vk_\eps}$ (\cref{fig:mac-chain-oracle}) for challenge bit $b\in\zo{}$, along with an input $(\mvk,s,m_b,\vk_\eps)$.
    Adversary $\Ahat$ runs the UTE security game except that it uses $\cO$ to simulate $\A$'s oracle queries and the secret information $(\mvk,s,m_b,\vk_\eps)$ to simulate $\B$'s oracle queries. 
    Note that $\Ahat^{\cO_{\ek,m_b,\vk_\eps}}(\mvk,s,m_b,\vk_\eps)$ perfectly simulates the UTE game with challenge bit $b$. 
    Setting $\B_{\mathsf{O2H}}$ to the an algorithm that measures a uniformly random query of $\Ahat$, by \cref{lem:O2H}, we have 
    \begin{align*}
    &\abs{\Pr[\Ahat^{\cO_{\ek,m_b,\vk_\eps}}(\mvk,s,m_b,\vk_\eps)=1]-\Pr[\Ahat^{\cO_\bot}(\mvk,s,m_b,\vk_\eps)=1]} \\
    &\leq 2q\cdot\sqrt{\Pr[\B_{\mathsf{O2H}}^{\cO_\bot}(\mvk,s,m_b,\vk_\eps)\in S]}=\nu(\lam),
    \end{align*}
    where $S$ is the set of valid tuples $(s,r,\setbk{\st_v}_{v\in V_r})$ that makes $\cO_{\ek,m_b,\vk_\eps}$ output $m_b$. 
    Since $s$ is uniform and independent of $\A$'s view, it must be that $\nu$ is negligible, since $q=\poly(\lam)$. By definition of $\Ahat$, this implies 
    \[
    \Pr_{b,\pp,\mvk,s,\vk_\eps,\ket{\sk_\eps}}
    \left[
    \B^{\cO_{\ek,m_b,\vk_\eps}}(\st,s)=b:\st\gets\A^{\cO_\bot}(\pp,\ket{\sk_\eps})
    \right]
    \ge \frac{1}{2}+\delta-\nu(\lam),
    \]
    where $b\getsr\zo{},s\getsr\zo{n+\lam},(\pp,\mvk)\gets\OSM.\Setup(1^\lam),(\ket{\sk_\eps},\vk_\eps)\gets\OSM.\KeyGen(\pp)$, as in the UTE security game with a random challenge bit $b$. We can rewrite the above expression as  
    \begin{align}\label{eq:avging}
        \Exp_{\pp,\mvk,\vk_\eps,\st}\left[
        \Pr_{b,s}
        \left[
        \B^{\cO_{\ek,m_b,\vk_\epsilon}}(\st,s)=b
        \right]
        \right] \ge \frac{1}{2}+\delta-\nu(\lam),
    \end{align}
    where $\st\gets\A^{\cO_\bot}(\pp,\ket{\sk_\eps})$.
    Let a tuple $(\pp,\mvk,\vk_\eps,\st)$ be ``good'' if we have 
    \begin{align*}
    \Pr_{b,s}
    \left[
    \B^{\cO_{\ek,m_b,\vk_\eps}}(\st,s)=b
    \right]
    \ge \frac{1}{2}+\frac{\delta}{2}-\nu(\lam).
    \end{align*}
    By \cref{eq:avging}, a tuple $(\pp,\mvk,\vk_\eps,\st)$ is good with probability at least $\delta/2$. Furthermore, for any $(\pp,\mvk,\vk_\eps,\st)$, we have 
    \begin{align}\label{eq:adv-oraclebot}
    \Pr_{b,s}
    \left[
    \B^{\cO_{\bot}}(\st,s)=b
    \right]
    = \frac{1}{2}.
    \end{align}
    Fix a good tuple $(\pp,\mvk,\vk_\eps,\st)$, let $S$ be the same set as above, and let $\B_\extr$ be an algorithm that measures a uniform query of $\B$. Then, by \cref{lem:O2H}, the definition of a good tuple, and \cref{eq:adv-oraclebot}, we have
    \begin{align}\label{eq:adv-extractor}
    \Pr_{b,s}\left[
    \B_\extr^{\cO_\bot}(\st,s,b)\in S
    \right]
    \ge (\delta/4q -\nu(\lam)/2q)^2=\delta_\extr(\lam).
    \end{align}
    We now construct an OSMAC adversary $\A_\OSM$ as follows:
    \begin{enumerate}
        \item On input $\pp$, adversary $\A_\OSM$ samples $(\ket{\sk_\eps},\vk_\eps)\gets\OSM.\KeyGen(\pp)$ and runs $\st\gets\A^{\cO_\bot}(\pp,\ket{\sk_\eps})$. 
        
        \item Adversary $\A_\OSM$ initializes a list $L$, and for each $i\in[(2^n+1)\cdot\lam/\delta_\extr(\lam)]$, it does the following:
        \begin{enumerate}
            \item Sample $s_i\getsr\zo{n+\lam}$ and $b_i\getsr\zo{}$.
            \item Run $(s',r_i,(\st_v)_{v\in V_{r_i}})\gets\B_\extr^{\cO_\bot}(\st,s_i,b_i)$, and add $(s',r_i,(\st_v)_{v\in V_{r_i}})$ to $L$ if $s'=s_i$.
        \end{enumerate}
        
        \item Adversary $\A_\OSM$ randomly chooses two elements $(s_j,r_j,(\st_v)_{v\in V_{r_j}})$ and $(s_k,r_k,(\st_v)_{v\in V_{r_k}})$ of $L$.
        If $s_j= s_k$ or $r_j\ne r_k$, $\A_\OSM$ outputs $\bot$. 
        Otherwise, $\A_\OSM$ outputs the first collision $(\vk_v,m,\sigma_v,m',\sigma_v')$ it finds among the two tuples. 
    \end{enumerate}
    Adversary $\A_\OSM$ runs in $2^n\cdot\poly(\lam)$ time for infinitely many $\lam\in\N$, since $\delta$ is non-negligible and $(\A,\B)$ runs in QPT.
    Furthermore, assuming $(\pp,\mvk,\vk_\eps,\st)$ is a good tuple, \cref{eq:adv-extractor} and the number of iterations ensure that $L$ contains $2^n+1$ valid MAC chains $(s_i,r_i,(\st_v)_{v\in V_{r_i}})$ with $1-2^{-O(\lam)}$ probability. 
    By the pigeonhole principle, there exists $(s_j,r_j,(\st_v)_{v\in V_{r_j}})$ and $(s_k,r_k,(\st_v)_{v\in V_{r_k}})$ among those $2^n+1$ valid MAC chains such that $r_j=r_k$, and with overwhelming probability we have $s_j\ne s_k$ for all items in $L$ since $s_i\getsr\zo{n+\lam}$ in each iteration. 
    Thus, such a pair of MAC chains must have a collision where different signatures are given on the same verification key. Specifically, if $\vk_{r_j}$ is not the verification key for the collision, then $\vk_{r_j}\neq\vk_{r_k}$, implying the chain must differ for some other $v\in V_{r_j}\setminus\setbk{r_j}$, since $\vk_\eps$ is the same for both. 
    The advantage of $\A_\OSM$ is then at least
    \[\frac{\delta}{2\abs{L}^2}\geq
    \frac{1}{\poly(\lam)\cdot2^{2n}}>\frac{1}{2^{n^2}},\]
    for infinitely many $\lam\in\N$, as desired. 
\end{proof}

\begin{theorem}[Attack on One-Way Unclonable Security]\label{thm:clonable-ute-attack}
    Suppose $\Pi_\OSM$ is correct and $n=\omega(\log\lam)$. Then, \cref{cons:clonable-ute} is not $k$-adversary one-way unclonable secure (\cref{def:k-UE}) for any $k=\poly(\lam)$ in the classical oracle model. 
\end{theorem}
\begin{proof}
    We give a two stage adversary $(\A,\B)$ to break one-way unclonable security for any $k=\poly(\lam)$ number of second stage adversaries (all second stage adversaries run the same algorithm $\B$):
    \begin{enumerate}
        \item At the beginning of the game, $\A$ gets a single ciphertext $\ketct=(\ket{\sk_\eps},\cO_{\ek,m,\vk_\eps},\pp)$ from the challenger and decides on any $k=\poly(\lam)$ number of second stage adversaries. 
        \item For all $r\in[k]$ (where $r\in\zo{n}$ is appropriately padded), $v\in V_{r}\setminus\setbk{\eps}$ adversary $\A$ samples
            \[
            (\ket{\sk_v},\vk_v)\gets\OSM.\KeyGen(\pp)\text{ and }(\ket{\sk_{v\oplus1}},\vk_{v\oplus1})\gets\OSM.\KeyGen(\pp).
            \] 
        \item For all $r\in[k]$, $v\in V_{r}\setminus\setbk{r}$ adversary $\A$ computes $\sigma_v\gets\OSM.\Sign(\ket{\sk_v},\vk_{v\|0}\|\vk_{v\|1})$. 
        \item Adversary $\A$ sets $\st_\eps=(\vk_0,\vk_1,\sigma_\eps)$ and for all $r\in[k]$ sets $\st_v=(\vk_v,\vk_{v\|0},\vk_{v\|1},\sigma_v)$ for $v\in V_{r}\setminus\setbk{\eps,r}$. For the $\ord{r}$ second stage adversary, $\A$ outputs \ifelsesubmiteq{(\cO_{\ek,m,\vk_\eps},r,\setbk{\st_v}_{v\in V_r\setminus\setbk{r}},\ket{\sk_r},\vk_r)} in its register. 
        \item On input $(r,\setbk{\st_v}_{v\in V_r\setminus\setbk{r}},\ket{\sk_r},\vk_r)$ and $\dk=s$, adversary $\B$ computes $\sigma_r\gets\OSM.\Sign(\ket{\sk_r},s)$, sets $\st_r=(\vk_r,\sigma_r)$, and outputs $\cO(s,r,\setbk{\st_v}_{v\in V_r})$.
    \end{enumerate}
    Since $k$ is polynomial, $(\A,\B)$ is efficient. Since the MAC chain is generated honestly for each $r\in[k]$ and $\B$ is given the honest decryption key $s$, the checks in $\cO_{\ek,m,\vk_\eps}$ pass with overwhelming probability by correctness of $\Pi_\OSM$. Thus, all $k$ second stage adversaries recover $m$ with overwhelming probability. 
\end{proof}

\paragraph{Parameter setting.} Setting $n=\log^2\lam$, satisfies the constraint in \cref{thm:clonable-ute-attack} and correspondingly assumes quasi-polynomial security of $\Pi_\OSM$ for \cref{thm:clonable-ute-secure}. This yields the following corollary:

\begin{corollary}[Clonable UTE]
    Assuming the quasi-polynomial quantum hardness of LWE in the classical oracle model, there exists an indistinguishability secure UTE scheme that is not one-way unclonable secure for an unbounded polynomial number of second stage adversaries. 
\end{corollary}
\section{Untelegraphable Encryption with Everlasting Security}
In this section, we explore the notion of everlasting security (\cref{def:everlasting}) for untelegraphable encryption schemes. In particular, we construct collusion-resistant untelegraphable encryption with everlasting security in the quantum random oracle model (QROM). We also define the notion of weakly-efficient shadow tomography (WEST) and show that such a notion leads to an attack on the everlasting security of any correct collusion-resistant UTE scheme. Finally, we show that working in the QROM or assuming some other non-falsifiable assumption may be inherent when trying to construct UTE with everlasting security, by giving a black-box separation of public-key everlasting one-way secure UTE and any falsifiable assumption, assuming sub-exponential one-way functions exist.   

\subsection{Collusion-Resistant Untelegraphable Encryption with Everlasting Security}
\label{sec:everlasting-ute}
In this section, we describe our construction of collusion-resistant untelegraphable encryption with everlasting security in the quantum random oracle model (QROM), which parallels the construction of certified everlasting hiding commitments in~\cite{C:HMNY22}. A modification of the construction can also be used to build single-ciphertext public key unclonable encryption with everlasting security. 

\begin{construction}[Everlasting Secure UTE]\label{cons:everlasting-ute}
    Let $\lam$ be a security parameter. Our construction relies on the following additional primitives:
    \begin{itemize}
        \item Let $\Pi_\SKE=(\SKE.\Gen,\SKE.\Enc,\SKE.\Dec)$ be an SKE scheme with message space $\zo{\lam}$.
        \item Let $\Pi_\OneUTE=(\OneUTE.\Gen,\OneUTE.\Enc,\OneUTE.\Dec)$ be a secret-key untelegraphable encryption scheme with message space $\cM=\setbk{\cM_\lam}_{\lam\in\N}$ and key space $\zo{\ell(\lam)}$. 
        \item Let $\cH=\setbk{H_\lam:\zo{\lam}\to\zo{\ell}}$ be a hash function family, which will be modeled as a random oracle in the security analysis. 
    \end{itemize}
    We now construct our everlasting UE scheme $\Pi_\UE=(\Gen,\Enc,\Dec)$ as follows:
    \begin{itemize}
        \item $\Gen(1^\lam):$ On input the security parameter $\lam$, the key generation algorithm samples $k\gets\SKE.\Gen(1^\lam)$, $H=H_\lam$ from $\cH$, and outputs $(k,H)$.
        \item $\Enc(\sk,m):$ On input the secret key $\sk=(k,H)$ and a message $m\in\cM_\lam$, the encryption algorithm samples $\sk_\OneUTE\gets\OneUTE.\Gen(1^\lam),r\getsr\zo{\lam}$, computes $\ket{\ct_\OneUTE}\gets\OneUTE.\Enc(\sk_\OneUTE,m),\ct_\SKE\gets\SKE.\Enc(k,r)$ and outputs 
        \[
        \ketct=(\ket{\ct_\OneUTE},\ct_\SKE,H(r)\xor\sk_\OneUTE).
        \]
        \item $\Dec(\sk,\ket{\ct}):$ On input the secret key $\sk=(k,H)$ and a quantum ciphertext $\ket{\ct}=(\ket{\ct_1},\ct_2,\ct_3)$, the decryption algorithm outputs \ifelsesubmiteq{\OneUTE.\Dec(H(\SKE.\Dec(k,\ct_2))\xor\ct_3,\ket{\ct_1}).}
    \end{itemize}
\end{construction}
\begin{theorem}[Correctness]\label{thm:everlasting-ute-correct}
    Suppose $\Pi_\OneUTE$ and $\Pi_\SKE$ are correct. 
    Then, \cref{cons:everlasting-ute} is correct. 
\end{theorem}
\begin{proof}
    Fix any $\lam\in\N$ and $m\in\cM_\lam$. Let $\sk=(k,H)\gets\Gen(1^\lam)$ and 
    \[\ketct=(\ket{\ct_\OneUTE},\ct_\SKE,H(r)\xor\sk_\OneUTE)\gets\Enc(\sk,m).\] 
    We consider the output of $\Dec(\sk,\ketct)$. 
    By $\Pi_\SKE$ correctness, $r\gets\SKE.\Dec(k,\ct_\SKE)$ with overwhelming probability. 
    By $\Pi_\OneUTE$ correctness, $m\gets\OneUTE.\Dec(H(r)\xor H(r)\xor\sk_\OneUTE,\ket{\ct_\OneUTE})$ with overwhelming probability, as desired.
\end{proof}

\begin{theorem}[Everlasting Collusion-Resistant Security]
\label{thm:everlasting-ute-secure}
    Suppose $\Pi_\OneUTE$ satisfies statistical indistinguishability security (\cref{def:UTE}) and $\Pi_\SKE$ satisfies CPA security. 
    Then, \cref{cons:everlasting-ute} satisfies everlasting collusion-resistant security in the quantumly-accessible random oracle model. 
\end{theorem}
\begin{proof}
    Suppose there exists a QPT first stage adversary $\A$ and an unbounded time second stage adversary $\B$ that has non-negligible advantage $\delta$ in the collusion-resistant security game and makes at most $Q$ message queries. 
    We define a sequence of hybrid experiments:
    \begin{itemize}
        \item $\hyb_0$: This is the collusion-resistant game with bit $b=0$. 
        In particular:
        \begin{itemize}
            \item At the beginning of the game, the challenger samples $k\gets\SKE.\Gen(1^\lam)$, $H=H_\lam$ from $\cH$, and sets $\sk=(k,H)$. 
            \item When adversary $\A$ makes its $\ord{i}$ query $(m^{(i)}_0,m^{(i)}_1)$, the challenger samples $\sk^{(i)}_\OneUTE\gets\OneUTE.\Gen(1^\lam)$, $r_i\getsr\zo{\lam}$, sets $h_i=H(r_i)\oplus\sk^{(i)}_\OneUTE$, and computes 
            \[
            \ket{\ct^{(i)}_\OneUTE}\gets \OneUTE.\Enc(\sk^{(i)}_\OneUTE,m^{(i)}_0)
            \eqand \ct^{(i)}_\SKE\gets\SKE.\Enc(k,r_i).
            \] 
            The challenger gives the ciphertext $\ket{\ct^{(i)}}=
            \parens{\ket{\ct^{(i)}_\OneUTE},\ct^{(i)}_\SKE,h_i}$ to $\A$.
            \item Adversary $\A$ then outputs a classical string $\st$, which is given to $\B$ along with $\sk$ from the challenger. 
            Adversary $\B$ outputs a bit $b_\B$, which is the output of the experiment. 
        \end{itemize}
        \item $\hyb_i$: Same as $\hyb_{i-1}$ except the challenger now computes  \ifelsesubmiteq{\diff{\ket{\ct^{(i)}_\OneUTE}\gets \OneUTE.\Enc(\sk^{(i)}_\OneUTE,m^{(i)}_1)}.}
    \end{itemize}
    Note that $\hyb_Q$ is the collusion-resistant game with bit $b=1$. 
    We write $\hyb_{i}(\A,\B)$ to denote the output distribution of
    an execution of $\hyb_{i}$ with adversary $(\A,\B)$.
    We now show that for each $i\in[Q]$, experiments $\hyb_{i-1}(\A,\B)$ and $\hyb_i(\A,\B)$ are indistinguishable.
    
    \begin{lemma}\label{lem:cr-everlasting}
        Suppose $\Pi_\OneUTE$ satisfies statistical indistinguishability security and $\Pi_\SKE$ satisfies one-way CPA security for QPT adversaries. 
        Then, for all $i\in[Q]$, there exists a negligible function $\negl(\cdot)$ such that for all $\lam\in\N$, we have
        \[
        \abs{\Pr[\hyb_{i-1}(\A,\B)=1]-\Pr[\hyb_{i}(\A,\B)=1]}=\negl(\lam).
        \]
    \end{lemma}
    \begin{proof}
        We define the following intermediate hybrid experiments:
        \begin{itemize}
            \item $\hyb'_{i-1}$: Same as $\hyb_{i-1}$ except \diff{the oracle given to $\A$ is replaced with $H_{r_i\to H'}$}, which is $H$ reprogrammed according to $H'$ on an input $r$ where $H'$ is another independent uniformly random function. In particular, $H_{r_i\to H'}(x)=H(x)$ when $x\neq r_i$ and $H_{r_i\to H'}(x)=H'(x)$ when $x=r_i$.
            \item $\hyb''_{i-1}$: Same as $\hyb'_{i-1}$ except for the following points:
            \begin{itemize}
                \item The challenger samples $\diff{h_i\getsr\zo{\ell}}$.
                \item The oracle given to $\A$ \diff{is replaced with $H'$}, where $H'$ is a uniformly random function that is independent of $H$.
                \item The oracle given to $\B$ \diff{is replaced with $H'_{r_i\to t}$}, which is $H'$ reprogrammed to $t=h_i\xor\sk^{(i)}_\OneUTE$ on input $r_i$. 
            \end{itemize}
            \item $\hyb'''_{i-1}$: Same as $\hyb''_{i-1}$ except the challenger now computes  \ifelsesubmiteq{\diff{\ket{\ct^{(i)}_\OneUTE}\gets \OneUTE.\Enc(\sk^{(i)}_\OneUTE,m^{(i)}_1)}.}
        \end{itemize}
        
        \begin{claim}\label{claim:eute-reprogram1}
            Suppose $\Pi_\SKE$ satisfies one-way CPA security for QPT adversaries. 
            Then, for all $i\in[Q]$, there exists a negligible function $\negl(\cdot)$ such that for all $\lam\in\N$, we have
            \[
            \abs{\Pr[\hyb_{i-1}(\A,\B)=1]-\Pr[\hyb'_{i-1}(\A,\B)=1]}=\negl(\lam).
            \]
        \end{claim}
        \begin{proof}
            Suppose $(\A,\B)$ distinguish the experiments with non-negligible probability $\delta'$. 
            We will use $(\A,\B)$ to construct an adversary $\B_\SKE$ that breaks one-way CPA security. 
            
            First, consider an adversary $\Ahat$ which is given an oracle $\cO$ that is either $H$ or $H_{r_i\to H'}$ and an input $(r_i,H)$ where $r_i\getsr\zo{\lam}$ and $H$ represents the entire truth table of $H$.
            Adversary $\Ahat$ runs $\hyb_{i-1}(\A,\B)$ except that it uses $\cO$ to simulate $\A$'s random oracle queries and the truth table of $H$ to simulate $\B$'s random oracle queries. We assume without loss of generality that $\A$ makes exactly $Q_{\A}$ queries. 
            Thus, $\Pr[\hyb_{i-1}(\A,\B)=1]=\Pr[\Ahat^H(r_i,H)=1]$ and 
            $\Pr[\hyb'_{i-1}(\A,\B)]=\Pr[\Ahat^{H_{r_i\to H'}}(r_i,H)=1]$.
            Setting $\B_{\mathsf{O2H}}$ to be an algorithm that measures a uniformly random query of $\Ahat$, by \cref{lem:O2H}, we have
            \[
            \delta'\leq 
            \abs{\Pr[\Ahat^H(r_i,H)=1]-\Pr[\Ahat^{H_{r_i\to H'}}(r_i,H)=1]}
            \leq 2Q_{\A}\cdot\sqrt{\Pr[\B_{\mathsf{O2H}}^{H_{r_i\to H'}}(r_i,H)=r_i]}.
            \]
            Let $\Bhat_{\mathsf{O2H}}$ be the same as $\B_{\mathsf{O2H}}$ except it does not take the truth table of $H$ as input and samples $h_i\getsr\zo{\ell}$. 
            Then, \[\Pr[\B_{\mathsf{O2H}}^{H_{r_i\to H'}}(r_i,H)=r_i]=\Pr[\Bhat_{\mathsf{O2H}}^{H_{r_i\to H'}}(r_i)=r_i]\] since the truth table for $H$ is only needed to generate $h_i=H(r_i)\xor \sk^{(i)}_\OneUTE$ (it halts before $\B$ is run) and $H_{r_i\to H'}$ reveals nothing about $H(r_i)$. 
            Furthermore, for any $r_i$, the function $H_{r_i\to H'}$ is random when $H,H'$ are random, meaning \[\Pr[\Bhat_{\mathsf{O2H}}^{H_{r_i\to H'}}(r_i)=r_i]=\Pr[\Bhat_{\mathsf{O2H}}^H(r_i)=r_i]\geq \parens{\frac{\delta'}{2Q_{\A}}}^2.\]
            Note that $\Bhat_{\mathsf{O2H}}$ ignores its input and simulates $\hyb_{i-1}(\A,\B)$ by sampling $h_i\getsr\zo{\ell}$, measuring a random query $j\getsr[Q_\A]$ of adversary $\A$, and outputting the measurement.
            With this in mind, we can construct $\B_\SKE$ as follows:
            \begin{enumerate}
                \item At the beginning of the game, $\B_\SKE$ is given $\ct^{(i)}_\SKE$ from its challenger, samples $h_i\getsr\zo{\ell}$, and samples $i^\ast\getsr[Q_\A]$. 
                
                \item When adversary $\A$ makes the $\ord{j}$ query $(m^{(j)}_0,m^{(j)}_1)$, adversary $\B_\SKE$ does the following:
                \begin{itemize}
                    \item For each $j\in[Q]$, sample $\sk^{(j)}_\OneUTE\gets\OneUTE.\Gen(1^\lam)$ and compute 
                    \[
                    \ket{\ct^{(j)}_\OneUTE}\gets \OneUTE.\Enc(\sk^{(j)}_\OneUTE,m^{(j)}_0).
                    \] 
                    \item For $j\in[Q]\setminus\setbk{i}$, sample $r_j\getsr\zo{\lam}$, set $h_j=H(r_j)\oplus\sk^{(j)}_\OneUTE$, and
                    submit $r_j$ to the CPA challenger to get $\ct^{(j)}_\SKE$. 
                \end{itemize}
                Adversary $\B_\SKE$ replies with $\ket{\ct^{(j)}}=
                \parens{\ket{\ct^{(j)}_\OneUTE},\ct^{(j)}_\SKE,h_j}$.
                
                \item For random oracle queries, $\B_\SKE$ simulates them efficiently (see~\cref{sec:qrom}) and measures query $i^\ast$ that $\A$ makes. Adversary $\B_\SKE$ outputs the measurement outcome. 
            \end{enumerate}
            By definition $\B_\SKE$ outputs $r_i$ (which we denote as the random message used to construct $\ct^{(i)}_\SKE$) with probability $\Pr[\Bhat_{\mathsf{O2H}}^H(r_i)=r_i]\geq \parens{\frac{\delta'}{2Q_{\A}}}^2$, which is non-negligible. Thus, $\B_\SKE$ breaks one-way CPA security, a contradiction.
        \end{proof}

        \begin{claim}\label{claim:eute-reprogram2}
            For all $i\in[Q]$ and $\lam\in\N$, $\Pr[\hyb'_{i-1}(\A,\B)=1]=\Pr[\hyb''_{i-1}(\A,\B)=1]$.
        \end{claim}
        \begin{proof}
            First, $h_i=H(r_i)\xor \sk^{(i)}_\OneUTE$ in $\hyb'_{i-1}$ and $h_i$ is uniform in $\hyb''_{i-1}$, so the marginal of $h_i$ is identical in both distributions. 
            Next, $H_{r_i\to H'}$ is independent of $h_i$ in $\hyb'_{i-1}$ and is a uniformly random function, which is also the case in $\hyb''_{i-1}$ with $H'$, so the hybrids have the same joint distribution for the first stage oracle and $h_i$. 
            Moreover, in $\hyb'_{i-1}$, we have $H(r_i)=h_i\xor \sk^{(i)}_\OneUTE$ and in $\hyb''_{i-1}$ we have $H'_{r_i\to t}(r_i)=t=h_i\xor \sk^{(i)}_\OneUTE$.
            Lastly, the first and second stage oracles only differ at the point $r_i$ in both hybrids, meaning the joint distributions for both oracles and $h_i$ are identical.
            Since the oracles and $h_i$ are the only modified components, this completes the proof. 
        \end{proof}

        \begin{claim}\label{claim:eute-1ute}
            Suppose $\Pi_\OneUTE$ satisfies statistical indistinguishability security.
            Then, for all $i\in[Q]$, there exists a negligible function $\negl(\cdot)$ such that for all $\lam\in\N$, we have
            \[
            \abs{\Pr[\hyb''_{i-1}(\A,\B)=1]-\Pr[\hyb'''_{i-1}(\A,\B)=1]}=\negl(\lam).
            \]
        \end{claim}
        \begin{proof}
            Suppose $(\A,\B)$ distinguish the experiments with non-negligible probability $\delta'$. 
            We will use $(\A,\B)$ to construct an adversary $(\A',\B')$ that breaks indistinguishability security as follows:
            \begin{enumerate}
                \item Adversary $\A'$ first samples $r_j\getsr\zo{\lam},s_j\getsr\zo{\ell}$ for each $j\in[Q]$ and sets $H'(r_j)=s_j$ for $j\neq i$. 
                It also samples $\sk=k\gets\SKE.\Gen(1^\lam)$ and for all $j\in[Q]$ samples $\ct^{(j)}_\SKE\gets\SKE.\Enc(k,r_j)$. 
                
                \item Adversary $\A'$ answers random oracle queries at random or according to the values already set for $H'$ if they are queried. 

                \item When adversary $\A$ makes the $\ord{j}$ query $(m^{(j)}_0,m^{(j)}_1)$, adversary $\A'$ does the following:
                \begin{itemize}
                    \item For $j\neq i$, sample $\sk^{(j)}_\OneUTE\gets\OneUTE.\Gen(1^\lam)$ and compute $\ket{\ct^{(j)}_\OneUTE}\gets
                    \OneUTE.\Enc(\sk^{(j)}_\OneUTE,m^{(j)}_{b_j})$,
                    where $b_j=0$ for $j> i$ and $b_j=1$ for $j<i$. Set $h_j=s_j\xor\sk^{(j)}_\OneUTE$.
                    \item For $j=i$, sample $h_i\getsr\zo{\ell}$ and send $m^{(i)}_0,m^{(i)}_1$ to the UTE challenger to get $\ket{\ct^{(i)}_\OneUTE}$.
                \end{itemize}
                Adversary $\A'$ replies with $\ket{\ct^{(j)}}=\parens{\ket{\ct^{(j)}_\OneUTE},\ct^{(j)}_\SKE,h_j}$.
                
                \item When adversary $\A$ outputs $\st$, adversary $\A'$ passes the state $(\st,h_i,H')$ to $\B'$, where $H'$ is the current oracle information from $\A$'s queries and the values set for $H'$. 

                \item On input $(\st,h_i,H')$ and $\sk^{(i)}_\OneUTE$, adversary $\B'$ sets $t=h_i\xor\sk^{(i)}_\OneUTE$ and sets $H'(r_i)=t$. 
                It runs $\B$ on input $(\st,\sk)$ and answers random oracle queries at random or according to the values already set for $H'$ if they are queried. Adversary $\B'$ outputs whatever $\B$ outputs. 
            \end{enumerate}
            If the UTE challenger encrypts $m^{(i)}_0$, adversary $(\A',\B')$ perfectly simulates $\hyb''_{i-1}(\A,\B)$. 
            If the UTE challenger encrypts $m^{(i)}_1$, adversary $(\A',\B')$ perfectly simulates $\hyb'''_{i-1}(\A,\B)$.
            Thus, $(\A',\B')$ has advantage $\delta'$ in the indistinguishability security game, a contradiction. 
        \end{proof}

        \begin{claim}\label{claim:eute-unwind}
            Suppose $\Pi_\SKE$ satisfies one-way CPA security for QPT adversaries. 
            Then, for all $i\in[Q]$, there exists a negligible function $\negl(\cdot)$ such that for all $\lam\in\N$, we have
            \[
            \abs{\Pr[\hyb'''_{i-1}(\A,\B)=1]-\Pr[\hyb_{i}(\A,\B)=1]}=\negl(\lam).
            \]
        \end{claim}
        \begin{proof}
            Analogous to the proofs of \cref{claim:eute-reprogram1,claim:eute-reprogram2}.
        \end{proof}

        \noindent Combining \cref{claim:eute-1ute,claim:eute-reprogram1,claim:eute-reprogram2,claim:eute-unwind} proves the lemma by a hybrid argument.
        
    \end{proof}

    \noindent The theorem follows by a hybrid argument via \cref{lem:cr-everlasting}.
    
\end{proof}

\begin{remark}[Plain Model Insecurity]\label{rem:everlasting-ute-qfhe}
    We note that for any concrete hash function $H$, if $\Pi_\SKE$ is FHE for quantum circuits (QFHE) with classical ciphertexts for classical messages, then even one-way non-everlasting UTE security is broken. Namely, QFHE can be used to get a classical ciphertext that contains the message given $\ketct$ by computing the homomorphic evaluation of $\Dec$. Thus, the second-stage adversary can efficiently recover $m$ given the SKE key $k$. 
\end{remark}

\subsection{Impossibility of Weakly-Efficient Shadow Tomography}
\label{sec:west}
In this section, we show an efficient attack on any collusion-resistant UTE with everlasting security from weakly-efficient shadow tomography (WEST). In particular, this implies the impossibility of WEST for general mixed states in the QROM (since random oracles imply one-way functions and thus CPA-secure SKE). 

\paragraph{Weakly-efficient shadow tomography.} We first define the notion of weakly-efficient shadow tomography (WEST). This notion captures shadow tomography where the post-processing can run in unbounded time, but the initial stage must be efficient. In particular, this generalizes the notion of classical shadows~\cite{NatPhys:HKP20}, which has been shown to compute useful information about unknown quantum states when the family of measurements are projections to pure states. 

\begin{definition}[Weakly-Efficient Shadow Tomography]
\label{def:west} 
Let $E$ denote a uniform quantum circuit family with classical binary output that takes as input $i\in[M]$ and an $n$-qubit quantum state $\rho$. Let a shadow tomography procedure be slightly redefined to be a pair of algorithms $(\cT_1,\cT_2)$ where $\cT_1$ takes as input $E$ and $k$ copies of $\rho$, and outputs a classical string $\st$ such that $\Pr[\abs{\cT_2(i,\st)-\Pr[E(i,\rho)=1]}\le\eps]\ge 1-\delta$. We call said procedure \textit{weakly-efficient} if the number of copies $k$ and the runtime of $\cT_1$ are both $\poly(n, \log M, 1/\eps,\log 1/\delta)$. Note that $\cT_2$ can run in \textit{unbounded} time.
\end{definition}

\begin{theorem}[WEST Breaks $t$-Copy Everlasting Secure UTE]\label{thm:west-attack}
    Let $\lam$ be a security parameter.
    Suppose there exists an algorithm $\cT=(\cT_1,\cT_2)$ that satisfies \cref{def:west} for all mixed quantum states with number of copies $k=k(n,\log M,1/\eps,\log(1/\delta))$ for polynomial $k$. 
    Then, there does not exist $k$-copy everlasting secure untelegraphable encryption that also satisfies correctness. 
\end{theorem}
\begin{proof}
    Let $\Pi_\UTE=(\Gen,\Enc,\Dec)$ with $\poly(\lam)$ be a candidate correct UTE scheme with message space $\zo{}$ (without loss of generality) and decryption key space $\cK=\setbk{\cK_\lam}_{\lam\in\N}$. 
    Let $n=n(\lam)$ be the maximum size of a ciphertext, $M=\abs{\cK_\lam}$, $\eps<1/2$ be a constant, and $\delta=2^{-\lam}$.
    We define the quantum circuit $E_\lam(i,\rho)$ to output $\Dec(\dk,\ketct)\in\zo{}$ where the input is parsed as $i=\dk$ and $\rho=\ketct$.
    We now construct an adversary $(\A,\B)$ that breaks $k$-copy everlasting security of $\Pi_\UTE$:
    \begin{enumerate}
        \item On input $\setbk{\ket{\ct_i}}_{i\in[k]}$, adversary $\A$ runs $\cT_1(E_\lam,\setbk{\ket{\ct_i}}_{i\in[k]})$ to get a classical state $\st$, which $\A$ outputs as its classical state.
        \item On input $\dk$ and $\st$, adversary $\B$ outputs 1 if $\cT_2(\dk,\st)>1/2$ and 0 otherwise.
    \end{enumerate}
    By definition of UTE, $n=\poly(\lam)$ and $\log M=\poly(\lam)$. By the choice of $\delta$, $\log(1/\delta)=\lam$ and $1/\eps=3$. Thus, $\cT_1$ runs in QPT and $(\A,\B)$ meets the efficiency condition. Since $\Pr[\abs{\cT_2(i,\st)-\Pr[E(i,\rho)=1]}\le\eps]\ge 1-\delta$, by the setting of $\eps,\delta$, we have \[\Pr[\abs{\cT_2(i,\st)-\Pr[E(i,\rho)=1]}<1/2]\ge 1-2^{-\lam}.\]
    By correctness of UTE, we have $\Pr[E(i,\rho)=1]=\negl(\lam)$ when $\setbk{\ket{\ct_i}}_{i\in[k]}$ are encryptions of 0 and $\Pr[E(i,\rho)=1]\ge1-\negl(\lam)$ when $\setbk{\ket{\ct_i}}_{i\in[k]}$ are encryptions of 1. This implies $\cT_2(\dk,\st)>1/2$ for encryptions of 1 and $\cT_2(\dk,\st)<1/2$ for encryptions of 0 with probability $1-2^{-\lam}$. Thus, $(\A,\B)$ wins the $k$-copy security game with overwhelming probability. 
\end{proof}

\begin{corollary}
    Weakly-efficient shadow tomography for all mixed quantum states does not exist in the quantum random oracle model. 
\end{corollary}
\begin{proof}
    Follows by \cref{thm:everlasting-ute-correct,thm:everlasting-ute-secure,thm:west-attack} and the fact that collusion-resistant security implies $t$-copy security.
    To adapt the above attack to \cref{cons:everlasting-ute} with only query access to the hash $H$, we can define the circuit $E$ to run the decryption algorithm of \cref{cons:everlasting-ute} modified to take the value $H(r)$ as input.
    The classical part of the ciphertext can be appended to $\st$, which will allow $H(r)$ to be computed after the key is revealed. Thus, even if the WEST procedure requires the circuit $E$ to be composed of atomic gates only (no hash gate), the attack still works. 
\end{proof}

\subsection{Impossibility of Public-Key Untelegraphable Encryption with Everlasting Security}
\label{sec:everlasting-impossible}
In this section, we give evidence that a non-falsifiable assumption is necessary to construct everlasting secure UTE (and thus UE as well). 
We do so by showing that any black-box reduction that proves everlasting one-way security of a public-key UTE (UTPKE) scheme with respect to an assumption $A$ can be used to construct an efficient adversary that breaks $A$, assuming one-way functions exist. 
This complements the observation that \cref{cons:everlasting-ute} is broken for any concrete choice of hash function (see \cref{rem:everlasting-ute-qfhe}) if the classical encryption scheme is quantum fully homomorphic.

\paragraph{Falsifiable cryptographic assumptions.} We start by recalling the notion of a falsifiable cryptographic assumption. Intuitively, this notion refers to any cryptographic assumption which can be broken in a way that can efficiently checked, such as one-way functions. We give the formal definition from~\cite{GW11}.

\begin{definition}[{Falsifiable Cryptographic Assumption~\cite{GW11}}]
\label{def:falsifiable}
    A falsifiable cryptographic assumption consists of an interactive challenger $\C$ and a constant $c\in[0,1)$. 
    On security parameter $\lam$, the challenger $\C(1^\lam)$ interacts with an algorithm $\A(1^\lam)$ and may output a special symbol $\win$. If this occurs, we say that $\A(1^\lam)$ wins $\C(1^\lam)$.
    The assumption associated with the tuple $(\C,c)$ says that for all efficient (possibly non-uniform) algorithms $\A$, there exists a negligible function $\negl$ such that for all $\lam\in\N$, $\Pr[\A(1^\lam) \text{ wins } \C(1^\lam)]\leq c+\negl(\lam)$, where the probability is over the random coins of $\C$ and $\A$.
    We say that a (possibly inefficient) algorithm $\A$ breaks such an assumption if its probability of winning exceeds that of the assumption. When considering quantum assumptions, we refer to the purified version of algorithm $\A$ as a unitary $\mat{U}_\A$.
\end{definition}

\noindent We also define what it means for a black-box reduction to show everlasting security of a UTE scheme:

\begin{definition}[Black-Box Reduction for Everlasting Secure UTPKE]
\label{def:bb-red}
    A black-box reduction showing everlasting security of a UTE scheme based on a falsifiable assumption $(\C,c)$ is an efficient oracle access algorithm $\cR^{(\cdot)}$ such that, for every (possibly inefficient) adversary $(\A,\B)$ that wins the everlasting UTPKE security game with non-negligible advantage, algorithm $\cR^{\mat{U}_\A}$ breaks the assumption.
    We model oracle query to $\mat{U}_\A$ as follows: given an input state $\ket{\psi}=(\pk,\ketct)$, the response is the output register $\qreg{R}_\outsf$ of $\mat{U}_\A\ket{\psi}$. 
    Furthermore, we do not give oracle access to the inverse unitary $\mat{U}_\A^\dagger$.
    Note that since $\B$ runs in unbounded time in the everlasting security game, so oracle access to $\B$ is not given. 
\end{definition}

\noindent We now state and prove the main theorem of this section:

\begin{theorem}[Impossibility of Everlasting One-Way Secure UTPKE]
\label{thm:everlasting-impossible}
    Let $\lam\in\N$ be a security parameter and $\delta\in(0,1)$ be a constant.
    Let $\Pi_\UTE=(\Gen,\Enc,\Dec)$ be any candidate UTPKE scheme that satisfies correctness. 
    Suppose post-quantum CPA-secure encryption exists for size $2^{\lam^\delta}$ adversaries such that the ciphertexts of different messages are statistically far apart.
    Then, for any falsifiable assumption $(\C,c)$, one of the following holds:
    \begin{itemize}
        \item The assumption $(\C,c)$ is false.
        \item There is no black-box reduction (\cref{def:bb-red}) showing everlasting one-way security of $\Pi_\UTE$ based on the assumption $(\C,c)$.
    \end{itemize}
\end{theorem}
\begin{proof}
    Let $\lam\in\N$ be the security parameter, and suppose the length of the keys output from $\Gen(1^\lam)$ is $\abs{\pk}+\abs{\sk}\leq\lam^c$ for constant $c\in\N$.
    Let $n=n(\lam)=\lam^{(c+2)/\delta}$ be a CPA-secure encryption security parameter and let $\Pi_\SKE=\SKE.(\Gen,\Enc,\Dec)$ be an SKE scheme that is CPA-secure for size $2^{n^\delta}=2^{\lam^{c+2}}$ adversaries and has ciphertexts that are statistically far apart.
    Now, consider an inefficient adversary $\A$ that on input $(1^\kappa,\pk,\ketct)$ does the following:
    \begin{enumerate}
        \item Measures the $\pk$ register (since it is supposed to be classical).
        \item Rejection samples $(\pk',\sk')\gets\Gen(1^\kappa)$ until $\pk'=\pk$ and computes $m\gets\Dec(\sk',\ketct)$.
        \item Samples $\sk_\SKE\gets\SKE.\Gen(1^{n(\kappa)})$ and outputs $\st=\SKE.\Enc(\sk_\SKE,m)$. 
    \end{enumerate}
    Define the second stage adversary $\B$ to find $m$ by brute-force, which is possible since ciphertexts for $\Pi_\SKE$ must be statistically far from each other.
    Clearly $(\A,\B)$ breaks everlasting one-way security with non-negligible advantage.
    
    We now show that for every polynomial-size distinguisher $\cD(1^\lam)$, there exists a simulator $\cS(1^\lam)$ and negligible function $\negl$ such that for all $\lam\in\N$,
    \[\Pr[\cD^{\mat{U}_\cA}(1^\lam)=1]-\Pr[\cD^{\mat{U}_\cS(1^\lam)}(1^\lam)=1]=\negl(\lam).\]
    We define an efficient non-uniform simulator $\cS(1^\lam)$ for $\A$ that on input a query $(1^\kappa,\pk,\ketct)$ does the following:
    \begin{enumerate}
        \item If $\kappa\leq\kappa^\ast(\lam)=\floor{\log^{1/(c+1)}\lam}$, the simulator measures the $\pk$ register, uses its advice to get a valid key $\sk$ for $\pk$, and runs the steps of $\A$ from that point. The advice consists of a table $T_\lam$ that maps public keys $\pk\in\zo{\kappa^c}$ to a corresponding secret key $\sk$ for all $\kappa\leq\kappa^\ast(\lam)$.
        \item If $\kappa>\kappa^\ast(\lam)$, the simulator measures the $\pk$ register, samples $\sk_\SKE\gets\SKE.\Gen(1^n)$ and outputs $\st=\SKE.\Enc(\sk_\SKE,0)$.
    \end{enumerate}
    Note that $T_\lam$ has size at most $\kappa^\ast(\lam)2^{\kappa^\ast(\lam)}=\poly(\lam)$ by the bound on the key size, so $\cS$ is efficient and has polynomial-sized advice.
    Moreover, by an averaging argument, there exists a choice of $T_\lam$ such that for the adversary $\cA(T_\lam)$ that answers queries where $\kappa\leq\kappa^\ast(\lam)$ using the table, we have
    \[\Pr[\cD^{\mat{U}_\cA(T_\lam)}(1^\lam)=1]\geq \Pr[\cD^{\mat{U}_\cA}(1^\lam)=1],\] so we pick this table as advice. We now appeal to CPA-security of $\Pi_\SKE$ to show
    \[\Pr[\cD^{\mat{U}_\cA(T_\lam)}(1^\lam)=1]-\Pr[\cD^{\mat{U}_\cS(1^\lam)}(1^\lam)=1]=\negl(\lam).\]
    Suppose $\cD(1^\lam)$ makes at most $q(\lam)$ queries. For $i\in[q(\lam)]$, define $\cO_i$ to be the oracle that answers the first $i$ queries with $\cA(T_\lam)$ and the last $q(\lam)-i$ queries with $\cS(\lam)$. 
    Then, for $i\in[q(\lam)]$,
    \[
    \abs{\Pr[\cD^{\cO_{i-1}}(1^\lam)=1]-\Pr[\cD^{\cO_{i}}(1^\lam)=1]}=\negl(\lam).
    \]
    Suppose $\cD$ can distinguish $\cO_{i-1}$ and $\cO_{i}$. 
    We construct a non-uniform adversary $\B$ that breaks CPA-security as follows:
    \begin{enumerate}
        \item We first define the advice distribution. 
        For a security parameter $\lam\in\N$, adversary $\B$ is given $T_\lam$ along with a transcript of $\cD^{\cA(T_\lam)}(1^\lam;r)$ for the first $i$ queries and $i-1$ responses along with the randomness string $r$ and the message $m$ that $\cA(T_\lam)$ would derive (see step 2 of $\A$) in the $\ord{i}$ query. 
        
        \item Adversary $\B$ starts running $\cD(1^\lam;r)$ and answers the first $i-1$ queries as in the transcript. 
        
        \item If the $\ord{i}$ query satisfies $\kappa\leq\kappa^\ast(\lam)$, respond with $T_\lam$. Otherwise, submit the message pair $(0,m)$ to the CPA challenger for security parameter $1^{n(\kappa)}$ to get back $\st$ and respond with it.

        \item For the remaining queries, answer how $\cS(1^\lam)$ would. 
        Output whatever $\cD$ outputs. 
    \end{enumerate}
    Adversary $\B$ is efficient so long as $\cD$ is. If $\st$ is an encryption of $m$, adversary $\B$ simulates $\cO_i$ and if $\st$ is an encryption of 0, adversary $\B$ simulates $\cO_{i-1}$.   
    By assumption on the security of $\Pi_\SKE$, for $\kappa>\kappa^\ast(\lam)$, no adversaries of size 
    \[s(\lam)=2^{n(\kappa)^\delta}> 2^{n(\kappa^\ast(\lam))^\delta}=2^{\log^{c'}\lam},\]
    can break CPA security with non-negligible probability, where $c'=(c+2)/(c+1)>1$. 
    Since $s(\lam)$ is super polynomial in $\lam$, adversary $\B$ contradicts CPA security of $\Pi_\SKE$. 
    Thus, 
    \[\Pr[\cD^{\mat{U}_\cA(T_\lam)}(1^\lam)=1]-\Pr[\cD^{\mat{U}_\cS(1^\lam)}(1^\lam)=1]=\negl(\lam)\]
    follows by a hyrbid argument.
    To conclude the proof, consider a candidate reduction $\cR$. If we have
    \[
    \Pr[\cR^{\mat{U}_\A}(1^\lam) \text{ wins } \C(1^\lam)]> c+\eps(\lam)
    \]
    for non-negligible $\eps$, then it must also be the case that 
    \[
    \Pr[\cR^{\mat{U}_\cS(1^\lam)}(1^\lam) \text{ wins } \C(1^\lam)]> c+\eps(\lam)-\negl(\lam),
    \]
    which implies $(\C,c)$ is false.
\end{proof}

\begin{corollary}[Impossibility of Everlasting One-Way Secure UTPKE from OWFs]
\label{cor:everlasting-impossible}
    Let $\lam\in\N$ be a security parameter, $\delta\in(0,1)$ be a constant, and suppose post-quantum one-way functions exist for size $2^{\lam^\delta}$ adversaries.
    Then, any candidate UTPKE cannot have a black-box proof of everlasting one-way security based on a falsifiable assumption $(\C,c)$, unless the assumption is false.
\end{corollary}
\begin{proof}
    Follows by \cref{thm:everlasting-impossible} and the fact that CPA-secure encryption with statistically far ciphertexts for different messages can be constructed from one-way functions. 
\end{proof}
\section{Secret Sharing Resilient to Joint and Unbounded Classical Leakage}
In this section, we define and construct unbounded joint-leakage-resilient secret sharing (UJLRSS) for all policies from collusion-resistant UTE and classical secret sharing for all policies. 
Our definition extends the notion of unbounded leakage-resilient secret sharing defined in~\cite{CGLR24} to the setting of joint leakage.
\subsection{Secret Sharing Definitions}
\label{sec:utss-def}
In this section, we define UJLRSS and compare our definition to the recently introduced definitions of ULRSS in~\cite{CGLR24} and USS in~\cite{AGLL24}. 
Intuitively, the security notion we consider involves groups of non-qualifying parties as first stage adversaries, where each group classically leaks their shares, except for one group which gets its honest shares. 
Given the honest set of shares and all the classical leakage, it should be hard to distinguish any two secrets from each other.
In addition, we recall the definition of classical secret sharing for policies.
We give the formal definitions below, starting with the definition of an access policy. 
Much of the prose for this is taken from~\cite{C:JKKKPW17}.

\begin{definition}[Access Policy]
    An access policy $P$ on parties $[n]$ is a monotone set of subsets of $[n]$. 
    In particular, $P\subseteq2^{[n]}$ and for all $S\in P$ and $S\subseteq S'$ it holds that $S'\in P$.  
\end{definition}
\noindent In this work, we will be treating policies as functions $P:2^{[n]}\to\zo{}$ that output 1 on input $X$ if $X\in P$. 
Furthermore, we will assume $P$ can be described by a \textit{monotone Boolean circuit}. These are directed acyclic graphs in which leaves are labeled by input variables and every internal node is labeled by an AND or OR operation. 
We assume the gates in the circuit have fan-in $\kin$ and fan-out at most $\kout$, so for a given gate $g$ we represent the number of output wires as $\kout(g)$.
The computation is done in the natural way from the leaves to the root which corresponds to the output of the computation. 

\paragraph{Classical secret sharing.} We now recall the notion of classical secret sharing for access policies. This notion can be constructed for general policies from any classical CPA-secure encryption scheme, which itself can be constructed from one-way functions~\cite{VNSRK03}. 

\begin{definition}[Secret Sharing]\label{def:SS}
    Let $\lam$ be a security parameter and $n$ be the number of parties. 
    An secret sharing scheme for access policy $P:2^{[n]}\to\zo{}$ with message space $\cM=\setbk{\cM_\lam}_{\lam\in\N}$ is a tuple of PPT algorithms $\SSscheme=(\share,\reconst)$ with the following syntax:
    \begin{itemize}
        \item $\share(1^\lam,1^n,m)\to(s_1,\dots,s_n)$: On input the security parameter $\lam$, the number of parties $n$, and a message $m\in\cM_\lam$, the sharing algorithm outputs a tuple of shares $(s_1,\dots,s_n)$ for each party. 
        \item $\reconst(X,s_X)\to m$: On input the shares of a subset of parties $(X\subseteq[n],s_X=\setbk{s_i}_{i\in X})$, the reconstruction algorithm outputs a message $m\in\cM_\lam$.
    \end{itemize}
    We require that $\SSscheme$ satisfy the following properties:
    \begin{itemize}
        \item \textbf{Correctness:} For all $\lam,n\in\N$, $m\in\cM_\lam$, and $X$ such that $P(X)=1$, we have 
        \[
        \Pr[m\gets\reconst(X,s_X):(s_1,\dots,s_n)\gets\share(1^\lam,1^n,m)]\geq 1-\negl(\lam).
        \]
        \item \textbf{Security:} For a security parameter $\lam$, number of parties $n$, a bit $b\in\zo{}$, and an adversary $\A$, we define the security game as follows:
        \begin{enumerate}
            \item At the beginning of the game, adversary $\A$ outputs messages $m_0,m_1\in\cM_\lam$, along with a set $X$ such that $P(X)=0$. 
            
            \item The challenger samples $(s_1,\dots,s_n)\gets\share(1^\lam,1^n,m_b)$, and gives $s_{X}$ to $\A$. Adversary $\A$ outputs a bit $b'$, which is the output of the experiment. 
        \end{enumerate}
        We say $\SSscheme$ is secure if for all PPT adversaries $\A$ there exists a negligible function $\negl(\cdot)$ such that for all $\lam\in\N$, 
        \[\abs{\Pr[b'=1|b=0]-\Pr[b'=1|b=1]}= \negl(\lam)\] in the above security game.
    \end{itemize}
\end{definition}

\noindent We now define UJLRSS for access policies. 

\begin{definition}[Unbounded Joint-Leakage-Resilient Secret Sharing]
\label{def:UTSS}
    Let $\lam$ be a security parameter and $n$ be the number of parties. 
    An unbounded joint-leakage-resilient secret sharing scheme for access policy $P:2^{[n]}\to\zo{}$ with message space $\cM=\setbk{\cM_\lam}_{\lam\in\N}$ is a tuple of QPT algorithms $\Pi_\UTSS=(\share,\reconst)$ with the following syntax:
    \begin{itemize}
        \item $\share(1^\lam,1^n,m)\to(\ket{s_1},\dots,\ket{s_n})$: On input the security parameter $\lam$, the number of parties $n$, and a message $m\in\cM_\lam$, the sharing algorithm outputs a tuple of quantum shares $(\ket{s_1},\dots,\ket{s_n})$ for each party. 
        \item $\reconst(X,\ket{s_X})\to m$: On input the shares of a subset of parties $(X\subseteq[n],\ket{s_X}=\setbk{\ket{s_i}}_{i\in X})$, the reconstruction algorithm outputs a message $m\in\cM_\lam$.
    \end{itemize}
    We require that $\Pi_\UTSS$ satisfy the following properties:
    \begin{itemize}
        \item \textbf{Correctness:} For all $\lam,n\in\N$, $m\in\cM_\lam$, and $X$ such that $P(X)=1$, we have 
        \[
        \Pr[m\gets\reconst(X,\ket{s_X}):(\ket{s_1},\dots,\ket{s_n})\gets\share(1^\lam,1^n,m)]\geq 1-\negl(\lam).
        \]
        \item \textbf{UJLR Security:} For a security parameter $\lam$, number of parties $n$, a bit $b\in\zo{}$, and a two-stage adversary $(\A_0,\A_1,\dots,\A_n,\B)$, we define the secret sharing security game with joint and unbounded classical leakage as follows:
        \begin{enumerate}
            \item At the beginning of the game, $\A_0$ outputs messages $m_0,m_1\in\cM_\lam$, along with a partition of $[n]$ denoted as $(V^\ast,V_1,\dots,V_\ell)$ for $\ell\leq n-1$, where $P(V_i)=P(V^\ast)=0$ for $i\in[\ell]$. 
            
            \item The challenger samples $(\ket{s_1},\dots,\ket{s_n})\gets\share(1^\lam,1^n,m_b)$, and gives $\ket{s_{V_i}}$ to $\A_i$ for $i\in[\ell]$. 
            
            \item For $i\in[\ell]$, adversary $\A_i$ outputs a classical state $\st_i$, which is given to $\B$. Additionally, the challenger gives $\ket{s_{V^\ast}}$ to $\B$.

            \item Adversary $\B$ outputs a bit $b_\B$, which is the output of the experiment. 
        \end{enumerate}
        We say $\Pi_\UTSS$ satisfies UJLR security if for all QPT adversaries $(\A_0,\A_1,\dots,\A_n,\B)$ there exists a negligible function $\negl(\cdot)$ such that for all $\lam\in\N$, 
        \[\abs{\Pr[b_\B=1|b=0]-\Pr[b_\B=1|b=1]}= \negl(\lam)\] in the above security game.
    \end{itemize}
\end{definition}

\paragraph{Comparison with~\cite{CGLR24}.}
The definition of unbounded leakage-resilient secret sharing (ULRSS) in~\cite{CGLR24} has the same adversary structure as \cref{def:UTSS}, except each first-stage adversary gets a single share and tries to classically leak it.
Our notion of UJLRSS allows any set of unqualified shares to be an input to the leakage function. 
While ULRSS does guarantee individual shares cannot be classically leaked, a first-stage adversary can easily break the construction of~\cite{CGLR24} if it is given two shares. 

\paragraph{Comparison with~\cite{AGLL24}.} The definition of USS in~\cite{AGLL24} has the same adversary structure as \cref{def:UTSS}, except each first-stage adversary gets a single share and tries to clone it, and the second stage adversaries do not get honest shares unless the first-stage adversary sends them. 
Collusion in the USS definition is captured via an entanglement graph, where an edge corresponds to two first stage adversaries sharing entanglement. 
The USS security notions are then compared by the number of connected components in the entanglement graph, with 1 component being the strongest. 
However, this is counter-intuitive for defining collusion, since if all parties collude there should never be security. 

Furthermore, the USS definition does not guarantee that individual shares are unclonable. For instance, consider a 2-out-of-2 USS scheme where one share is a UE ciphertext encrypting the secret and the other is the (classical) UE secret key. The classical share can easily be cloned or classically leaked for a second stage adversary with the honest ciphertext share, which is also quite counter-intuitive.

However, some of this seems inherent, since having enough groups of parties colluding to cover the set $[n]$ would allow a cloning adversary to possibly output a different set of honest shares to each second-stage adversary, which would immediately break security.
Thus, unbounded classical leakage-resilience yields both a more natural and satisfiable definition for the task of secret sharing.

\subsection{Constructing UJLRSS for General Policies}\label{sec:utss}
In this section, we describe our construction of UJLRSS for all policies from collusion-resistant UTE and classical secret sharing for policies. The construction follows the structure of Yao's scheme (see~\cite{VNSRK03,C:JKKKPW17}). 

\begin{construction}[UJLRSS for Policies]\label{cons:utss}
    Let $\lam$ be a security parameter, $n$ be the number of parties, $P:2^{[n]}\to\zo{}$ be a policy with corresponding circuit $\phat$ of depth $d$ and size $t$, and $\cM=\setbk{\cM_\lam}_{\lam\in\N}$ be a message space. 
    Our construction relies on the following additional primitives:
    \begin{itemize}
        \item Let $\SSscheme=(\SSshare,\SSreconst)$ be a classical secret sharing scheme for $P$ with message space $\cM$ and share space $\cS=\setbk{\cS_\lam}_{\lam\in\N}$.
        \item Let $\Pi_\UTE=(\UTE.\Gen,\UTE.\Enc,\UTE.\Dec)$ be a collusion-resistant UTE scheme with message space $\cS\cup\zo{\ell(\lam)}$ and key space $\zo{\ell}$.
    \end{itemize}
    We construct our unbounded joint-leakage-resilient secret sharing scheme $\Pi_\UTSS=(\share,\reconst)$ as follows:
    \begin{itemize}
        \item $\share(1^\lam,1^n,m)$: On input the security parameter $\lam$, the number of parties $n$, and a message $m\in\cM_\lam$, the sharing algorithm does the following:
        \begin{enumerate}
            \item Sample $(s'_1,\dots,s'_n)\gets\SSshare(1^\lam,1^n,m)$.
            
            \item For the root (output) gate $\gstar\in\phat$, sample $(\ek_{\gstar},\dk_{\gstar})\gets\UTE.\Gen(1^\lam)$ and for $i\in[n]$ compute $\ket{\ct^{(\gstar)}_{i,1}}\gets\UTE.\Enc(\ek_{\gstar},s'_i)$. 
            
            \item For each gate $g$ in the policy circuit $\phat$ (starting from the root), let $w'_1,\dots,w'_{\kout(g)}$ represent the output wires with values $\val(w'_1),\dots,\val(w'_{\kout})$ (unless $g=\gstar$) and $w_1,\dots,w_{\kin}$ represent the input wires. 
            If $g\neq\gstar$, sample $(\ek_g,\dk_g)\gets\UTE.\Gen(1^\lam)$ and do the following:
            \begin{itemize}
                \item If $g$ is an AND gate, sample the values $\val(w_1),\dots,\val(w_{\kin})$ uniformly from $\zo{\ell}$ such that $\val(w_1)\oplus\cdots\oplus \val(w_{\kin})=\dk_g$. 
                \item If $g$ is an OR gate, set the values $\val(w_1),\dots,\val(w_{\kin})$ all to $\dk_g$.
            \end{itemize}
            For $g\neq \gstar$, $i\in[n]$, and $j\in[\kout(g)]$, compute $\ket{\ct^{(g)}_{i,j}}\gets\UTE.\Enc(\ek_g,\val(w'_j))$.
            
            \item Let $W_i$ be the set of input wires associated with the $\ord{i}$ input variable $x_i$. Then, for each $i\in[n]$ output the share
            \begin{equation}\label{eq:utss-share}
            \ket{s_i}=\parens{
            \setbk{w,\val(w)}_{w\in W_i}, 
            \setbig{g,\setbig{\ket{\ct^{(g)}_{i,j}}}_{j\in[\kout(g)]}}_{g\in\phat}
            }.
            \end{equation}
        \end{enumerate}
        
        \item $\reconst(X,\ket{s_X})$: On input the shares of a subset of parties $(X\subseteq[n],\ket{s_X}=\setbk{\ket{s_i}}_{i\in X})$, the reconstruction algorithm does the following:
        \begin{enumerate}
            \item If $P(X)=0$, output $\bot$. Otherwise, fix some $i^\ast\in X$ and for $i\in X$ parse $\ket{s_i}$ as in \cref{eq:utss-share}.
            For $i\in X$ and $w\in W_i$ add $(w,\val(w))$ to a table $T$ mapping wires to values.

            \item For $i=1$ to $d-1$, do the following:
            \begin{enumerate}
                \item For each gate $g$ of depth $i$ that outputs 1 on $X$ (denoted as $\phat_g(X)=1$ going forward), compute $\dk_g$ using $T$ and compute $\val(w'_j)\gets\UTE.\Dec\parens{\dk_g,\ket{\ct^{(g)}_{i^\ast,j}}}$ for $j\in[\kout]$.
                \item For $j\in[\kout]$, add $(w'_j,\val(w'_j))$ to $T$. 
            \end{enumerate}

            \item For the output gate $g^\ast$, use $T$ to compute $\dk_{g^\ast}$ and compute $s'_i\gets\UTE.\Dec\parens{\dk_{g^\ast},\ket{\ct^{(g^\ast)}_{i,1}}}$ for $i\in X$. Output $\SSreconst(X,s'_X)$.
        \end{enumerate}
    \end{itemize}
\end{construction}

\begin{theorem}[Correctness]
    Suppose $\Pi_\UTE$ and $\SSscheme$ are correct. 
    Then, \cref{cons:utss} is correct. 
\end{theorem}
\begin{proof}
    Take any $\lam,n\in\N$, $m\in\cM_\lam$, and any set $X\subseteq[n]$ such that $P(X)=1$.
    Let $(\ket{s_1},\dots,\ket{s_n})\gets\share(1^\lam,1^n,m)$ and consider the output of $\reconst(X,\ket{s_X})$.
    Since $P(X)=P_{\gstar}(X)=1$, $\reconst$ does not output $\bot$, and furthermore there exists a path of gates that output 1 from the leaves of $\phat$ to the root (output) gate $\gstar$.
    Thus, for all such gates $g$ (of increasing depth) and a fixed $\istar\in X$, we have 
    \[
    \val(w'_j)\gets\UTE.\Dec\parens{\dk_g,\ket{\ct^{(g)}_{i^\ast,j}}}
    \text{ for } j\in[\kout(g)]
    \] 
    with overwhelming probability by correctness of $\Pi_\UTE$. 
    We also have $s'_i\gets\UTE.\Dec\parens{\dk_{g^\ast},\ket{\ct^{(g^\ast)}_{i,1}}}$ for $i\in X$ again by $\Pi_\UTE$ correctness, where $s'_i$ are honest shares generated by $\SSshare$ by definition of $\share$. 
    Therefore, by $\SSscheme$ correctness, we have $m\gets\SSreconst(X,s'_X)$ with overwhelming probability, since $\SSscheme$ is a secret sharing scheme for $P$. 
    The theorem follows by a union bound. 
\end{proof}

\begin{theorem}[UJLR Security]
    Suppose $\SSscheme$ is secure for QPT adversaries and $\Pi_\UTE$ satisfies collusion-resistant security (\cref{def:cr-sec}) with second stage encryption queries (\cref{rem:second-stage-query}).
    Then, \cref{cons:utss} satisfies UJLR security. 
\end{theorem}
\begin{proof}
    Let $n=n(\lam),d=d(\lam),t=(\lam)$ be arbitrary polynomials. 
    Suppose there exists a QPT adversary $(\A=(\A_0,\A_1,\dots,\A_n),\B)$ that has non-negligible advantage $\delta$ in the UJLR security game (\cref{def:UTSS}). 
    For a set $X\subseteq[n]$, such that $P(X)=0$, let $L_X$ be the list of gates $g$ ordered by increasing depth such the $\phat_g(X)=0$. 
    Let $L_X[j]$ denote the $\ord{j}$ gate in $L_X$. 
    We now define a sequence of hybrid experiments using indices $i\in[\ell]$ and $j\in[t]$:
    \begin{itemize}
        \item $\hybb{1,0}$: This is the UJLR security game with challenge bit $b\in\zo{}$. In particular:
        \begin{itemize}
            \item At the beginning of the game, $\A_0$ outputs messages $m_0,m_1\in\cM_\lam$, along with a partition $(V^\ast,V_1,\dots,V_\ell)$ of $[n]$ for $\ell\leq n-1$, where $P(V_i)=P(V^\ast)=0$ for $i\in[\ell]$. 
            
            \item The challenger samples $(s'_1,\dots,s'_n)\gets\SSshare(1^\lam,1^n,m_b)$, $(\ek_{\gstar},\dk_{\gstar})\gets\UTE.\Gen(1^\lam)$, and for $i\in[n]$ computes $\ket{\ct^{(\gstar)}_{i,1}}\gets\UTE.\Enc(\ek_{\gstar},s'_i)$.
            For each $g\in\phat$ (from the root), the challenger does the following:
            \begin{itemize}
                \item If $g\neq\gstar$, sample $(\ek_g,\dk_g)\gets\UTE.\Gen(1^\lam)$.
                \item If $g$ is an AND gate, sample the values $\val(w_1),\dots,\val(w_{\kin})$ uniformly from $\zo{\ell}$ such that $\val(w_1)\oplus\cdots\oplus \val(w_{\kin})=\dk_g$. 
                \item If $g$ is an OR gate, set the values $\val(w_1),\dots,\val(w_{\kin})$ all to $\dk_g$.
                \item If $g\neq \gstar$, for $i\in[n]$, and $j\in[\kout(g)]$ compute $\ket{\ct^{(g)}_{i,j}}\gets\UTE.\Enc(\ek_g,\val(w'_j))$.
            \end{itemize}
            For each $i\in[n]$, the challenger sets 
            \[
            \ket{s_i}=\parens{
            \setbk{w,\val(w)}_{w\in W_i}, 
            \setbig{g,\setbig{\ket{\ct^{(g)}_{i,j}}}_{j\in[\kout(g)]}}_{g\in\phat}
            }.
            \]
            The challenger gives $\ket{s_{V_i}}$ to $\A_i$ for $i\in[\ell]$. 
            
            \item For $i\in[\ell]$, adversary $\A_i$ outputs a classical state $\st_i$, which is given to $\B$. Additionally, the challenger gives $\ket{s_{V^\ast}}$ to $\B$. Adversary $\B$ outputs a bit $b_\B$, which is the output of the experiment. 
        \end{itemize}
        
        \item $\hybb{i,j}$: Same as $\hybb{i,j-1}$ except for gate $g=L_X[j]$, the challenger samples 
        \[\diff{\ket{\ct^{(g)}_{i',j'}}\gets\UTE.\Enc(\ek_g,0)}\] for $i'\in V_i$ and $j'\in[\kout(g)]$. Here we overload 0 to mean the all 0s string or its equivalent in $\cS$.
    \end{itemize}
    Note that $\hybb{i,t}$ is the same as $\hybb{i+1,0}$ for $i\in[\ell-1]$. 
    We write $\hybb{i,j}(\A,\B)$ to denote the output distribution of
    an execution of $\hybb{i,j}$ with adversary $(\A,\B)$.
    We show that $\hybb{i,j-1}(\A,\B)$ and $\hybb{i,j}(\A,\B)$ are indistinguishable for any $i\in[\ell]$ and $j\in[t]$, and then show $\hyb^{(0)}_{\ell,t}(\A,\B)$ and $\hyb^{(1)}_{\ell,t}(\A,\B)$ are indistinguishable. 

    \begin{lemma}\label{lem:utss-cr}
        Suppose $\Pi_\UTE$ satisfies collusion-resistant security with second stage encryption queries.
        Then, there exists a negligible function $\negl(\cdot)$ such that for all $\lam\in\N$, we have
        \[
        \abs{\Pr[\hybb{i,j-1}(\A,\B)=1]-\Pr[\hybb{i,j}(\A,\B)=1]}=\negl(\lam).
        \]
    \end{lemma}
    \begin{proof}
        Suppose $(\A,\B)$ distinguish the experiments with non-negligible probability $\delta'$. 
        We use $(\A,\B)$ to construct an adversary $(\A',\B')$ that breaks collusion-resistant security as follows:
        \begin{enumerate}
            \item Adversary $\A'$ starts by running $\A_0$ to get $m_0,m_1\in\cM_\lam$ and $(V^\ast,V_1,\dots,V_\ell)$. 
            Adversary $\A'$ samples $(s'_1,\dots,s'_n)\gets\SSshare(1^\lam,1^n,m_b)$.

            \item For each $g\in\phat$ (from the root), adversary $\A'$ does the following:
            \begin{enumerate}
                \item If $g\neq L_{V_i}[j]$, sample $(\ek_g,\dk_g)\gets\UTE.\Gen(1^\lam)$.
                \item If $g\neq L_{V_i}[j]$ is an AND gate, sample values $\val(w_1),\dots,\val(w_{\kin})$ uniformly from $\zo{\ell}$ such that $\val(w_1)\oplus\cdots\oplus \val(w_{\kin})=\dk_g$. 
                If $g\neq L_{V_i}[j]$ is an OR gate, set all values $\val(w_1),\dots,\val(w_{\kin})$ to $\dk_g$. 
                If $g= L_{V_i}[j]$ is an AND gate, sample $\val(w)\getsr\zo{\ell}$ for all input wires such that $\phat_w(V_i)=1$ (the wire is 1 on input $V_i$).
                \item If 
                $g=L_{V_i}[j']$ such that $j'<j$, compute $\ket{\ct^{(g)}_{i',j''}}\gets\UTE.\Enc(\ek_g,0)$ for $j''\in[\kout(g)]$ and $i'\in V_{i}$.
                \item If $g=L_{V_i}[j]$, submit $\abs{V_i}$ copies of the messages $\setbk{0,\val(w'_{j'})}_{j'\in[\kout(g)]}$ to the UTE challenger to get $\ket{\ct^{(g)}_{i',j'}}$ for $i'\in V_i$ and $j'\in[\kout(g)]$.
                \item If $g=L_{V_i}[j']$ such that $j'>j$, compute $\ket{\ct^{(g)}_{i',j''}}\gets\UTE.\Enc(\ek_g,\val(w'_{j''}))$ for $j''\in[\kout(g)]$ and $i'\in V_{i}$.
            \end{enumerate}

            \item For $i'\in V_i$, adversary $\A'$ constructs $\ket{s_{i'}}$ as in \cref{eq:utss-share} and gives $\ket{s_{V_i}}$ to $\A_i$ to get $\st_i$.
            Adversary $\A'$ outputs $\st=(\st_i,\setbk{\ek_g,\dk_g}_{g\neq L_{V_i}[j]},s'_{[n]},T)$, where $T$ is the current wire to value mapping.

            \item On input $\dk_{\hat{g}}$ for $\hat{g}=L_{V_i}[j]$ and the state $\st=(\st_i,\setbk{\ek_g,\dk_g}_{g\neq \hat{g}},s'_{[n]},T)$, adversary $\B'$ first samples the missing values for the input wires of $\hat{g}$. 

            \item For index gate pairs $(i'',g)$ such that $g\in L_{V_{i'}}$ for some $i'\in[i-1]$ and additionally $i''\in V_{i'}$, adversary $\B'$ computes $\ket{\ct^{(g)}_{i'',j'}}\gets\UTE.\Enc(\ek_g,0)$ for $j'\in[\kout(g)]$.
            Otherwise, if $g\neq\gstar$, adversary $\B'$ computes $\ket{\ct^{(g)}_{i'',j}}\gets\UTE.\Enc(\ek_g,\val(w'_{j'}))$ for $j'\in[\kout(g)]$ and if $g=\gstar$, it computes $\ket{\ct^{(\gstar)}_{i'',1}}\gets\UTE.\Enc(\ek_g,s'_{i''})$. 
            For generating ciphertexts without $\ek_{\hat{g}}$, adversary $\B'$ makes an encryption query to its challenger. 

            \item For $i'\in [n]\setminus V_i$, adversary $\B'$ constructs $\ket{s_{i'}}$ as in \cref{eq:utss-share} and gives $\ket{s_{V_{i''}}}$ to $\A_{i''}$ to get $\st_{i''}$ for $i''\neq i$. It gives $(\st_1,\dots,\st_\ell,\ket{s_{V^\ast}})$ to $\B$ and outputs whatever $\B$ outputs.
            
        \end{enumerate}
        Clearly $(\A',\B')$ is QPT if $(\A,\B)$ is.
        If the UTE challenger encrypts honest values, $(\A',\B')$ simulates $\hybb{i,j-1}(\A,\B)$ perfectly.
        If the UTE challenger encrypts 0s, $(\A',\B')$ simulates $\hybb{i,j}(\A,\B)$ perfectly.
        Thus, $(\A',\B')$ has advantage $\delta'$ in the collusion-resistant security game, a contradiction. 
    \end{proof}

    \begin{lemma}\label{lem:utss-ss}
        Suppose $\SSscheme$ is secure for QPT adversaries. 
        Then, there exists a negligible function $\negl(\cdot)$ such that for all $\lam\in\N$, we have
        \[
        \abs{\Pr[\hyb^{(0)}_{\ell,t}(\A,\B)=1]-\Pr[\hyb^{(1)}_{\ell,t}(\A,\B)=1]}=\negl(\lam).
        \]
    \end{lemma}
    \begin{proof}
        Suppose $(\A,\B)$ distinguish the experiments with non-negligible probability $\delta'$. 
        We use $(\A,\B)$ to construct an adversary $\A'$ that breaks security of $\SSscheme$:
        \begin{enumerate}
            \item Adversary $\A'$ starts by running $\A_0$ to get $m_0,m_1\in\cM_\lam$ and $(V^\ast,V_1,\dots,V_\ell)$. Adversary $\A'$ sends $(m_0,m_1,V^\ast)$ to its challenger and gets back the shares $s'_{V^\ast}$.
            
            \item Adversary $\A'$ samples $(\ek_{\gstar},\dk_{\gstar})\gets\UTE.\Gen(1^\lam)$, computes 
            \[
            \ket{\ct^{(\gstar)}_{i,1}}\gets\UTE.\Enc(\ek_{\gstar},s'_i) \text{ for }  i\in V^\ast,
            \] 
            and computes 
            \[
            \ket{\ct^{(\gstar)}_{i,1}}\gets\UTE.\Enc(\ek_{\gstar},0) \text{ for }  i\in[n]\setminus V^\ast.
            \] 
            
            \item For each $g\in\phat$ (from the root), adversary $\A'$ does the following:
            \begin{enumerate}
                \item If $g\neq\gstar$, sample $(\ek_g,\dk_g)\gets\UTE.\Gen(1^\lam)$.
                \item If $g$ is an AND gate, sample values $\val(w_1),\dots,\val(w_{\kin})$ uniformly from $\zo{\ell}$ such that $\val(w_1)\oplus\cdots\oplus \val(w_{\kin})=\dk_g$. 
                If $g$ is an OR gate, set all values $\val(w_1),\dots,\val(w_{\kin})$ to $\dk_g$.
                \item For gates $g\neq \gstar$: if $(i',g)$ is such that $g\in L_{V_i}$ for some $i\in[\ell]$ and additionally $i'\in V_i$, compute $\ket{\ct^{(g)}_{i',j}}\gets\UTE.\Enc(\ek_g,0)$ for $j\in[\kout(g)]$. Otherwise, compute $\ket{\ct^{(g)}_{i',j}}\gets\UTE.\Enc(\ek_g,\val(w'_j))$ for $j\in[\kout(g)]$.
            \end{enumerate}

            \item Adversary $\A'$ constructs $\ket{s_i}$ as in \cref{eq:utss-share} and gives $\ket{s_{V_i}}$ to $\A_i$ for $i\in[\ell]$ to get $\st_i$.
            Adversary $\A'$ gives $(\st_1,\dots,\st_\ell,\ket{s_{V^\ast}})$ to $\B$ and outputs whatever $\B$ outputs. 
            
        \end{enumerate}
        Clearly $\A'$ is QPT if $(\A,\B)$ is. 
        If the secret sharing challenger shares $m_0$, $\A'$ perfectly simulates $\hyb^{(0)}_{\ell,t}(\A,\B)$.
        If the secret sharing challenger shares $m_1$, $\A'$ perfectly simulates $\hyb^{(1)}_{\ell,t}(\A,\B)$. 
        Thus, $\A'$ has advantage $\delta'$ in the secret sharing security game, a contradiction. 
    \end{proof}
    \noindent Combining \cref{lem:utss-cr,lem:utss-ss} proves the theorem by a hybrid argument. 
\end{proof}

\begin{corollary}[UJLR Secret Sharing]\label{cor:utss}
    Assuming the existence of post-quantum one-way functions, there exists a secure UJLR secret sharing scheme for all polynomial-size monotone Boolean circuits.
\end{corollary}

\begin{remark}[UJLRSS for Boolean Formulas]
    One can simplify \cref{cons:utss} to get UJLRSS for Boolean formulas from only query-bounded collusion-resistant UTE, which can be constructed from pseudorandom states. 
    As sketched in~\cref{sec:tech-overview}, each share is a classical share of $\dk$ for the policy $P$ along with a UTE ciphertext encrypting a classical share of $m$ for $P$. 
    Since secret sharing for Boolean formulas can be constructed with perfect information theoretic security, this allows any unqualified set of shares of $\dk$ to be perfectly simulated when appealing to UTE security for each group in the partition. 
\end{remark}

\section{Untelegraphable Functional Encryption}
In this section, we define and construct untelegraphable functional encryption (UTFE). 
Particularly, we construct both public-key and secret-key functional encryption with untelegraphable ciphertexts.
We then sketch a simple way to leverage UTFE to construct untelegraphable differing inputs obfuscation by additionally working in the classical oracle model. 

\subsection{Building Blocks and Definitions}\label{sec:utfe-defs}
In this section, we define the notions that will be useful for our constructions of public-key and secret-key UTFE. 
\paragraph{Pseudorandom function.} We start by recalling the notion of a pseudorandom function (PRF).

\begin{definition}[Pseudorandom Function]\label{def:PRF}
    Let $\cK=\{\cK_\lam\}_{\lam\in\N}$, $\cX=\{\cX_\lam\}_{\lam\in\N}$, and $\cY=\{\cY_\lam\}_{\lam\in\N}$ be ensembles of finite sets indexed by a security parameter $\lam$. 
    Let $\prf = \set{\prf_\lam}_{\lam \in \N}$ be an efficiently-computable collection of functions
    $\prf_\lam \colon \cK_\lam\times\cX_\lam\to\cY_\lam$.
    We say that $\prf$ is secure if for all efficient adversaries $\cA$ there exists a negligible function
    $\negl(\cdot)$ such that for all $\lam\in\N$:
    \[
    \abs{\Pr[\cA^{\prf_\lam(k,\cdot)}(1^\lam)=1 : k \getsr \cK_\lam] -
    \Pr[\cA^{f_\lam(\cdot)}(1^\lam)=1 : f_\lam\getsr\mathsf{Funs}[\cX_\lam,\cY_\lam]]}=\negl(\lam),
    \]
    where $\mathsf{Funs}[\cX_\lam,\cY_\lam]$ is the set of all functions from $\cX_\lam$ to $\cY_\lam$.
\end{definition}

\paragraph{Single-decryptor functional encryption.} We now recall the definition of single-decryptor functional encryption (SDFE)~\cite{AC:KitNis22,CG24}. 
This is the natural generalization of single-decryptor encryption to the functional encryption setting. 
We define the challenge-only definition from~\cite{AC:KitNis22}, since it is sufficient for the constructions in this work. 
We omit the definition of testing a $\gamma$-good FE distinguisher since we only use a basic property of it in~\cref{cor:ap-utk} and nowhere else. 
We refer to~\cite[Definition~7.29]{AC:KitNis22} for the full definition along with~\cite[Section~7.1]{AC:KitNis22} for the corresponding preliminaries.

\begin{definition}[Single-Decryptor Functional Encryption]\label{def:SDFE}
    A single-decryptor functional encryption scheme with message space $\cM=\setbk{\cM_\lam}_{\lam\in\N}$ and function space $\cF=\setbk{\cF_\lam}_{\lam\in\N}$ is a tuple of QPT algorithms $\Pi_\SDFE=(\Setup,\KeyGen,\Enc,\Dec)$ with the following syntax:
    \begin{itemize}
        \item $\Setup(1^\lam)\to(\pk,\msk)$: On input the security parameter $\lam$, the setup algorithm outputs a classical public key $\pk$ and master secret key $\msk$.
        \item $\KeyGen(\msk,f)\to\ket{\sk_f}$: On input the master key $\msk$ and a function $f\in\cF_\lam$, the key generation algorithm outputs a quantum function secret key $\ket{\sk_f}$.
        \item $\Enc(\pk,m)\to\ct$: On input the public key $\pk$ and a message $m\in\cM_\lam$, the encryption algorithm outputs a classical ciphertext $\ct$.
        \item $\Dec(\ket{\sk_f},\ct)\to m$: On input a quantum function secret key $\ket{\sk_f}$ and a ciphertext $\ct$, the decryption algorithm outputs a message $m\in\cM_\lam$. 
    \end{itemize}
    We require that $\Pi_\SDFE$ satisfy the following properties:
    \begin{itemize}
        \item \textbf{Correctness:} For all $\lam\in\N$, $m\in\cM_\lam$, and $f\in\cF_\lam$,
        \[
        \Pr\left[f(m)=y:
        \begin{array}{c}
        (\pk,\msk)\gets\Setup(1^\lam)\\
        \ct\gets\Enc(\pk,m)\\
        \ket{\sk_f}\gets\KeyGen(\msk,f)\\
        y\gets\Dec(\ket{\sk_f},\ct)
        \end{array}\right]\geq 1-\negl(\lam).
        \]
        \item \textbf{$\gamma$-Anti-Piracy Security:} For a security parameter $\lam$, an advantage parameter $\gamma\in[0,1]$, and an adversary $\A$, we define the anti-piracy security game as follows:
        \begin{enumerate}
            \item The challenger starts by sampling $(\pk,\msk)\gets\Setup(1^\lam)$ and giving $\pk$ to adversary $\A$.
            
            \item Adversary $\A$ outputs a function $f\in\cF_\lam$.
            The challenger samples $\ket{\sk_f}\gets\KeyGen(\msk,f)$ and gives $\ket{\sk_f}$ to $\A$.
    
            \item Adversary $\A$ outputs two messages $m_0,m_1\in\cM_\lam$ along with two second stage adversaries $\B=(\rho[\qreg{R}_\B],\mat{U}_\B)$ and $\C=(\rho[\qreg{R}_\C],\mat{U}_\C)$, where $\rho$ is a quantum state over registers $\qreg{R}_\B$ and $\qreg{R}_\C$, and $\mat{U}_\B$ and $\mat{U}_\C$ are general quantum circuits. 
            
            \item The challenger runs the test for a $\gamma$-good FE distinguisher with respect to $(\pk,m_0,m_1,f)$ on $\B$ and $\C$. The challenger outputs $b=1$ if both tests pass and $b=0$ otherwise.
        \end{enumerate}
        We say an SDFE scheme satisfies $\gamma$-anti-piracy security if for all QPT adversaries $\A$ there exists a negligible function $\negl(\cdot)$ such that for all $\lam\in\N$, $\Pr[b=1]\leq \negl(\lam)$ in the above security game. 
    \end{itemize}
\end{definition}

\begin{theorem}[{SDFE~\cite{AC:KitNis22}}]\label{thm:sdfe}
    Assuming sub-exponentially secure indistinguishability obfuscation and one-way functions, and the quantum hardness of learning with errors (LWE), there exists an SDFE scheme for $\Ppoly$ that satisfies $\gamma$-anti-piracy security for any inverse polynomial $\gamma=\gamma(\lam)$.
\end{theorem}

\paragraph{Functional encryption with untelegraphable keys.} We also define the untelegraphable analog of SDFE, which we call functional encryption with untelegraphable keys (FE-UTK) and show that SDFE implies FE-UTK.

\begin{definition}[FE with Untelegraphable Keys]\label{def:FE-UTK}
    An FE-UTK scheme with message space $\cM=\setbk{\cM_\lam}_{\lam\in\N}$ and function space $\cF=\setbk{\cF_\lam}_{\lam\in\N}$ is a tuple of QPT algorithms $\Pi_\UTK=(\Setup,\KeyGen,\Enc,\Dec)$ with the same syntax and correctness requirement as \cref{def:SDFE}. We require that $\Pi_\UTK$ satisfy the following security property:
    \begin{itemize}
        \item \textbf{UTK Security:} For a security parameter $\lam$, a bit $b\in\zo{}$, and a two-stage adversary $(\A,\B=(\B_0,\B_1))$, we define the UTK security game as follows:
        \begin{enumerate}
            \item The challenger starts by sampling $(\pk,\msk)\gets\Setup(1^\lam)$ and giving $\pk$ to adversary $\A$.
            
            \item Adversary $\A$ outputs a function $f\in\cF_\lam$. The challenger samples $\ket{\sk_f}\gets\KeyGen(\msk,f)$ and gives $\ket{\sk_f}$ to $\A$. 
            Adversary $\A$ outputs a classical string $\st$, which is then given to $\B$.
            
            \item On input $\st$, adversary $\B_0$ outputs two messages $m_0,m_1\in\cM_\lam$ along with a state $\st_\B$, which is given to $\B_1$.

            \item The challenger computes $\ct\gets\Enc(\pk,m_{b})$ and gives $\ct$ to $\B_1$. Adversary $\B_1$ outputs a bit $b'$, which is the output of the experiment.
        \end{enumerate}
        We say a FE-UTK scheme satisfies UTK security if for all QPT adversaries $(\A,\B)$ there exists a negligible function $\negl(\cdot)$ such that for all $\lam\in\N$, \[\abs{\Pr[b'=1|b=0]-\Pr[b'=1|b=1]}= \negl(\lam)\] in the above security game. 
    \end{itemize}
\end{definition}

\begin{corollary}[Anti-Piracy Security Implies UTK Security]
\label{cor:ap-utk}
    Suppose a scheme $\Pi_\SDFE$ satisfies anti-piracy security. Then, $\Pi_\SDFE$ satisfies UTK security.
\end{corollary}
\begin{proof}
    Given an adversary $(\A,\B=(\B_0,\B_1))$ that wins the UTK security game with non-negligible advantage $\delta>0$, we construct $\A'$ that wins the $\gamma$-anti-piracy game for inverse polynomial $\gamma(\lam)$:
    \begin{enumerate}
        \item Adversary $\A'$ runs $\A$ and gives $\st$ to $\B_0$ to get $\st_\B$ and messages $m_0,m_1$. 
        \item $\A'$ defines both $\B'=\C'=(\st_\B,\mat{U}_{\B_1})$ and gives them to the challenger. 
    \end{enumerate}
    By assumption, the probability of $(\A,\B)$ guessing a random challenge bit is $1/2+\delta$, where for infinitely many $\lam\in\N$, the advantage $\delta(\lam)\geq1/p(\lam)$ for some polynomial $p$.
    Thus, on such values of $\lam$, the challenger's distinguisher test will pass with inverse polynomial probability by setting $\gamma(\lam)=1/p'(\lam)$ for $p'(\lam)=100p(\lam)$, say. Thus, for infinitely many $\lam\in\N$, the $\gamma$-anti-piracy game will output 1 with some inverse polynomial probability, which is non-negligible.
\end{proof}

\paragraph{Untelegraphable functional encryption.} We now formally define the notions of untelegraphable secret-key and public-key functional encryption (abbreviated UTSKFE and UTPKFE, respectively). 
These definitions extend plain public-key and secret-key functional encryption to have quantum ciphertexts which satisfy untelegraphability. 
In this work, we consider a relaxation of the standard untelegraphable security definition that would give out the master key to the second-stage adversary.
Specifically, the alternate security game allows the second-stage adversary to make queries to \textit{any} function secret key, even one that distinguishes the challenge messages. 
This relaxation is similar to the notion of SDFE, which has no master key it can give out since the ciphertext is the additional information given to the second-stage adversaries.

\begin{definition}[Untelegraphable Secret-Key Functional Encryption]
\label{def:UTSKFE}
    An untelegraphable secret-key functional encryption scheme with message space $\cM=\setbk{\cM_\lam}_{\lam\in\N}$ and function space $\cF=\setbk{\cF_\lam}_{\lam\in\N}$ is a tuple of QPT algorithms $\Pi_\UTSKFE=(\Setup,\KeyGen,\Enc,\Dec)$ with the following syntax:
    \begin{itemize}
        \item $\Setup(1^\lam)\to\msk$: On input the security parameter $\lam$, the setup algorithm outputs a classical master secret key $\msk$.
        \item $\KeyGen(\msk,f)\to\sk_f$: On input the master key $\msk$ and a function $f\in\cF_\lam$, the key generation algorithm outputs a function secret key $\sk_f$.
        \item $\Enc(\msk,m)\to\ket{\ct}$: On input the master key $\msk$ and a message $m\in\cM_\lam$, the encryption algorithm outputs a quantum ciphertext $\ketct$.
        \item $\Dec(\sk_f,\ketct)\to m$: On input a function secret key $\sk_f$ and a quantum ciphertext $\ketct$, the decryption algorithm outputs a message $m\in\cM_\lam$. 
    \end{itemize}
    We require that $\Pi_\UTSKFE$ satisfy the following properties:
    \begin{itemize}
        \item \textbf{Correctness:} For all $\lam\in\N$, $m\in\cM_\lam$, and $f\in\cF_\lam$,
        \[
        \Pr\left[f(m)=y:
        \begin{array}{c}
        \msk\gets\Setup(1^\lam)\\
        \ketct\gets\Enc(\msk,m)\\
        \sk_f\gets\KeyGen(\msk,f)\\
        y\gets\Dec(\sk_f,\ketct)
        \end{array}\right]\geq 1-\negl(\lam).
        \]
        \item \textbf{Untelegraphable Security with Unrestricted Key Queries:} For a security parameter $\lam$, a bit $b\in\zo{}$, and a two-stage adversary $(\A,\B)$, we define the untelegraphable security game as follows:
        \begin{enumerate}
            \item At the beginning of the game, the challenger samples $\msk\gets\Setup(1^\lam)$.
            \item Adversary $\A$ can now make key-generation queries by specifying a function $f\in\cF_\lam$. The challenger replies with $\sk_f\gets\KeyGen(\msk,f)$.
            \item Adversary $\A$ can also make encryption queries by specifying a message $m\in\cM_\lam$. The challenger replies with $\ketct\gets\Enc(\msk,m)$.
            \item Adversary $\A$ outputs two messages $m_0,m_1\in\cM_\lam$. If $f(m_0)=f(m_1)$ holds for all functions $f$ queried by $\A$, the challenger generates a ciphertext $\ket{\ctstar}\gets\Enc(\msk,m_b)$ and sends $\ket{\ctstar}$ to $\A$.
            \item Adversary $\A$ is again allowed to make key-generation queries, so long as $f(m_0)=f(m_1)$ holds and can still make encryption queries as before. 
            \item Adversary $\A$ outputs a classical string $\st$ which is given to adversary $\B$. 
            \item Adversary $\B$ can now make encryption queries and key-generation queries where the function $f\in\cF_\lam$ is not restricted. Adversary $\B$ outputs a bit $b_\B$, which is the output of the experiment. 
        \end{enumerate}
        We say a UTSKFE scheme $\Pi_\UTSKFE$ satisfies untelegraphable security if for all QPT adversaries $(\A,\B)$ there exists a negligible function $\negl(\cdot)$ such that for all $\lam\in\N$, 
        \[\abs{\Pr[b_\B=1|b=0]-\Pr[b_\B=1|b=1]}= \negl(\lam)\] in the above security game. We say a scheme $\Pi_\UTSKFE$ satisfies untelegraphable security \textit{with master key revealing} if security still holds when $\B$ is given $\msk$ instead of key-generation access.  We say a scheme $\Pi_\UTSKFE$ is \textit{single-ciphertext secure} if the adversary only gets to make the challenge ciphertext query and does not otherwise have access to an encryption oracle. 
    \end{itemize}
\end{definition}

\begin{definition}[Untelegraphable Public-Key Functional Encryption]\label{def:UTPKFE}
    We define untelegraphable PKFE the same as untelegraphable SKFE (\cref{def:UTSKFE}), with the following basic modifications:
    \begin{itemize}
        \item The setup algorithm $\Setup$ now outputs a public key and master key pair $(\pk,\msk)$.
        \item The encryption algorithm $\Enc$ now takes as input the public key $\pk$ instead of the master key $\msk$.
        \item Instead of making encryption queries, the adversary gets the public key $\pk$ from the challenger at the start of the security game. 
    \end{itemize}
\end{definition}

\paragraph{Plain PKFE and SKFE.} We also use plain public-key and secret-key FE in our constructions. The adaptive security definitions are the same as \cref{def:UTSKFE} except the first stage adversary outputs a bit $b'$. The syntax is also the same except ciphertexts are classical. We omit the definitions due to redundancy. 

\subsection{Constructing Untelegraphable Secret-Key Functional Encryption}
\label{sec:utskfe}
In this section, we construct single-ciphertext secure untelegraphable SKFE from FE-UTK (\cref{def:FE-UTK}) and standard SKFE. 
\begin{construction}[Untelegraphable SKFE]\label{cons:utskfe}
    Let $\lam$ be a security parameter. 
    Our construction relies on the following additional primitives:
    \begin{itemize}
        \item Let $\Pi_\UTK=\UTK.(\Setup,\KeyGen,\Enc,\Dec)$ be an FE scheme with untelegraphable keys with function space $\Ppoly$. 
        \item Let $\Pi_\SKFE=\SKFE.(\Setup,\KeyGen,\Enc,\Dec)$ be a secret-key FE scheme  with message space $\cM=\setbk{\cM_\lam}_{\lam\in\N}$ and function space $\cF=\setbk{\cF_\lam}_{\lam\in\N}$.
    \end{itemize}
    We construct our UTSKFE scheme $\Pi_\UTSKFE=(\Setup,\KeyGen,\Enc,\Dec)$ as follows:
    \begin{itemize}
        \item $\Setup(1^\lam)$: On input the security parameter $\lam$, the setup algorithm samples \[(\pk_\UTK,\msk_\UTK)\gets\UTK.\Setup(1^\lam) \eqand \msk_\SKFE\gets\SKFE.\Setup(1^\lam),\] and outputs $\msk=(\pk_\UTK,\msk_\UTK,\msk_\SKFE)$. 

        \item $\KeyGen(\msk,f)$: On input the master key $\msk=(\pk_\UTK,\msk_\UTK,\msk_\SKFE)$ and a function $f\in\cF_\lam$, the key generation algorithm samples \ifelsesubmiteq{\sk'_f\gets\SKFE.\KeyGen(\msk_\SKFE,f)} and outputs $\sk_f\gets\UTK.\Enc(\pk_\UTK,\sk'_f)$.

        \item $\Enc(\msk, m)$: On input the master key $\msk=(\pk_\UTK,\msk_\UTK,\msk_\SKFE)$ and a message $m\in\cM_\lam$, the encryption algorithm computes $\ct_\SKFE\gets\SKFE.\Enc(\msk_\SKFE,m)$ and outputs \[\ketct\gets\UTK.\KeyGen(\msk_\UTK,G[\ct_\SKFE]),\] where $G[\ct_\SKFE]$ is the circuit which takes an SKFE function key $\sk'_f$ as input and outputs $\SKFE.\Dec(\ct_\SKFE,\sk'_f)$.

        \item $\Dec(\sk_f,\ketct)$: On input a function secret key $\sk_f$ and a quantum ciphertext $\ketct$, the decryption algorithm outputs $\UTK.\Dec(\ketct,\sk_f)$. 
    \end{itemize}
\end{construction}

\begin{theorem}[Correctness]\label{thm:utskfe-correct}
    If $\Pi_\UTK$ and $\Pi_\SKFE$ are correct, then \cref{cons:utskfe} is correct. 
\end{theorem}
\begin{proof}
    Take any $\lam\in\N$, $m\in\cM_\lam$, and $f\in\cF_\lam$. 
    Let $\msk=(\pk_\UTK,\msk_\UTK,\msk_\SKFE)\gets\Setup(1^\lam)$. 
    As in $\Enc$, compute $\ct_\SKFE\gets\SKFE.\Enc(\msk_\SKFE,m)$ and sample $\ketct\gets\UTK.\KeyGen(\msk_\UTK,G[\ct_\SKFE])$.
    As in $\KeyGen$, sample \ifelsesubmiteq{\sk'_f\gets\SKFE.\KeyGen(\msk_\SKFE,f)} and set $\sk_f\gets\UTK.\Enc(\pk_\UTK,\sk'_f)$. 
    We consider the output of $\Dec(\sk_f,\ketct)$.
    By correctness of $\Pi_\UTK$, we have $G[\ct_\SKFE](\sk'_f)\gets\UTK.\Dec(\ketct,\sk_f)$.
    By correctness of $\Pi_\SKFE$, we have $f(m)\gets G[\ct_\SKFE](\sk'_f)$, as desired. 
\end{proof}

\begin{theorem}[Untelegraphable Security]\label{thm:utskfe-secure}
    Suppose $\Pi_\SKFE$ satisfies adaptive security and $\Pi_\UTK$ satisfies UTK security (\cref{def:FE-UTK}).
    Then, \cref{cons:utskfe} satisfies single-ciphertext untelegraphable security with unrestricted key queries (\cref{def:UTSKFE}). 
\end{theorem}
\begin{proof}
    Suppose there exists a QPT two-stage adversary $\AB$ that has non-negligible advantage $\delta$ in the UTSKFE security game and adversary $\B$ makes at most $q$ queries. We define a sequence of hybrid experiments indexed by a bit $i\in[q]$:
    \begin{itemize}
        \item $\hybb{0}$: This is the UTSKFE security game with challenge bit $b\in\zo{}$. In particular:
        \begin{itemize}
            \item The challenger starts by sampling \[(\pk_\UTK,\msk_\UTK)\gets\UTK.\Setup(1^\lam) \eqand \msk_\SKFE\gets\SKFE.\Setup(1^\lam).\]
            It sets $\msk=(\pk_\UTK,\msk_\UTK,\msk_\SKFE)$.
            
            \item When adversary $\A$ makes a key query on function $f\in\cF_\lam$, the challenger samples \[\sk'_f\gets\SKFE.\KeyGen(\msk_\SKFE,f)\] and replies with $\sk_f\gets\UTK.\Enc(\pk_\UTK,\sk'_f)$. 
            
            \item When adversary $\A$ outputs two messages $m_0,m_1\in\cM_\lam$, if $f(m_0)=f(m_1)$ holds for all functions $f$ queried by $\A$, the challenger computes $\ct_\SKFE\gets\SKFE.\Enc(\msk_\SKFE,m_b)$ and replies with \[\ketct\gets\UTK.\KeyGen(\msk_\UTK,G[\ct_\SKFE]).\]

            \item Adversary $\A$ outputs a classical state $\st$ after it is finished making key queries such that $f(m_0)=f(m_1)$. The challenger answers the queries as before and gives $\st$ to $\B$.

            \item When adversary $\B$ makes a key query, the challenger answers as it did for $\A$, without checking $f(m_0)=f(m_1)$. Adversary $\B$ outputs a bit $b'$, which is the output of the experiment. 
        \end{itemize}

        \item $\hybb{i}$: Same as $\hybb{i-1}$ except when answering adversary $\B$'s $\ord{i}$ key query on function $f_i$, the adversary replies with $\diff{\sk_{f_i}\gets\UTK.\Enc(\pk_\UTK,0)}$, where we overload 0 to mean the all 0s string. 
    \end{itemize}
    We write $\hybb{i}\AB$ to denote the output distribution of
    an execution of $\hybb{i}$ with adversary $\AB$. 
    We now argue that each adjacent pair of distributions are indistinguishable.
    \begin{lemma}\label{lem:utskfe0-1}
        Suppose $\Pi_\UTK$ is UTK secure. 
        Then, there exists a negligible function $\negl(\cdot)$ such that for all $\lam\in\N$, we have
        \[
        \abs{\Pr[\hybb{i-1}(\A,\B)=1]-\Pr[\hybb{i}(\A,\B)=1]}=\negl(\lam).
        \]
    \end{lemma}
    \begin{proof}
        Suppose $\AB$ distinguishes the experiments with non-negligible probability $\delta'$. We use $\AB$ to construct an adversary $(\A',\B'=(\B'_0,\B'_1))$ that breaks UTK security as follows:
        \begin{enumerate}
            \item At the start of the game, adversary $\A'$ gets $\pk_\UTK$ from its challenger. Adversary $\A'$ samples $\msk_\SKFE\gets\SKFE.\Setup(1^\lam)$.

            \item When adversary $\A$ makes a key query on $f\in\cF_\lam$, adversary $\A'$ samples $\sk'_f\gets\SKFE.\KeyGen(\msk_\SKFE,f)$ and replies with $\sk_f\gets\UTK.\Enc(\pk_\UTK,\sk'_f)$. 

            \item When adversary $\A$ outputs two messages $m_0,m_1\in\cM_\lam$, if $f(m_0)=f(m_1)$ holds for all functions $f$ queried by $\A$, adversary $\A'$ computes $\ct_\SKFE\gets\SKFE.\Enc(\msk_\SKFE,m_b)$ and queries the function $G[\ct_\SKFE]$ to the UTK challenger to get back a quantum state that it labels $\ketct$ and gives to $\A$.

            \item When adversary $\A$ outputs a classical state $\st$, adversary $\A'$ gives $\st$ to adversary $\B$ and $(\st,\pk_\UTK)$ to $\B'_0$. 

            \item Adversary $\B'_0$ starts running adversary $\B$ with randomness $r$. For $\B$'s $\ord{j}$ key queries where $j<i$ on $f_j$, adversary $\B'_0$ replies with $\sk_{f_j}\gets\UTK.\Enc(\pk_\UTK,0)$.
            For the $\ord{i}$ key query $f_i$, adversary $\B'_0$ samples $\sk'_{f_i}\gets\SKFE.\KeyGen(\msk_\SKFE,f_i)$ and outputs the message pair $(0,\sk'_{f_i})$ along with $\st_\B=(\st,\pk_\UTK,r,\setbk{\sk_{f_j}}_{j\in[i-1]})$. 
            Adversary $\B'_1$ gets $\sk_{f_i}$ from the UTK challenger along with $\st_\B$, reruns $\B$ with randomness $r$, answers the first $i-1$ key queries as $\B'_0$ did, and answers the $\ord{i}$ query with $\sk_{f_i}$.
            For the remaining key queries on $f_j$, adversary $\B'_1$ samples $\sk'_{f_j}\gets\SKFE.\KeyGen(\msk_\SKFE,f_j)$ and replies with $\sk_{f_j}\gets\UTK.\Enc(\pk_\UTK,\sk'_{f_j})$.

            \item Adversary $\B'_1$ outputs whatever $\B$ outputs.
        \end{enumerate}
        Clearly $(\A',\B'=(\B'_0,\B'_1))$ is efficient if $\AB$ is. 
        If $\sk_{f_i}$ is an encryption of $\sk'_{f_i}$, adversary $(\A',\B')$ simulates $\hybb{i-1}(\A,\B)$.
        If $\sk_{f_i}$ is an encryption of 0, adversary $(\A',\B')$ simulates $\hybb{i}(\A,\B)$.
        Thus, adversary $(\A',\B')$ breaks UTK security with advantage $\delta'$, a contradiction. 
    \end{proof}

    \begin{lemma}\label{lem:utskfe-final}
        Suppose $\Pi_\SKFE$ is adaptively secure for QPT adversaries. 
        Then, there exists a negligible function $\negl(\cdot)$ such that for all $\lam\in\N$, we have
        \[
        \abs{\Pr[\hyb^{(0)}_{q}(\A,\B)=1]-\Pr[\hyb^{(1)}_{q}(\A,\B)=1]}=\negl(\lam).
        \]
    \end{lemma}
    \begin{proof}
        Suppose $\AB$ distinguishes the experiments with non-negligible probability $\delta'$. 
        We use $\AB$ to construct an adversary $\A'$ that breaks SKFE security as follows:
        \begin{enumerate}
            \item Adversary $\A'$ starts by sampling $(\pk_\UTK,\msk_\UTK)\gets\UTK.\Setup(1^\lam)$. 

            \item When adversary $\A$ makes a key query on $f$, adversary $\A'$ forwards it to the SKFE challenger to get $\sk'_f$ and replies with $\sk_f\gets\UTK.\Enc(\pk_\UTK,\sk'_f)$.

            \item When adversary $\A$ outputs two messages $m_0,m_1\in\cM_\lam$, if $f(m_0)=f(m_1)$ holds for all functions $f$ queried by $\A$, adversary $\A'$ sends $m_0,m_1$ to the SKFE challenger to get back $\ct_\SKFE$, and replies with \ifelsesubmiteq{\ketct\gets\UTK.\KeyGen(\msk_\UTK,G[\ct_\SKFE]).}

            \item When adversary $\A$ outputs a classical state $\st$, adversary $\A'$ gives $\st$ to adversary $\B$.

            \item Adversary $\A'$ replies to all key queries by $\B$ with $\sk_{f}\gets\UTK.\Enc(\pk_\UTK,0)$, and outputs whatever $\B$ outputs. 
        \end{enumerate}
        Clearly $\A'$ is efficient if $\AB$ is.
        If $\ct_\SKFE$ is an encryption of $m_0$, adversary $\A'$ simulates $\hyb^{(0)}_{q}(\A,\B)$.
        If $\ct_\SKFE$ is an encryption of $m_1$, adversary $\A'$ simulates $\hyb^{(1)}_{q}(\A,\B)$.
        Thus, adversary $\A'$ breaks SKFE security with advantage $\delta'$, a contradiction. 
    \end{proof}
    \noindent Combining \cref{lem:utskfe0-1,lem:utskfe-final}, the theorem follows by a standard hybrid argument. 
\end{proof}

\begin{corollary}[Untelegraphable SKFE]\label{cor:utskfe}
    Assuming sub-exponentially secure indistinguishability obfuscation and one-way functions, and the quantum hardness of learning with errors (LWE), there exists a PKFE scheme for $\Ppoly$ that satisfies untelegraphable security with unrestricted key queries.
\end{corollary}
\begin{proof}
    Follows immediately by combining \cref{thm:utskfe-correct,thm:utskfe-secure,cor:ap-utk,thm:sdfe}, and noting that indistinguishability obfuscation and one-way functions imply SKFE. 
\end{proof}

\subsection{Constructing Untelegraphable Public-Key Functional Encryption}
\label{sec:utpkfe}
In this section, we construct untelegraphable PKFE from untelegraphable SKFE, standard PKFE, and CPA-secure SKE. The construction and security proof follow the blueprint in~\cite{C:ABSV15}.

\begin{construction}[Untelegraphable PKFE]\label{cons:utpkfe}
    Let $\lam$ be a security parameter. 
    Our construction relies on the following additional primitives:
    \begin{itemize}
        \item Let $\Pi_\PKFE=\PKFE.(\Setup,\KeyGen,\Enc,\Dec)$ be a public-key FE scheme with function space $\Ppoly$. 
        
        \item Let $\Pi_\SKFE=\SKFE.(\Setup,\KeyGen,\Enc,\Dec)$ be an untelegraphable SKFE scheme with message space $\cM=\setbk{\cM_\lam}_{\lam\in\N}$ and function space $\cF=\setbk{\cF_\lam}_{\lam\in\N}$. Let $\SKFE.\KeyGen$ have randomness space $\zo{\rho(\lam)}$.
        
        \item Let $\Pi_\SKE=\SKE.(\Gen,\Enc,\Dec)$ be a secret-key encryption scheme.
        
        \item Let $\prf = \set{\prf_\lam}_{\lam \in \N}$, where $\prf_\lam \colon \cK_\lam\times\cF_\lam\to\zo{\rho(\lam)}$ be a pseudorandom function with key space $\cK=\{\cK_\lam\}_{\lam\in\N}$.
    \end{itemize}
    We construct our UTPKFE scheme $\Pi_\UTPKFE=(\Setup,\KeyGen,\Enc,\Dec)$ as follows:
    \begin{itemize}
        \item $\Setup(1^\lam)$: On input the security parameter $\lam$, the setup algorithm samples \[(\pk_\PKFE,\msk_\PKFE)\gets\PKFE.\Setup(1^\lam) \eqand \sk_\SKE\gets\SKE.\Gen(1^\lam),\] and outputs $(\pk=\pk_\PKFE,\msk=(\msk_\PKFE,\sk_\SKE))$.
        
        \item $\KeyGen(\msk,f)$: On input the master key $\msk=(\msk_\PKFE,\sk_\SKE)$ and a function $f\in\cF_\lam$, the key generation algorithm computes $\ct_\SKE\gets\SKE.\Enc(\sk_\SKE,0)$ and outputs \[\sk_f\gets\PKFE.\KeyGen(\msk_\PKFE, G[\ct_\SKE,f]),\] where $G[\ct_\SKE,f]$ is a circuit that takes as input a bit $b\in\zo{}$, three keys $\msk_\SKFE$, $k$, and $\sk_\SKE$, and outputs $\SKFE.\KeyGen(\msk_\SKFE,f;\prf_\lam(k,f))$ if $b=0$ or $\SKE.\Dec(\sk_\SKE,\ct_\SKE)$ if $b=1$.
        
        \item $\Enc(\pk,m)$: On input the public key $\pk=\pk_\PKFE$ and a message $m\in\cM_\lam$, the encryption algorithm samples $\msk_\SKFE\gets\SKFE.\Setup(1^\lam)$ and $k\getsr\cK_\lam$, computes $\ket{\ct_\SKFE}\gets\SKFE.\Enc(\msk_\SKFE,m)$ and 
        \[\ct_\PKFE\gets\PKFE.\Enc(\pk_\PKFE,(0,\msk_\SKFE,k,\bot)),\] 
        and outputs $\ketct=(\ket{\ct_\SKFE},\ct_\PKFE)$.
        
        \item $\Dec(\sk_f,\ketct)$: On input a function secret key $\sk_f$ and a quantum ciphertext $\ketct=(\ket{\ct_\SKFE},\ct_\PKFE)$, the decryption algorithm computes $\sk'_f\gets\PKFE.\Dec(\sk_f,\ct_\PKFE)$ and outputs $\SKFE.\Dec(\sk'_f,\ket{\ct_\SKFE})$.
    \end{itemize}
\end{construction}

\begin{theorem}[Correctness]\label{thm:utpkfe-correct}
    If $\Pi_\SKFE$ and $\Pi_\PKFE$ are correct and $\prf$ is secure, then \cref{cons:utpkfe} is correct. 
\end{theorem}
\begin{proof}
    Take any $\lam\in\N$, $m\in\cM_\lam$, and $f\in\cF_\lam$. 
    Let $(\pk=\pk_\PKFE,\msk=(\msk_\PKFE,\sk_\SKE))\gets\Setup(1^\lam)$. 
    As in $\Enc$, let $\msk_\SKFE\gets\SKFE.\Setup(1^\lam)$, $k\getsr\cK_\lam$, $\ket{\ct_\SKFE}\gets\SKFE.\Enc(\msk_\SKFE,m)$,
    \[\ct_\PKFE\gets\PKFE.\Enc(\pk_\PKFE,(0,\msk_\SKFE,k,\bot)),\] and set $\ketct=(\ket{\ct_\SKFE},\ct_\PKFE)$.
    As in $\KeyGen$, let $\ct_\SKE\gets\SKE.\Enc(\sk_\SKE,0)$ and \[\sk_f\gets\PKFE.\KeyGen(\msk_\PKFE, G[\ct_\SKE,f]).\]
    We consider the output of $\Dec(\sk_f,\ketct)$. 
    By definition of $G[\ct_\SKE,f]$ and $\Pi_\PKFE$ correctness, we have $\sk'_f\gets\PKFE.\Dec(\sk_f,\ct_\PKFE)$, where \ifelsesubmiteq{\sk'_f=\SKFE.\KeyGen(\msk_\SKFE,f;\prf_\lam(k,f)).}
    By security of $\prf$ and $\Pi_\SKFE$ correctness, we have 
    $m\gets\SKFE.\Dec(\sk'_f,\ket{\ct_\SKFE})$ with overwhelming probability, as desired. 
\end{proof}

\begin{theorem}[Untelegraphable Security]\label{thm:utpkfe-secure}
    Suppose $\Pi_\PKFE$ satisfies adaptive security for QPT adversaries, $\prf$ is a secure PRF for QPT adversaries, $\Pi_\SKE$ satisfies CPA-security for QPT adversaries and correctness, and $\Pi_\SKFE$ satisfies single-ciphertext untelegraphable security with unrestricted key queries (\cref{def:UTSKFE}).
    Then, \cref{cons:utpkfe} satisfies untelegraphable security with unrestricted key queries (\cref{def:UTPKFE}). 
\end{theorem}
\begin{proof}
    Suppose there exists a QPT two-stage adversary $\AB$ that has non-negligible advantage $\delta$ in the UTPKFE security game. 
    We define a sequence of hybrid experiments:
    \begin{itemize}
        \item $\hybb{0}$: This is the UTPKFE security game with challenge bit $b\in\zo{}$. In particular:
        \begin{itemize}
            \item The challenger starts by sampling \[(\pk_\PKFE,\msk_\PKFE)\gets\PKFE.\Setup(1^\lam) \eqand \sk_\SKE\gets\SKE.\Gen(1^\lam),\] sets $(\pk=\pk_\PKFE,\msk=(\msk_\PKFE,\sk_\SKE))$, and gives $\pk$ to adversary $\A$.
            
            \item When adversary $\A$ makes a key query on function $f\in\cF_\lam$, the challenger computes $\ct_\SKE\gets\SKE.\Enc(\sk_\SKE,0)$ and replies with
            \[\sk_f\gets\PKFE.\KeyGen(\msk_\PKFE, G[\ct_\SKE,f]).\]
            
            \item When adversary $\A$ outputs two messages $m_0,m_1\in\cM_\lam$, if $f(m_0)=f(m_1)$ holds for all functions $f$ queried by $\A$, the challenger samples $\msk_\SKFE\gets\SKFE.\Setup(1^\lam)$ and $k\getsr\cK_\lam$, computes $\ket{\ct_\SKFE}\gets\SKFE.\Enc(\msk_\SKFE,m_b)$ and
            \[\ct_\PKFE\gets\PKFE.\Enc(\pk_\PKFE,(0,\msk_\SKFE,k,\bot)),\] 
            and replies with $\ketct=(\ket{\ct_\SKFE},\ct_\PKFE)$.

            \item Adversary $\A$ outputs a classical state $\st$ after it is finished making key queries such that $f(m_0)=f(m_1)$. The challenger answers the queries as before and gives $\st$ to $\B$.

            \item When adversary $\B$ makes a key query, the challenger answers as it did for $\A$, without checking $f(m_0)=f(m_1)$. Adversary $\B$ outputs a bit $b'$, which is the output of the experiment. 
        \end{itemize}

        \item $\hybb{1}$: Same as $\hybb{0}$ except the challenger now samples \diff{$\msk_\SKFE\gets\SKFE.\Setup(1^\lam)$ and $k\getsr\cK_\lam$ at the beginning of the game} and for each key query on $f\in\cF_\lam$ it now computes $\diff{\sk'_f\gets\SKFE.\KeyGen(\msk_\SKFE,f;\prf_\lam(k,f))}$ and $\ct_\SKE\gets\SKE.\Enc(\sk_\SKE,\diff{\sk'_f})$. 

        \item $\hybb{2}$: Same as $\hybb{1}$ except the challenger now computes 
        \[\ct_\PKFE\gets\PKFE.\Enc(\pk_\PKFE,\diff{(1,\bot,\bot,\sk_\SKE)})\]
        when constructing the challenge ciphertext.

        \item $\hybb{3}$: Same as $\hybb{2}$ except for each key query on $f\in\cF_\lam$ the challenger now samples $\diff{r_f\getsr\zo{\rho(\lam)}}$ and computes $\sk'_f\gets\SKFE.\KeyGen(\msk_\SKFE,f;\diff{r_f})$. Note that $r_f$ is reused if $f$ is queried more than once.
    \end{itemize}
    We write $\hybb{i}\AB$ to denote the output distribution of
    an execution of $\hybb{i}$ with adversary $\AB$. 
    We now argue that each adjacent pair of distributions are indistinguishable.
    
    \begin{lemma}\label{lem:utpkfe0-1}
        Suppose $\Pi_\SKE$ is CPA-secure for QPT adversaries. 
        Then, there exists a negligible function $\negl(\cdot)$ such that for all $\lam\in\N$, we have
        \[
        \abs{\Pr[\hybb{0}(\A,\B)=1]-\Pr[\hybb{1}(\A,\B)=1]}=\negl(\lam).
        \]
    \end{lemma}
    \begin{proof}
        Suppose $\AB$ distinguishes the experiments with non-negligible probability $\delta'$. We use $\AB$ to construct an adversary $\A'$ that breaks CPA-security as follows: 
        \begin{enumerate}
            \item Adversary $\A'$ starts by sampling $(\pk_\PKFE,\msk_\PKFE)\gets\PKFE.\Setup(1^\lam)$, $\msk_\SKFE\gets\SKFE.\Setup(1^\lam)$, and $k\getsr\cK_\lam$. It gives $\pk=\pk_\PKFE$ to $\A$.

            \item When adversary $\A$ makes a key query on $f$, adversary $\A'$ computes \[\sk'_f\gets\SKFE.\KeyGen(\msk_\SKFE,f;\prf_\lam(k,f))\] and submits $(0,\sk'_f)$ to the CPA challenger to get back $\ct_\SKE$. Adversary $\A'$ replies to $\A$ with 
            \[\sk_f\gets\PKFE.\KeyGen(\msk_\PKFE, G[\ct_\SKE,f]).\]

            \item When adversary $\A$ outputs two messages $m_0,m_1\in\cM_\lam$, if $f(m_0)=f(m_1)$ holds for all functions $f$ queried by $\A$, adversary $\A'$ computes $\ket{\ct_\SKFE}\gets\SKFE.\Enc(\msk_\SKFE,m_b)$, computes
            \[\ct_\PKFE\gets\PKFE.\Enc(\pk_\PKFE,(0,\msk_\SKFE,k,\bot)),\] 
            and replies with $\ketct=(\ket{\ct_\SKFE},\ct_\PKFE)$.

            \item When adversary $\A$ outputs a classical state $\st$, adversary $\A'$ gives $\st$ to adversary $\B$. When adversary $\B$ makes a key query, adversary $\A'$ answers as it did for $\A$, without checking $f(m_0)=f(m_1)$. Adversary $\A'$ outputs whatever $\B$ outputs.
        \end{enumerate}
        Clearly $\A'$ is efficient if $\AB$ is. 
        If each $\ct_\SKE$ is an encryption of 0, adversary $\A'$ simulates $\hybb{0}\AB$.
        If each $\ct_\SKE$ is an encryption of $\sk'_f$, adversary $\A'$ simulates $\hybb{1}\AB$.
        Thus, adversary $\A'$ breaks CPA-security with advantage $\delta'$, as desired.
    \end{proof}

    \begin{lemma}\label{lem:utpkfe1-2}
        Suppose $\Pi_\PKFE$ is adaptively secure for QPT adversaries and $\Pi_\SKE$ is correct. 
        Then, there exists a negligible function $\negl(\cdot)$ such that for all $\lam\in\N$, we have
        \[
        \abs{\Pr[\hybb{1}(\A,\B)=1]-\Pr[\hybb{2}(\A,\B)=1]}=\negl(\lam).
        \]
    \end{lemma}
    \begin{proof}
        Suppose $\AB$ distinguishes the experiments with non-negligible probability $\delta'$. We use $\AB$ to construct an adversary $\A'$ that breaks PKFE security as follows:
        \begin{enumerate}
            \item At the start of the game, adversary $\A'$ gets $\pk_\PKFE$ from its challenger and gives it to $\A$. 
            Adversary $\A'$ also samples $\sk_\SKE\gets\SKE.\Gen(1^\lam)$, $\msk_\SKFE\gets\SKFE.\Setup(1^\lam)$, and $k\getsr\cK_\lam$.

            \item When adversary $\A$ makes a key query on $f$, adversary $\A'$ computes 
            \[\sk'_f\gets\SKFE.\KeyGen(\msk_\SKFE,f;\prf_\lam(k,f))\] and $\ct_\SKE\gets\SKE.\Enc(\sk_\SKE,\sk'_f)$. Adversary $\A'$ queries its challenger for $G[\ct_\SKE,f]$ to get back $\sk_f$, which it given to $\A$.

            \item When adversary $\A$ outputs two messages $m_0,m_1\in\cM_\lam$, if $f(m_0)=f(m_1)$ holds for all functions $f$ queried by $\A$, adversary $\A'$ computes $\ket{\ct_\SKFE}\gets\SKFE.\Enc(\msk_\SKFE,m_b)$, submits the pair of messages \[(0,\msk_\SKFE,k,\bot)\eqand(1,\bot,\bot,\sk_\SKE)\] to its challenger to get $\ct_\PKFE$, and replies with $\ketct=(\ket{\ct_\SKFE},\ct_\PKFE)$.

            \item When adversary $\A$ outputs a classical state $\st$, adversary $\A'$ gives $\st$ to adversary $\B$. When adversary $\B$ makes a key query, adversary $\A'$ answers as it did for $\A$, without checking $f(m_0)=f(m_1)$. Adversary $\A'$ outputs whatever $\B$ outputs.
        \end{enumerate}
        Clearly $\A'$ is efficient if $\AB$ is.
        If $\ct_\PKFE$ is an encryption of $(0,\msk_\SKFE,k,\bot)$, adversary $\A'$ simulates $\hybb{1}\AB$.
        If $\ct_\PKFE$ is an encryption of $(1,\bot,\bot,\sk_\SKE)$, adversary $\A'$ simulates $\hybb{2}\AB$.
        Adversary $\A'$ is admissible, since for all functions $f$, we have \[G[\ct_\SKE,f](0,\msk_\SKFE,k,\bot)=\sk'_f=G[\ct_\SKE,f](1,\bot,\bot,\sk_\SKE),\]
        by correctness of$\Pi_\SKE$. 
        Thus, adversary $\A'$ breaks PKFE security with advantage $\delta'$, as desired.
    \end{proof}

    \begin{lemma}\label{lem:utpkfe2-3}
        Suppose $\prf$ is a secure PRF for QPT adversaries. 
        Then, there exists a negligible function $\negl(\cdot)$ such that for all $\lam\in\N$, we have
        \[
        \abs{\Pr[\hybb{2}(\A,\B)=1]-\Pr[\hybb{3}(\A,\B)=1]}=\negl(\lam).
        \]
    \end{lemma}
    \begin{proof}
        Suppose $\AB$ distinguishes the experiments with non-negligible probability $\delta'$. We use $\AB$ to construct an adversary $\A'$ that breaks PRF security as follows:
        \begin{enumerate}
            \item Adversary $\A'$ starts by sampling $\msk_\SKFE\gets\SKFE.\Setup(1^\lam)$,
            \[(\pk_\PKFE,\msk_\PKFE)\gets\PKFE.\Setup(1^\lam) \eqand \sk_\SKE\gets\SKE.\Gen(1^\lam),\] 
            sets $(\pk=\pk_\PKFE,\msk=(\msk_\PKFE,\sk_\SKE))$, and gives $\pk$ to adversary $\A$.

            \item When adversary $\A$ makes a key query on function $f\in\cF_\lam$, adversary $\A'$ queries its challenger on $f$ to get $r_f$, computes 
            $\sk'_f\gets\SKFE.\KeyGen(\msk_\SKFE,f;r_f)$ and $\ct_\SKE\gets\SKE.\Enc(\sk_\SKE,\sk'_f)$, and replies with
            \[\sk_f\gets\PKFE.\KeyGen(\msk_\PKFE, G[\ct_\SKE,f]).\]

            \item When adversary $\A$ outputs two messages $m_0,m_1\in\cM_\lam$, if $f(m_0)=f(m_1)$ holds for all functions $f$ queried by $\A$, adversary $\A'$ computes $\ket{\ct_\SKFE}\gets\SKFE.\Enc(\msk_\SKFE,m_b)$, computes
            \[\ct_\PKFE\gets\PKFE.\Enc(\pk_\PKFE,(1,\bot,\bot,\sk_\SKE)),\] 
            and replies with $\ketct=(\ket{\ct_\SKFE},\ct_\PKFE)$.

            \item When adversary $\A$ outputs a classical state $\st$, adversary $\A'$ gives $\st$ to adversary $\B$. When adversary $\B$ makes a key query, adversary $\A'$ answers as it did for $\A$, without checking $f(m_0)=f(m_1)$. Adversary $\A'$ outputs whatever $\B$ outputs.
        \end{enumerate}
        Clearly $\A'$ is efficient if $\AB$ is.
        If each $r_f$ is from $\prf_\lam$, adversary $\A'$ simulates $\hybb{2}\AB$.
        If each $r_f$ is from a random function, adversary $\A'$ simulates $\hybb{3}\AB$.
        Thus, adversary $\A'$ breaks PRF security with advantage $\delta'$, as desired.
    \end{proof}

    \begin{lemma}\label{lem:utpkfe-final}
        Suppose $\Pi_\SKFE$ satisfies untelegraphable security. 
        Then, there exists a negligible function $\negl(\cdot)$ such that for all $\lam\in\N$, we have
        \[
        \abs{\Pr[\hyb^{(0)}_{3}(\A,\B)=1]-\Pr[\hyb^{(1)}_{3}(\A,\B)=1]}=\negl(\lam).
        \]
    \end{lemma}
    \begin{proof}
        Suppose $\AB$ distinguishes the experiments with non-negligible probability $\delta'$. We use $\AB$ to construct an adversary $(\A',\B')$ that breaks untelegraphable SKFE security as follows:
        \begin{enumerate}
            \item Adversary $\A'$ starts by sampling 
            \[(\pk_\PKFE,\msk_\PKFE)\gets\PKFE.\Setup(1^\lam) \eqand \sk_\SKE\gets\SKE.\Gen(1^\lam),\] 
            sets $(\pk=\pk_\PKFE,\msk=(\msk_\PKFE,\sk_\SKE))$, and gives $\pk$ to adversary $\A$.

            \item When adversary $\A$ makes a key query on function $f\in\cF_\lam$, adversary $\A'$ queries its challenger on $f$ to get $\sk'_f$, computes $\ct_\SKE\gets\SKE.\Enc(\sk_\SKE,\sk'_f)$ and replies with
            \[\sk_f\gets\PKFE.\KeyGen(\msk_\PKFE, G[\ct_\SKE,f]).\]
            If $f$ is queried multiple times, adversary $\A'$ reuses $\sk'_f$. 

            \item When adversary $\A$ outputs two messages $m_0,m_1\in\cM_\lam$, if $f(m_0)=f(m_1)$ holds for all functions $f$ queried by $\A$, adversary $\A'$ queries its challenger on $(m_0,m_1)$ to get $\ket{\ct_\SKFE}$, computes
            \[\ct_\PKFE\gets\PKFE.\Enc(\pk_\PKFE,(1,\bot,\bot,\sk_\SKE)),\] 
            and replies with $\ketct=(\ket{\ct_\SKFE},\ct_\PKFE)$.

            \item When adversary $\A$ outputs a classical state $\st$, adversary $\A'$ gives $\st$ to adversaries $\B,\B'$. When adversary $\B$ makes a key query on $f$, adversary $\B'$ queries its challenger on $f$ to get $\sk'_f$ and replies with $\sk_f$ as in step 2. Adversary $\B'$ outputs whatever $\B$ outputs.
        \end{enumerate}
        Clearly $(\A',\B')$ is efficient and admissible if $\AB$ is.
        If $\ket{\ct_\SKFE}$ is an encryption of $m_0$, adversary $(\A',\B')$ simulates $\hyb^{(0)}_{2}(\A,\B)$.
        If $\ket{\ct_\SKFE}$ is an encryption of $m_1$, adversary $(\A',\B')$ simulates $\hyb^{(1)}_{2}(\A,\B)$.
        Thus, adversary $(\A',\B')$ breaks untelegraphable SKFE security with advantage $\delta'$, a contradiction. 
    \end{proof}
    \noindent Combining \cref{lem:utpkfe0-1,lem:utpkfe1-2,lem:utpkfe2-3,lem:utpkfe-final}, the theorem follows by a standard hybrid argument. 
\end{proof}

\begin{corollary}[Untelegraphable PKFE]\label{cor:utpkfe}
    Assuming sub-exponentially secure indistinguishability obfuscation and one-way functions, and the quantum hardness of learning with errors (LWE), there exists a PKFE scheme for $\Ppoly$ that satisfies untelegraphable security with unrestricted key queries.
\end{corollary}
\begin{proof}
    Follows immediately by \cref{thm:utpkfe-correct,thm:utpkfe-secure,cor:utskfe}, and noting that indistinguishability obfuscation and one-way functions imply PKFE and CPA-secure SKE.
\end{proof}

\begin{remark}[Untelegraphable diO]
    As an application of UTPKFE, we can construct untelegraphable differing inputs obfuscation in the classical oracle model.
    Security says given many copies of an obfuscation of $C_0$ or $C_1$, where finding $x$ such that $C_0(x)\neq C_1(x)$ is computationally hard, it is hard to distinguish $C_0$ from $C_1$ even given a differing input $x$ after outputting a classical telegraph of the obfuscations.
    An obfuscation of a circuit $C$ is an untelegraphable encryption of message $C$ along with a classical oracle of the key generation circuit, which on input $x$ outputs $\sk_{f_x}$, where $f_x$ takes a circuit $C$ as input and outputs $C(x)$. 
    Security follows by switching the first stage classical oracle to one that outputs $\bot$ when queried on a differing input, and then appealing to UTPKFE security directly. 
\end{remark}

\section*{Acknowledgments}
We thank Alper \c{C}akan and Vipul Goyal for providing helpful comments on related work. Jeffrey Champion is supported by NSF CNS-2140975 and CNS-2318701.
}

\bibliographystyle{alpha}
\bibliography{reference,abbrev3,crypto}

\newcommand{\etalchar}[1]{$^{#1}$}
\begin{thebibliography}{HMNY22}

\bibitem[Aar19]{Shadow2}
Scott Aaronson.
\newblock Shadow tomography of quantum states.
\newblock {\em {SIAM} J. Comput.}, 49(5):STOC18--368, 2019.

\bibitem[AB24]{AB24}
Prabhanjan Ananth and Amit Behera.
\newblock A modular approach to unclonable cryptography.
\newblock In {\em CRYPTO}, 2024.

\bibitem[ABSV15]{C:ABSV15}
Prabhanjan Ananth, Zvika Brakerski, Gil Segev, and Vinod Vaikuntanathan.
\newblock From selective to adaptive security in functional encryption.
\newblock In Rosario Gennaro and Matthew J.~B. Robshaw, editors, {\em CRYPTO~2015, Part~II}, volume 9216 of {\em {LNCS}}, pages 657--677. Springer, Berlin, Heidelberg, August 2015.

\bibitem[AGKZ20]{AGKZ20}
Ryan Amos, Marios Georgiou, Aggelos Kiayias, and Mark Zhandry.
\newblock One-shot signatures and applications to hybrid quantum/classical authentication.
\newblock In {\em STOC}, page 255–268, 2020.

\bibitem[AGLL24]{AGLL24}
Prabhanjan Ananth, Vipul Goyal, Jiahui Liu, and Qipeng Liu.
\newblock Unclonable secret sharing.
\newblock In {\em ASIACRYPT}, 2024.

\bibitem[AHU19]{C:AmbHamUnr19}
Andris Ambainis, Mike Hamburg, and Dominique Unruh.
\newblock Quantum security proofs using semi-classical oracles.
\newblock In Alexandra Boldyreva and Daniele Micciancio, editors, {\em CRYPTO~2019, Part~II}, volume 11693 of {\em {LNCS}}, pages 269--295. Springer, Cham, August 2019.

\bibitem[AK07]{AK07}
Scott Aaronson and Greg Kuperberg.
\newblock Quantum versus classical proofs and advice.
\newblock In {\em CCC}, 2007.

\bibitem[AK21]{EPRINT:AnaKal21}
Prabhanjan Ananth and Fatih Kaleoglu.
\newblock Unclonable encryption, revisited.
\newblock Cryptology ePrint Archive, Report 2021/412, 2021.

\bibitem[AKL{\etalchar{+}}22]{C:AKLLZ22}
Prabhanjan Ananth, Fatih Kaleoglu, Xingjian Li, Qipeng Liu, and Mark Zhandry.
\newblock On the feasibility of unclonable encryption, and more.
\newblock In Yevgeniy Dodis and Thomas Shrimpton, editors, {\em CRYPTO~2022, Part~II}, volume 13508 of {\em {LNCS}}, pages 212--241. Springer, Cham, August 2022.

\bibitem[AKY24]{AKY24}
Prabhanjan Ananth, Fatih Kaleoglu, and Henry Yuen.
\newblock Simultaneous haar indistinguishability with applications to unclonable cryptography.
\newblock {\em arXiv:2405.10274}, 2024.

\bibitem[AN02]{AN02}
Dorit Aharonov and Tomer Naveh.
\newblock Quantum np - a survey, 2002.

\bibitem[AQY22]{C:AnaQiaYue22}
Prabhanjan Ananth, Luowen Qian, and Henry Yuen.
\newblock Cryptography from pseudorandom quantum states.
\newblock In Yevgeniy Dodis and Thomas Shrimpton, editors, {\em CRYPTO~2022, Part~I}, volume 13507 of {\em {LNCS}}, pages 208--236. Springer, Cham, August 2022.

\bibitem[BDK{\etalchar{+}}11]{BDKPPSY11}
Boaz Barak, Yevgeniy Dodis, Hugo Krawczyk, Olivier Pereira, Krzysztof Pietrzak, Fran{\c{c}}ois{-}Xavier Standaert, and Yu~Yu.
\newblock Leftover hash lemma, revisited.
\newblock In {\em CRYPTO}, 2011.

\bibitem[BL20]{TQC:BL20}
Anne Broadbent and S{\'{e}}bastien Lord.
\newblock Uncloneable quantum encryption via oracles.
\newblock In Steven~T. Flammia, editor, {\em 15th Conference on the Theory of Quantum Computation, Communication and Cryptography, {TQC} 2020, June 9-12, 2020, Riga, Latvia}, volume 158 of {\em LIPIcs}, pages 4:1--4:22. Schloss Dagstuhl - Leibniz-Zentrum f{\"{u}}r Informatik, 2020.

\bibitem[{\c{C}}G24a]{CG24pre}
Alper {\c{C}}akan and Vipul Goyal.
\newblock Unbounded leakage-resilient encryption and signatures.
\newblock {\em {IACR} Cryptol. ePrint Arch.}, 2024.

\bibitem[{\c{C}}G24b]{CG24}
Alper {\c{C}}akan and Vipul Goyal.
\newblock Unclonable cryptography with unbounded collusions and impossibility of hyperefficient shadow tomography.
\newblock In {\em TCC}, 2024.

\bibitem[CG24c]{STOC:ColGun24}
Andrea Coladangelo and Sam Gunn.
\newblock How to use quantum indistinguishability obfuscation.
\newblock In Bojan Mohar, Igor Shinkar, and Ryan {O'Donnell}, editors, {\em 56th ACM STOC}, pages 1003--1008. {ACM} Press, June 2024.

\bibitem[{\c{C}}GLR24]{CGLR24}
Alper {\c{C}}akan, Vipul Goyal, Chen{-}Da Liu{-}Zhang, and Jo{\~{a}}o Ribeiro.
\newblock Unbounded leakage-resilience and leakage-detection in a quantum world.
\newblock In {\em TCC}, 2024.

\bibitem[CHK05]{TCC:CanHalKat05}
Ran Canetti, Shai Halevi, and Jonathan Katz.
\newblock Adaptively-secure, non-interactive public-key encryption.
\newblock In Joe Kilian, editor, {\em TCC~2005}, volume 3378 of {\em {LNCS}}, pages 150--168. Springer, Berlin, Heidelberg, February 2005.

\bibitem[Die82]{Die82}
D.~Dieks.
\newblock Communication by epr devices.
\newblock {\em Phys. Lett. A}, 1982.

\bibitem[DRS04]{DRS04}
Yevgeniy Dodis, Leonid Reyzin, and Adam~D. Smith.
\newblock Fuzzy extractors: How to generate strong keys from biometrics and other noisy data.
\newblock In {\em EUROCRYPT}, 2004.

\bibitem[FK18]{FK18}
Bill Fefferman and Shelby Kimmel.
\newblock Quantum vs. classical proofs and subset verification.
\newblock In {\em Mathematical Foundations of Computer Science}, 2018.

\bibitem[GW11]{GW11}
Craig Gentry and Daniel Wichs.
\newblock Separating succinct non-interactive arguments from all falsifiable assumptions.
\newblock In {\em STOC}, 2011.

\bibitem[HKP20]{NatPhys:HKP20}
Hsin-Yuan Huang, Richard Kueng, and John Preskill.
\newblock Predicting many properties of a quantum system from very few measurements.
\newblock {\em Nature Physics}, 2020.

\bibitem[HMNY21]{AC:HMNY21}
Taiga Hiroka, Tomoyuki Morimae, Ryo Nishimaki, and Takashi Yamakawa.
\newblock Quantum encryption with certified deletion, revisited: Public key, attribute-based, and classical communication.
\newblock In Mehdi Tibouchi and Huaxiong Wang, editors, {\em ASIACRYPT~2021, Part~I}, volume 13090 of {\em {LNCS}}, pages 606--636. Springer, Cham, December 2021.

\bibitem[HMNY22]{C:HMNY22}
Taiga Hiroka, Tomoyuki Morimae, Ryo Nishimaki, and Takashi Yamakawa.
\newblock Certified everlasting zero-knowledge proof for {QMA}.
\newblock In Yevgeniy Dodis and Thomas Shrimpton, editors, {\em CRYPTO~2022, Part~I}, volume 13507 of {\em {LNCS}}, pages 239--268. Springer, Cham, August 2022.

\bibitem[ISV{\etalchar{+}}17]{ISVWY17}
Gene Itkis, Emily Shen, Mayank Varia, David Wilson, and Arkady Yerukhimovich.
\newblock Bounded-collusion attribute-based encryption from minimal assumptions.
\newblock In {\em PKC}, 2017.

\bibitem[JKK{\etalchar{+}}17]{C:JKKKPW17}
Zahra Jafargholi, Chethan Kamath, Karen Klein, Ilan Komargodski, Krzysztof Pietrzak, and Daniel Wichs.
\newblock Be adaptive, avoid overcommitting.
\newblock In Jonathan Katz and Hovav Shacham, editors, {\em CRYPTO~2017, Part~I}, volume 10401 of {\em {LNCS}}, pages 133--163. Springer, Cham, August 2017.

\bibitem[JL00]{EC:JarLys00}
Stanislaw Jarecki and Anna Lysyanskaya.
\newblock Adaptively secure threshold cryptography: Introducing concurrency, removing erasures.
\newblock In Bart Preneel, editor, {\em EUROCRYPT~2000}, volume 1807 of {\em {LNCS}}, pages 221--242. Springer, Berlin, Heidelberg, May 2000.

\bibitem[JLS18]{C:JiLiuSon18}
Zhengfeng Ji, Yi-Kai Liu, and Fang Song.
\newblock Pseudorandom quantum states.
\newblock In Hovav Shacham and Alexandra Boldyreva, editors, {\em CRYPTO~2018, Part~III}, volume 10993 of {\em {LNCS}}, pages 126--152. Springer, Cham, August 2018.

\bibitem[KN22]{AC:KitNis22}
Fuyuki Kitagawa and Ryo Nishimaki.
\newblock Functional encryption with secure key leasing.
\newblock In Shweta Agrawal and Dongdai Lin, editors, {\em ASIACRYPT~2022, Part~IV}, volume 13794 of {\em {LNCS}}, pages 569--598. Springer, Cham, December 2022.

\bibitem[KN23]{KN23}
Fuyuki Kitagawa and Ryo Nishimaki.
\newblock One-out-of-many unclonable cryptography: Definitions, constructions, and more.
\newblock In {\em TCC}, 2023.

\bibitem[KNTY19]{C:KNTY19}
Fuyuki Kitagawa, Ryo Nishimaki, Keisuke Tanaka, and Takashi Yamakawa.
\newblock Adaptively secure and succinct functional encryption: Improving security and efficiency, simultaneously.
\newblock In Alexandra Boldyreva and Daniele Micciancio, editors, {\em CRYPTO~2019, Part~III}, volume 11694 of {\em {LNCS}}, pages 521--551. Springer, Cham, August 2019.

\bibitem[Kre21]{Kre21}
W.~Kretschmer.
\newblock Quantum pseudorandomness and classical complexity.
\newblock {\em TQC 2021}, 2021.

\bibitem[LLPY24]{LLPY24}
Xingjian Li, Qipeng Liu, Angelos Pelecanos, and Takashi Yamakawa.
\newblock Classical vs quantum advice under classically-accessible oracle.
\newblock In {\em ITCS}, 2024.

\bibitem[MPY24]{MPY24}
Tomoyuki Morimae, Alexander Poremba, and Takashi Yamakawa.
\newblock Revocable quantum digital signatures.
\newblock In {\em TQC}, pages 5:1--5:24, 2024.

\bibitem[MQU07]{MU07}
Jörn Müller-Quade and Dominique Unruh.
\newblock Long-term security and universal composability.
\newblock In {\em TCC}, 2007.

\bibitem[MST21]{MST21}
Christian Majenz, Christian Schaffner, and Mehrdad Tahmasbi.
\newblock Limitations on uncloneable encryption and simultaneous one-way-to-hiding.
\newblock {\em arXiv:2103.14510}, 2021.

\bibitem[NN23]{NN23}
Anand Natarajan and Chinmay Nirkhe.
\newblock A distribution testing oracle separating qma and qcma.
\newblock In {\em CCC}, 2023.

\bibitem[NZ24]{NZ24}
Barak Nehoran and Mark Zhandry.
\newblock A computational separation between quantum no-cloning and no-telegraphing.
\newblock In {\em ITCS}, 2024.

\bibitem[Par70]{Par70}
James~L. Park.
\newblock The concept of transition in quantum mechanics.
\newblock {\em Foundations of Physics}, 1970.

\bibitem[Reg05]{Reg05}
Oded Regev.
\newblock On lattices, learning with errors, random linear codes, and cryptography.
\newblock In {\em STOC}, pages 84--93. ACM Press, 2005.

\bibitem[VNS{\etalchar{+}}03]{VNSRK03}
V.~Vinod, Arvind Narayanan, K.~Srinathan, C.~Pandu Rangan, and Kwangojo Kim.
\newblock On the power of computational secret sharing.
\newblock In {\em INDOCRYPT}, 2003.

\bibitem[VR89]{VR89}
K.~Vogel and H.~Risken.
\newblock Determination of quasiprobability distributions in terms of probability distributions for the rotated quadrature phase.
\newblock {\em Physical Review A}, 1989.

\bibitem[Wer98]{Wer98}
R.~F. Werner.
\newblock Optimal cloning of pure states.
\newblock {\em Physical Review A}, 1998.

\bibitem[Wie83]{Wiesner83}
Stephen Wiesner.
\newblock Conjugate coding.
\newblock {\em SIGACT News}, 15(1):78--88, 1983.

\bibitem[WZ82]{WZ82}
William~K. Wootters and Wojciech Zurek.
\newblock A single quantum cannot be cloned.
\newblock {\em Nature}, 1982.

\bibitem[Zha12]{C:Zhandry12}
Mark Zhandry.
\newblock Secure identity-based encryption in the quantum random oracle model.
\newblock In Reihaneh Safavi-Naini and Ran Canetti, editors, {\em CRYPTO~2012}, volume 7417 of {\em {LNCS}}, pages 758--775. Springer, Berlin, Heidelberg, August 2012.

\end{thebibliography}

\appendix
\ifsubmit{
  \newpage
  \begin{center}
    {\Large \underline{Supplementary Material}}
  \end{center}
\input{omitted}

}
\fi

\end{document}